%% file: main.tex
\documentclass[11pt]{article}
\usepackage[colorlinks=true, citecolor=blue]{hyperref}
\usepackage{kpfonts}
\usepackage[utf8]{inputenc}
\usepackage{fullpage}
\usepackage{amsmath, amsthm, amssymb, color}
\usepackage{physics, cleveref}
\usepackage{algorithm}
\usepackage[noend]{algpseudocode}
\usepackage{graphicx}
\usepackage{float}
\usepackage{cancel}
\usepackage{appendix}
\usepackage{import}
\usepackage{relsize}
\usepackage[all]{nowidow}
\usepackage[backend=biber,style=alphabetic,maxnames=10,sorting=ynt]{biblatex}
\usepackage{subfig}
\usepackage{thmtools, thm-restate}
\usepackage{tikz}
\usetikzlibrary{math, shapes.misc, decorations.pathmorphing}

\newtheorem{lemma}{Lemma}
\newtheorem{problem}{Problem}
\newtheorem{definition}{Definition}
\newtheorem{corollary}{Corollary}
\newtheorem{theorem}{Theorem}

\newtheorem{innercustomcon}{Conjecture}
\newenvironment{customcon}[1]
  {\renewcommand\theinnercustomcon{#1}\innercustomcon}
  {\endinnercustomcon}

\newcommand{\nc}{\newcommand}
\nc{\rnc}{\renewcommand}

\nc\benum{\begin{enumerate}}
\nc\eenum{\end{enumerate}}
\nc\bit{\begin{itemize}}
\nc\eit{\end{itemize}}
\nc\ot{\otimes}

\def\be#1\ee{\begin{equation} #1 \end{equation}}
\def\ba#1\ea{\begin{align} #1 \end{align}}
\def\bas#1\eas{\begin{align*}#1\end{align*}}

\def\bpm#1\epm{\begin{pmatrix}#1\end{pmatrix}}
\nc{\non}{\nonumber}

\DeclareMathOperator{\poly}{poly}
\DeclareMathOperator*{\E}{\mathbb{E}}
\DeclareMathOperator*{\Wg}{wg}

\DeclareMathOperator*{\dist}{dist}

\newcommand{\SharpP}{\textsf{\#P}}

\definecolor{orange}{rgb}{1,0.5,0}
\definecolor{purple}{rgb}{0.5,0,0.5}
\definecolor{dark green}{rgb}{0,0.4,0}

\graphicspath{{./Figures/}}

\Crefformat{equation}{Eq.~(#2#1#3)}

\title{Efficient classical simulation of random shallow 2D quantum circuits}

\author{John Napp\thanks{Center for Theoretical Physics, MIT, Cambridge, USA. email:\href{mailto:napp@mit.edu}{napp@mit.edu}} \and Rolando L. La Placa\thanks{Center for Theoretical Physics, MIT, Cambridge, USA. email:\href{mailto:rlaplaca@mit.edu}{rlaplaca@mit.edu}} \and Alexander M. Dalzell\thanks{Institute for Quantum Information and Matter, Caltech, Pasadena, USA. email:\href{mailto:adalzell@caltech.edu}{adalzell@caltech.edu}} \and Fernando G. S. L. Brand\~ao\thanks{Institute for Quantum Information and Matter, Caltech, Pasadena, USA; \newline \indent Amazon Web Services, Pasadena, USA. email:\href{mailto:fbrandao@caltech.edu}{fbrandao@caltech.edu}} \and Aram W. Harrow\thanks{Center for Theoretical Physics, MIT, Cambridge, USA. email:\href{mailto:aram@mit.edu}{aram@mit.edu}}}

\begin{document}
\maketitle
\vspace{-1em}
\begin{abstract}
Random quantum circuits are commonly viewed as hard to simulate classically. In some regimes this has been formally conjectured --- in the context of \textit{deep} 2D circuits, this is the basis for Google's recent announcement of ``quantum computational supremacy'' --- and there had been no evidence against the more general possibility that for circuits with uniformly random gates, approximate simulation of typical instances is almost as hard as exact simulation. We prove that this is not the case  by exhibiting a \emph{shallow} random circuit family that cannot be efficiently classically simulated exactly under standard hardness assumptions, but can be simulated approximately for all but a superpolynomially small fraction of circuit instances in time linear in the number of qubits and gates; this example limits the robustness of recent worst-case-to-average-case reductions for random circuit simulation. While our proof is based on a contrived random circuit family, we  furthermore conjecture that sufficiently shallow constant-depth random circuits are  efficiently simulable more generally. To this end, we propose and analyze two simulation algorithms. Implementing one of our algorithms for the depth-3 ``brickwork'' architecture, for which exact simulation is hard, we found that a laptop could simulate typical instances on a $409\times 409$ grid with variational distance error less than $0.01$ in approximately one minute per sample, a task intractable for previously known circuit simulation algorithms. Numerical evidence indicates that the algorithm remains efficient asymptotically.

Key to both our rigorous complexity separation and our conjecture is an observation that 2D shallow random circuit simulation can be reduced to a simulation of a form of 1D dynamics consisting of alternating rounds of random local unitaries and weak measurements. Similar processes have recently been the subject of an intensive research focus, which has found numerically that the dynamics generally undergo a phase transition from an efficient-to-simulate regime to an inefficient-to-simulate regime as measurement strength is varied. Via a mapping from random quantum circuits to classical statistical mechanical models, we give analytical evidence that a similar computational phase transition occurs for our algorithms as parameters of the circuit architecture like the local Hilbert space dimension and circuit depth are varied, and additionally that the 1D dynamics corresponding to sufficiently shallow random quantum circuits falls within the efficient-to-simulate regime.

\end{abstract}

\thispagestyle{empty}

\newpage
{
\hypersetup{linkcolor=black}
\tableofcontents
}
\thispagestyle{empty}

\newpage

 \setcounter{page}{1}
\section{Introduction}
\subsection{How hard is it for classical computers to simulate quantum circuits?}
As quantum computers add more qubits and gates, where is the line between classically
simulable and classically hard to simulate? And once the size and runtime of
the quantum computer are chosen, which gate sequence is hardest to simulate?

So far, our answers to these questions have been informal or incomplete.  On the simulation side,  \cite{markov2008simulating} showed that a quantum circuit could be classically simulated by contracting a tensor network with cost exponential in the treewidth of the graph induced by the circuit.  When applied to $n$ qubits in a line running a circuit with depth $d$, the simulation cost of this algorithm is $\exp(\widetilde{\Theta}(\min(n,d)))$.  More generally we could consider $n=L_1L_2$ qubits arranged in an $L_1\times L_2$ grid running for depth $d$, in which case the simulation cost would be
\be \exp(\widetilde{\Theta}(\min(L_1L_2, L_1d, L_2d))) \label{eq:2D-contract}.\ee
In other words, we can think of the computation as taking up a space-time volume of $L_1\times L_2\times d$ and the simulation cost is dominated by the size of the smallest cut bisecting this volume.
An exception is for $d=1$ or $d=2$, which have simple exact simulations~\cite{terhal2002adaptive}.  Some restricted classes such as stabilizer circuits \cite{gottesman1998heisenberg} or one dimensional systems that are sufficiently unentangled~\cite{vidal2003efficient, vidal2004efficient, osborne2006efficient} may also be simulated efficiently. However, the conventional wisdom has been that in general, for 2D circuits with $d\geq 3$, the simulation cost scales as \Cref{eq:2D-contract}.

These considerations led IBM to propose the benchmark of ``quantum volume''~\cite{qvolume} which in our setting is $\exp(\sqrt{d\min(L_1,L_2)})$; this does not exactly coincide with \Cref{eq:2D-contract} but qualitatively captures a similar phenomenon. The idea of quantum volume is to compare quantum computers with possibly different architectures by evaluating their performance on a simple benchmark.  This benchmark task is to perform $n$ layers of random two-qubit gates on $n$ qubits, and being able to perform this with $\lesssim 1$ expected gate errors corresponds to a quantum volume of $\exp(n)$\footnote{Our calculation of quantum volume for 2D circuits above uses the additional fact that, assuming for simplicity that $L_1\leq L_2$, we can simulate a fully connected layer of gates on $L_2x$ qubits (for $ x \leq L_1$) with $O(x L_2 / L_1)$ locally connected 2D  layers using the methods of \cite{Rosenbaum13}. Then $x$ is chosen to maximize $\min(L_2x, d/(xL_2/L_1))$.}. Google's quantum computing group has also proposed random unitary circuits as a benchmark task for quantum computers~\cite{google1}.  While their main goal has been quantum computational supremacy~\cite{google2,googleSupremacy}, random circuits could also be used to diagnose errors including those that go beyond single-qubit error models by more fully exploring the configuration space of the system~\cite{qvolume}.

These proposals from industry reflect a rough consensus that simulating a 2D random quantum circuit should be nearly as hard as exactly simulating an arbitrary circuit with the same architecture, or in other words that random circuit simulation is nearly as hard as the \emph{worst case}, given our current state of knowledge. To the contrary, we prove (assuming standard complexity-theoretic conjectures) that for a certain family of constant-depth architectures, classical simulation of typical instances with small allowed error is easy, despite worst-case simulation being hard (by which we mean, it is classically intractable to simulate an arbitrary random circuit realization with arbitrarily-small error). For these architectures, we show that a certain algorithm exploiting the randomness of the gates and the allowed small simulation error can run much more quickly than the scaling in \Cref{eq:2D-contract}, running in time $O(L_1 L_2)$. While our proof is architecture specific, we give numerical and analytical evidence that for sufficiently low constant values of $d$, the algorithm remains efficient more generally. The intuitive reason for this is that the simulation of 2D shallow random circuits can be reduced to the simulation of a form of effective 1D dynamics which includes random local unitaries and weak measurements.  The measurements then cause the 1D process to generate much less entanglement than it could in the worst case, making efficient simulation possible. Before discussing this in greater detail, we  review the main arguments for the prevailing belief that random circuit simulation should be nearly as hard as the worst case.

\benum
\item{\em Evidence from complexity theory.}
\benum
\item {\em Hardness of sampling from post-selected universality.}
  A long line of work has shown that it is worst-case hard to either sample from the output distributions of quantum circuits or compute their output probabilities~\cite{terhal2002adaptive,aaronson2005quantum,bremner2010classical,AA13,BMS15,BMS17, harrow2017quantum}.  While the requirement of worst-case simulation is rather strong, these results do apply to any quantum circuit family that becomes universal once post-selection is allowed, thereby including noninteracting bosons and depth-3 circuits.  The hardness results are also based on the widely believed conjecture that the polynomial hierarchy is infinite, or more precisely that approximate counting is weaker than exact counting. Since these results naturally yield  worst-case hardness, they do not obviously imply that random circuits should be hard.  In some cases, additional conjectures can be made to extend the hardness results to some form of average-case hardness (as well as ruling out approximate simulations)~\cite{AA13,BMS15,AaronsonC16}, but these conjectures have not received widespread scrutiny.  Besides stronger conjectures, these hardness results usually require that the quantum circuits have an ``anti-concentration''  property, meaning roughly that their outputs are not too far from the uniform distribution~\cite{HM18}.
While random circuits are certainly not the only route to anti-concentration (simply performing Hadamard gates on the all $\ket 0$ state also works) they are a natural way to combine anti-concentration with an absence of any obvious structure (e.g.~Clifford gates) that might admit a simple simulation.

  \item{\em Average-case hardness of computing output probabilities} \cite{bouland2019complexity, movassagh2018efficient, movassagh2019}.  It is known that random circuit simulation admits worst-to-average case reductions for the computation of output probabilities. In particular, the ability to near-exactly compute the probability of some output string for a $1-1/\poly(n)$ fraction of Haar-random circuit instances on some architecture is essentially as hard as computing output probabilities for an arbitrary circuit instance with this architecture, which is known to be \SharpP-hard even for certain depth-3 architectures. The existence of such a  worst-to-average-case reduction could be taken as evidence for the hardness of random circuits. Our algorithms circumvent these hardness results by computing output probabilities with small error, rather than near-exactly.

 \eenum

\item {\em Near-maximal entanglement in random circuits.} Haar-random states on $n$ qudits are nearly
  maximally entangled across all cuts simultaneously~\cite{page1993average, hayden2006aspects}.
  Random quantum circuits on $L\times L\times \cdots$ arrays of qudits achieve
  similar near-maximal entanglement across all possible cuts once the depth is $\Omega(L)$~\cite{dahlsten2007emergence, HM18} and before this time, the entanglement often spreads ``ballistically''~\cite{tsunami,bertini2019entanglement}. Random tensor networks with large bond dimension nearly obey a min-flow/max-cut-type theorem \cite{hayden2016holographic, hastings2017asymptotics}, again meaning that they achieve nearly maximal values of an entanglement-like quantity. These results suggest that when running algorithms based on tensor contraction, random gates should be nearly the hardest possible gates to simulate.

\item {\em Absence of algorithms taking advantage of random inputs.} There are not many algorithmic techniques known that simulate random circuits more easily than worst-case circuits.  There are a handful of exceptions.  In the presence of any constant rate of noise, random circuits~\cite{YungG17,GaoD18}, IQP circuits~\cite{BMS17} and (for photon loss) boson sampling~\cite{KalaiK14,OszmaniecB18} can be efficiently simulated.  These results can also be viewed as due to the fact that fault-tolerant quantum computing is not a generic phenomenon and requires structured circuits to achieve (see \cite{BMS17} for discussion in the context of IQP).
Permanents of random matrices whose entries have small nonzero mean can be approximated efficiently~\cite{EldarM18}, while the case of boson sampling corresponds to entries with zero mean and the approach of \cite{EldarM18} is known to fail there. A heuristic approximate simulation algorithm based on tensor network contraction  \cite{pan2019contracting} was recently proposed  and applied to random circuits, although for this algorithm it is unclear how the approximations made are related to the overall simulation error incurred  (in contrast, our algorithm based on matrix product states can bound the overall simulation error it is making, even when comparison with exact simulation is not feasible). In practice, evidence for a hardness conjecture often is no more than the absence of algorithms.
Indeed, while some approximation algorithms are known for estimating output probabilities of constant-depth circuits~\cite{BGM19}, IQP circuits~\cite{SB09} and boson sampling~\cite{AA13} up to additive error $\delta$ in time $\poly(n,1/\delta)$, these are not very helpful for random circuits where typical output probabilities are $\sim 2^{-n}$.

\eenum

For the case of constant depth, there have been some quantum computational supremacy proposals that do not use uniformly random circuits, mostly based on the MBQC (measurement-based quantum computing) model \cite{raussendorf2001one}.  This means first preparing a cluster state and then measuring it in mostly equatorial bases, or equivalently performing $e^{i\theta Z}$ for various angles $\theta$ and then measuring in the $X$ basis.  This is far from performing uniformly random nearest-neighbor gates up to the same depth and then measuring in a fixed basis.  In many cases, the angles $\theta$ are also chosen to implement a specific family of circuits as well \cite{gao2017quantum,miller2017quantum, bermejo2018architectures}.  Previously it had not been clear whether this difference is important for the classical complexity or not.

Despite the above intuitive arguments for why the simulation of uniformly random circuits should be nearly as hard as the worst case, we (1) prove that there exist  architectures for which this is not the case, and (2) give evidence that this result is not architecture specific, but is rather a general property of sufficiently shallow  random circuits. As described in more detail below, we propose two simulation algorithms. One (which we also implement numerically) is based on a 2D-to-1D mapping in conjunction with tensor network methods, and the other is based on exactly simulating small subregions which are then ``stitched'' together. The performance of both algorithms is related to certain entropic quantities.

We also give evidence of computational phase transitions even for noiseless simulation of quantum circuits.  Previously it was known that phase transitions between classical and quantum computation existed as a function of the noise parameter in conventional quantum computation \cite{shor1996fault, aharonov1996polynomial, harrow2003robustness, razborov2004upper, virmani2005classical, buhrman2006new, kempe2008upper} as well as in MBQC \cite{raussendorf2005long, barrett2009transitions}.  We do not know of any previous work showing phase transitions even for noiseless computation, except for the gap between depth-2 and depth-3 circuits given by Terhal-DiVincenzo \cite{terhal2002adaptive} and the  phase transition as a function of rate of qubit loss during the preparation of a 2D cluster state for MBQC \cite{browne2008phase}.

We believe this latter transition, arising from the percolation phase transition on a square lattice, is in fact closely related to our results. Intuitively, a measurement on the system may be viewed as ``good'' (entanglement destroying) or ``bad'' (entanglement preserving). If all measurements are random, some fraction of these measurements will be ``good''. If this fraction becomes sufficiently large, the entanglement-destroying effects win out, and the system can be efficiently simulated classically. In the special case of a $Z$-basis measurement being performed on each qubit of a 2D cluster state with some probability $p$, the result of \cite{browne2008phase} shows that this may be made rigorous via results from percolation theory; the resulting quantum state may be efficiently simulated classically if $1-p$ falls below the percolation threshold $p_c$ for a square lattice, and the resulting state supports universal MBQC if $1-p$ falls above $p_c$. In contrast, in our setting a measurement is performed on each site in a Haar-random basis with probability one. While the percolation argument is no longer directly applicable in this setting, we nonetheless find evidence in favor of a computational phase transition driven by local qudit dimension $q$ and circuit depth $d$.

\subsection{Our results}
We give two classes of results, which we summarize in more detail below. The first consists of rigorous separations in complexity between worst-case simulation\footnote{Unless specified otherwise, we use \emph{worst-case simulation} to refer to the problem of exactly simulating an arbitrary circuit instance.} and approximate average-case simulation (for sampling) and between near-exact average-case simulation and approximate average-case simulation (for computing output probabilities) for random circuit families defined with respect to certain circuit architectures. While these results are rigorous, they are proved with respect to a contrived architecture and therefore do not address the question of whether random shallow circuits are classically simulable more generally. To address this issue, we also give conjectures on the performance of our algorithms for more general and more natural architectures. Our second class of results consists of analytical and numerical evidence in favor of these conjectures.

\subsubsection{Provable complexity separations}
We now summarize our provable results for particular circuit architectures. We first define more precisely what we mean by an ``architecture''.
\begin{definition}[Architecture]
An \emph{architecture} $A$ is an efficiently computable mapping from  positive integers $L$ to circuit layouts $A(L)$ defined on rectangular grids with sidelengths $L \times f(L)$ for some function $f(L) \leq \poly(L)$. A ``circuit layout'' is a specification of locations of gates in space and time and the number of qudits acted on by each gate. (The gate itself is not specified.) For any architecture $A$, we obtain the associated Haar-random circuit family acting on qudits of constant dimension $q$, $C_{A,q}$, by specifying every gate in $A$ to be distributed according to the Haar measure and to act on qudits of dimension $q$ which are initialized in a product state $\ket{1}^{\otimes(L \times f(L))}$.
\end{definition}
In this paper, we only consider architectures that are constant depth and spatially 2-local (that is, a gate either acts on a single site or two adjacent sites); therefore, ``architecture'' for our purposes always refers to a constant-depth spatially 2-local architecture.  The above definition permits architectures for which the layout of the circuit itself may be different for different sizes. However, it is natural for a circuit architecture to be spatially periodic, and furthermore for the ``unit cells'' of the architecture to be independent of $L$. We formalize this as a notion of \emph{uniformity}, which we define more precisely below.
\begin{definition}[Uniformity]
  We call a constant-depth architecture $A$ \emph{uniform} if there exists some spatially periodic circuit layout $B$ on an infinite square lattice such that, for all positive integers $L$, $A(L)$ is a restriction of $B$ to a rectangular sub-grid with sidelengths $L \times f(L)$ for some $f(L) \leq \poly(L)$. A random circuit family $C_{A,q}$ associated with a uniform architecture $A$ is said to be a uniform random circuit family.
\end{definition}
While uniformity is a natural property for a circuit architecture to possess, our provable separations are with respect to certain non-uniform circuit families. In particular, we prove that for any fixed $0 < c < 1$, there exists some non-uniform circuit architecture $A$ acting on $n$ qubits such that, if $C_A$ is the Haar-random circuit family associated with $A$ acting on qubits,
\begin{enumerate}
\item There does not exist a $\poly(n)$-time classical algorithm that exactly samples from the output distribution of arbitrary realizations of $C_A$ unless the polynomial hierarchy collapses to the third level.
\item Given an arbitrary fixed output string $\vb{x}$, there does not exist a $\poly(n)$-time classical algorithm for computing the probability of obtaining $\vb{x}$, $|\bra{\vb{x}}C_A \ket{0}^{\otimes n}|^2$, up to additive error $2^{-\tilde{\Theta}(n^2)}$ with probability at least $1-1/\poly(n)$ over choice of circuit instance, unless a $\SharpP$-hard function can be computed in randomized polynomial time.
\item There exists a classical algorithm that runs in time $O(n)$ and, with probability at least $1-2^{-n^c}$ over choice of circuit instance, samples from the output distribution of $C_A$ up to error at most $2^{-n^c}$ in total variation distance.
\item There exists a classical algorithm that runs in time $O(n)$ and, for an arbitrary output string $\vb{x}$, with probability at least $1-2^{-n^c}$ over choice of circuit instance, estimates $|\bra{\vb{x}}C_A \ket{0}^{\otimes n}|^2$ up to additive error $2^{-n}/2^{n^c}$. (This should be compared with $2^{-n}$, which is the average output probability over choices of $\vb{x}$.)
\end{enumerate}
The first two points above follow readily from prior works (respectively \cite{terhal2002adaptive} and \cite{movassagh2019}), while the latter two follow from an analysis of the behavior of one of our simulation algorithms for this architecture. These algorithms improve on the previously best known simulation time for this family of architectures of  $2^{\Theta(L)} = 2^{\Theta(n^{c'})}$ for some constant $c'(c) < 1$ based on an exact simulation based on tensor network contraction.  We refer to the architectures for which we prove the above separations as ``extended brickwork architectures'' (see \Cref{fig:extendedBrickwork} for a specification), as they are related to the ``brickwork architecture'' \cite{broadbent2009universal} studied in the context of MBQC.

\paragraph{Implications for worst-to-average-case reductions for random circuit simulation.}
Very recently, it was shown \cite{movassagh2019} that for any random circuit family with Haar-random gates for which it is $\SharpP$-hard to compute output probabilities in the worst case, there does not exist a $\poly(n)$-time algorithm for computing the output probability of some arbitrary output string $\vb{x}$ up to additive error $2^{-\tilde{\Theta}(n^2)}$ with high probability over the circuit realization, unless there exists a $\poly(n)$-time randomized algorithm for computing a $\SharpP$-hard function. Essentially, for Haar-random circuits, near-exact average-case computation of output probabilities is as hard as worst-case computation of output probabilities. Our results described above imply that the error tolerance for this hardness result cannot be improved to $2^{-n}/2^{n^{c}}$ for any $c < 1$.

This hardness result follows other prior work  \cite{bouland2019complexity, movassagh2018efficient} on the average-case hardness of random circuit simulation. In particular, the original paper \cite{bouland2019complexity} uses a different interpolation scheme than that used in \cite{movassagh2018efficient, movassagh2019}. Interestingly, as discussed in \Cref{app:supremacy}, we find that the interpolation scheme of \cite{bouland2019complexity} cannot be used to prove hardness results about our algorithms' performance on $C_A$, despite $C_A$ possessing worst-case hardness. While this observation may be of technical interest for future work on worst-to-average-case reductions for random circuit simulation, the alternative interpolation scheme of \cite{movassagh2019} does not suffer from this limitation.

While \cite{bouland2019complexity, movassagh2018efficient, movassagh2019} prove hardness results for the near-exact computation of output probabilites of random circuits, it is ultimately desirable to prove hardness for the Random Circuit Sampling (RCS) problem of sampling from the output distribution of a random circuit with small error in variational distance, as this is the computational task corresponding to the problem that the quantum computer solves. \emph{A priori}, one might hope that such a result could be proved via such a worst-to-average-case reduction. In particular, it was pointed out in these works that improving the error tolerance of the hardness result to  $2^{-n}/\poly(n)$ would be sufficient to prove hardness of RCS. Our work rules out such a proof strategy working by showing that even improving the error tolerance to $2^{-n}/2^{n^{c}}$ is unachievable. In particular, any proof of the hardness of RCS should be sensitive to the depth and should not be applicable to the worst-case-hard shallow random circuit ensembles that admit approximate average-case classical simulations.

\subsubsection{Conjectures for uniform architectures}
While the above results are provable, they are unfortunately proved with respect to a contrived non-uniform architecture, and furthermore do not provide good insight into how the simulation runtime scales with simulation error and simulable circuit fraction. An obvious question is then whether efficient classical simulation remains possible for ``natural'' random circuit families that are sufficiently shallow, and if so, how the runtime scales with system size and error parameters. We argue that it does, but that a computational phase transition occurs for our algorithms when the depth ($d$) or local Hilbert space dimension ($q$) becomes too large. Here we are studying the simulation cost as $n\to\infty$ for fixed $d$ and $q$.  Intuitively, there are many constant-depth random circuit families for which efficient classical simulation is possible, including many ``natural'' circuit architectures (it seems plausible that \emph{any} depth-3 random circuit family on qubits is efficiently simulable). However, we expect a computational phase transition to occur for sufficiently large constant depths or qudit dimensions, at which point our algorithms become inefficient. The location of the transition point will in general depend on the details of the architecture. The conjectures stated below are formalizations of this intuition.

We now state our conjectures more precisely. \Cref{con:efficient} essentially states that there are \emph{uniform} random circuit families for which worst-case simulation (in the sense of sampling or computing output probabilities) is hard, but approximate average-case simulation can be performed efficiently. (Worst-case hardness for computing probabilities also implies a form of average-case hardness for computing probabilities, as discussed above.)  This is stated in more-or-less the weakest form that seems to be true and would yield a polynomial-time simulation. However, we suspect that the scaling is somewhat more favorable.  Our numerical simulations and toy models are in fact consistent with a stronger conjecture, \Cref{con:aggressive}, which if true would yield stronger run-time bounds.  Conversely, \Cref{con:transition} states that if the depth or local qudit dimension of such an architecture is made to be a  sufficiently large constant, our two proposed algorithms experience computational phase transitions and become inefficient even for approximate average-case simulation.

\begin{customcon}{1}\label{con:efficient}
There exist uniform architectures and choices of $q$ such that, for the associated random circuit family $C_{A,q}$, (1) worst-case simulation of $C_{A,q}$ (in terms of sampling or computing output probabilities) is classically intractable unless the polynomial hierarchy collapses, and (2) our algorithms approximately simulate $C_{A,q}$ with high probability. More precisely, given parameters $\varepsilon$ and $\delta$, our algorithms run in time bounded by $\poly(n,1/\varepsilon,1/\delta)$ and can, with probability  $1-\delta$ over the random circuit instance, sample from the classical output distribution produced by $C_q$ up to variational distance error $\varepsilon$ and compute a fixed output probability up to additive error $\varepsilon/q^n$.
\end{customcon}

\begin{customcon}{1'}\label{con:aggressive}
For any uniform random circuit family $C_{A,q}$ satisfying the conditions of \Cref{con:efficient}, efficient simulation is possible with runtime replaced by $n^{1+o(1)}\cdot \exp(O(\sqrt{\log(1/\varepsilon\delta)}))$.
\end{customcon}

\begin{customcon}{2}\label{con:transition}
For any uniform random circuit family $C_{A,q}$ satisfying the conditions of \Cref{con:efficient}, there exists some constant $q^*$ such that our algorithms become inefficient for simulating $C_{A,q'}$ for any constant $q' \geq q^*$, where $C_{A,q'}$ has the same architecture as as $C_q$ but acts on qudits of dimension $q'$.  There also exists some constant $k^*$ such that, for any constant $k \geq k^*$, our algorithms become inefficient for simulating the composition of $k$ layers of the random circuit, $C_{A,q}^k \circ \dots \circ C_{A,q}^2 \circ C_{A,q}^1$, where each  $C_{A,q}^i$ is i.i.d. and distributed identically to $C_{A,q}$. In the inefficient regime, for fixed $\varepsilon$ and $\delta$ the runtime of our algorithms is $2^{O(L)}$.
\end{customcon}

Our evidence for these conjectures, which we elaborate upon below, consists primarily of (1) a rigorous reduction from the 2D simulation problem to a 1D simulation problem that can be efficiently solved with high probability if certain conditions on expected entanglement in the 1D state are met, (2) convincing numerical evidence that these conditions are indeed met for a specific worst-case-hard uniform random circuit families and that in this case the algorithm is extremely successful in practice, and (3) heuristic analytical evidence for both conjectures using a mapping from random unitary circuits to classical statistical mechanical models, and for \Cref{con:aggressive} using a toy model which can be more rigorously studied. The uniform random circuit family for which we collect the most evidence for classical simulability is associated with the depth-3 ``brickwork architecture'' \cite{broadbent2009universal} (see also \Cref{fig:extendedBrickwork} for a specification). We now present an overview of these three points, which are then developed more fully in the body of the paper.

\subsection{Overview of proof ideas and evidence for conjectures}\label{se:overview}
\subsubsection{Reduction to ``unitary-and-measurement'' dynamics.}
{
We reduce the problem of simulating a constant-depth quantum circuit acting on a $L \times L'$ grid of qudits to the problem of simulating an associated ``effective dynamics'' in 1D on $L$ qudits which is iterated for $L'$ timesteps, or alternatively on $L'$ qudits which is iterated for $L$ timesteps. This mapping is rigorous and is related to previous maps from 2D quantum systems to 1D system evolving in time~\cite{raussendorf2001one,kim2017holographic,kim2017noiseresilient}. The effective 1D dynamics is then simulated using the time-evolving block decimation algorithm of Vidal \cite{vidal2004efficient}. In analogy, we call this algorithm space-evolving block decimation (\texttt{SEBD}).  We rigorously bound the simulation error made by the algorithm in terms of quantities related to the entanglement spectra of the effective 1D dynamics and give conditions in which it is provably asymptotically efficient for sampling and estimating output probabilities with small error. \texttt{SEBD} is self-certifying in the sense that it can construct confidence intervals for its own simulation error and for the fraction of random circuit instances it can simulate --- we use this fact later to bound the error made in our numerical experiments.

A 1D unitary quantum circuit on $L$ qubits iterated for $L^{c}$ timesteps with $c > 0$ is generally hard to simulate classically in $\poly(L)$-time, as the entanglement across any cut can increase linearly in time. However, the form of 1D dynamics that a shallow circuit maps to includes measurements as well as unitary gates. While the unitary gates tend to build entanglement, the measurements tend to destroy entanglement and make classical simulation more tractable. It is \emph{a priori} unclear which effect has more influence. To illustrate the mapping, we introduce  the  simple worst-case-hard random circuit family consisting of a Haar-random single qubit gate applied to each site of a cluster state, and obtain an exact, closed-form expression for the effective 1D unitary-and-measurement dynamics simulated by \texttt{SEBD}. We call this model \texttt{CHR} for ``cluster-state with Haar-random measurements''.

Fortunately, unitary-and-measurement processes have been studied in a flurry of recent papers from the physics community \cite{li2018quantum, chan2018weak, skinner2019measurement, li2019measurement, szyniszewski2019entanglement, choi2019quantum, gullans2019dynamical, bao2019theory, jian2019measurement,gullans2019scalable,zabalo2019critical}. The consensus from this work is that processes consisting of entanglement-creating unitary evolution interspersed with entanglement-destroying measurements can be in one of two phases, where the entanglement entropy equilibrates to either an area law (constant), or to a volume law (extensive). When we vary parameters like the fraction of qudits measured between each round of unitary evolution, a phase transition is observed. The existence of a phase transition appears to be robust to variations in the exact model, such as replacing projective measurements on a fraction of the qudits with weak measurements on all of the qudits \cite{li2019measurement,szyniszewski2019entanglement}, or replacing Haar-random unitary evolution with Clifford \cite{li2018quantum,li2019measurement,gullans2019dynamical,choi2019quantum} or Floquet \cite{skinner2019measurement,li2019measurement} evolution. This suggests that the efficiency of the \texttt{SEBD} algorithm  depends on whether the particular circuit depth and architecture being simulated yields effective 1D dynamics that falls within the area-law or the volume-law regime. It also suggests a computational phase transition in the complexity of the \texttt{SEBD} algorithm. Essentially, decreasing the measurement strength or increasing the qudit dimension in these models is associated with moving toward a transition into the volume-law phase. Since increasing the 2D circuit depth is associated with decreasing the measurement strength and increasing the local dimension of the associated effective 1D dynamics, this already gives substantial evidence in favor of a computational phase transition in \texttt{SEBD}.

\texttt{SEBD} is provably inefficient if the effective 1D dynamics are on the volume-law side of the transition, and we expect it to be efficient on the area-law side because, in practice, dynamics obeying an area law for the von Neumann entanglement entropy are generally efficiently simulable. However, definitively proving that \texttt{SEBD} is efficient on the area-law side faces the obstacle that there are known contrived examples of states which obey an area law but cannot be efficiently simulated with matrix product states \cite{schuch2008entropy}.  We address this concern by directly studying the entanglement spectrum of unitary-and-measurement processes in the area-law phase. To do this, we introduce a toy model for such dynamics which may be of independent interest. For this model, we rigorously derive an asymptotic scaling of Schmidt values across some cut as $\lambda_i \propto \exp(-\Theta(\log^2 i))$ which is consistent with the scaling observed in our numerical simulations.  Moreover, for this toy model we show that with probability at least $1-\delta$, the equilibrium state after iterating the process can be $\varepsilon$-approximated by a state with Schmidt rank $r \leq \exp(O(\sqrt{\log(n/\varepsilon \delta)}))$. Taking this toy model analysis as evidence that the bond dimension of \texttt{SEBD} when simulating a circuit whose effective 1D dynamics is in an area-law phase obeys this  asymptotic scaling leads to \Cref{con:aggressive}.
}

\subsubsection{Proof idea for rigorous complexity separation with non-uniform architectures.}
We now explain the proof idea for the complexity separation discussed above. The main idea is that we define a non-uniform architecture such that, when all gates are Haar-random, the effective 1D dynamics consists of alternating rounds of unitary and measurement layers where in each measurement layer, a weak measurement is applied to each qubit. While the problem of rigorously proving an area-law/volume-law transition for general unitary-and-measurement processes is still open for the case of measurements with \emph{constant} measurement strength, for our particular architecture the measurement strength itself increases rapidly as the system size $n$ is increased (this is achieved using the non-uniformity of the architecture). Intuitively, then, by considering large enough $n$, the weak measurements can be made strong enough to destroy nearly all of the pre-existing entanglement in the 1D state.  This is the key idea that allows us to prove that, by always compressing the MPS describing the state to one of just constant bond dimension, the error incurred is very low and efficient simulation is possible.

The fact that the effective measurement strength increases rapidly with system size follows from a technical result about the decay of expected post-measurement entanglement entropy when a contiguous block of qubits in a state produced by a 1D random local circuit is measured. In particular, we show that if $\ket{\psi}_{ABC}$ is a 1D state produced by a depth-2 circuit with Haar-random 2-local gates, and subregion $B$ is measured in the computational basis, then the expected entanglement entropy on the post-measurement state satisfies
\begin{equation}
\E_b S(A)_{\psi_b} \leq 2^{-\Theta(|B|)}
\end{equation}
where $\ket{\psi_b}_{AC}$ is the post-measurement state after obtaining measurement outcome $b$, and $|B|$ denotes the number of qubits in region $B$. (While we only need and only prove this result for depth-2 circuits, we expect it to remain true for any constant depth.) A similar result was obtained by Hastings in the context of nonnegative wavefunctions rather than random shallow circuits \cite{hastings2016how}.

\subsubsection{Numerical evidence for conjectures.}
{
\paragraph{Excellent empirical performance of \texttt{SEBD}.} We numerically implemented the \texttt{SEBD} algorithm to simulate two different universal (for MBQC) uniform shallow circuit architectures with randomly chosen gates. The first model is depth-3 circuits with ``brickwork architecture'' (see \Cref{fig:brickwork} for a specification) and Haar-random two-qubit gates\footnote{A similar architecture but with a specific choice of gates rather than random gates has been proposed as a candidate for demonstrating quantum computational supremacy \cite{gao2017quantum}; \texttt{SEBD} becomes inefficient if such ``worst-case'' gates are chosen, as the effective 1D dynamics becomes purely unitary evolution without measurements in this case.}.  We found that a laptop using non-optimized code could efficiently simulate typical instances on a $409\times 409$ grid with variational distance error less than $0.01$. The mean runtime (averaged over random circuit instance) was on the order of one minute per sample. In principle the same algorithm could also be used to compute output probabilities with small additive error.  We also simulated the \texttt{CHR} model as defined above on lattices of sizes up to $50\times 50$. The slower-decaying entanglement spectrum of the effective 1D dynamics of the \texttt{CHR} model causes the maximal lattice size that we can simulate to be smaller, but allows the functional form of the spectrum to be better studied numerically, helping establish the asymptotic efficiency of the simulation as discussed below.

It is useful to compare our observed runtime with what is possible by previously known methods. The previously best-known method that we are aware of for computing output probabilities for these architectures would be to write the circuit as a tensor network and perform the contraction of the network \cite{villalonga2019establishing}. The cost of this process scales exponentially in the tree-width of a graph related to the quantum circuit, which for a 2D circuit is thought to scale roughly as the surface area of the minimal cut slicing through the circuit diagram, as in Eq.~\eqref{eq:2D-contract}. By this reasoning, we estimate that simulating a circuit with brickwork architecture on a $400 \times 400$ lattice using tensor network contraction would be roughly equivalent to simulating a depth-$40$ circuit on a $20 \times 20$ lattice with the architecture considered in \cite{villalonga2019establishing}, where the entangling gates are $CZ$ gates. We see that these tasks should be equivalent because the product of the dimensions of the bonds crossing the minimal cut is equal to $2^{200}$ in both cases: for the brickwork circuit, 100 gates cross the cut if we orient the cut vertically through the diagram in \Cref{fig:brickwork}(a) and each gate contributes a factor of 4; meanwhile, for the depth-$40$ circuit, only one fourth of the unitary layers will contain gates that cross the minimal cut, and each of these layers will have 20 such gates that each contribute a factor of 2 ($CZ$ gates have half the rank of generic gates). The task of simulating a depth-$40$ circuit on a $7 \times 7$ lattice was reported to require more than two hours using tensor network contraction on the 281 petaflop supercomputer Summit \cite{villalonga2019establishing}, and the exponentiality of the runtime suggests scaling this to $20 \times 20$ would take many orders of magnitude longer, a task that is decidedly intractable.

However, we emphasize that the important takeaway from the numerics is not merely the large circuit sizes we simulated but rather the fact that the numerics serve as evidence that the algorithm is efficient in the asymptotic limit $n \rightarrow \infty$, as we discuss presently.

\paragraph{Evidence for asymptotic efficiency.} Our simulations indicate that the average entanglement generated during the effective 1D dynamics of the \texttt{SEBD} algorithm for these architectures quickly saturates to a value independent of the number of qubits --- an area law for entanglement entropy. In fact, our numerics indicate that the effective 1D dynamics not only obeys an area law for the von Neumann entropy, but also has the stronger property of obeying an area law for some R\'{e}nyi entropies $S_\alpha$ with $\alpha < 1$. It has been proven that such a condition implies efficient representation by matrix product states \cite{schuch2008entropy}, providing  evidence for efficiency in the asymptotic limit.

To obtain a more precise estimate of the error of \texttt{SEBD} and attain further evidence of asymptotic efficiency, we also study the form of the entanglement spectra throughout the effective 1D dynamics. We observe that the spectrum of Schmidt values obeys a superpolynomial decay consistent with that predicted by our toy model for unitary-and-measurement dynamics in the area-law phase, providing further validation of our toy model which suggests that not only is \texttt{SEBD} polynomial-time, but obeys the even better scaling with error parameters given in \Cref{con:aggressive}. Overall, the numerics strongly support \Cref{con:efficient} and support the toy model which is the basis for the more aggressive \Cref{con:aggressive}.
}

\subsubsection{Analytical evidence for conjectures from statistical mechanics.}
{
In addition to providing strong numerical evidence that \texttt{SEBD} is efficient in the cases we considered, we also give analytical arguments for both algorithms' efficiency when acting on the depth-3 brickwork architecture using methods from statistical mechanics. We focus on the depth-3 brickwork architecture because it is a worst-case hard uniform architecture which is simple enough to be studied analytically. We also give evidence of computational phase transitions as qudit dimension and circuit depth are increased.

We map 2D shallow circuits with Haar-random gates to classical statistical mechanical models, utilizing techniques developed in   \cite{nahum2018operator,von2018operator,zhou2019emergent,hunter2019unitary,bao2019theory,jian2019measurement}, such that the free energy cost incurred by twisting boundary conditions of the stat mech model corresponds to quantities $\tilde{S}_k$, which we refer to as ``quasi-entropies'' of the output state of the quantum circuit. The quasi-entropy of index $k$ is related but not exactly equal to the R\'{e}nyi-$k$ entanglement entropy averaged over random circuit instances and measurement outcomes, denoted by  $\langle S_k \rangle$. We briefly define the quasi-entropy associated with a collection of states here. All logarithms are base-2 unless indicated otherwise.

\begin{definition}[Quasi-$k$ entropy]
For a collection $\mathcal{E} = \{ \rho_i \}$ of non-normalized bipartite states on system $AB$, we define the quasi-$k$ entropy of register $A$ as
\begin{equation}
\tilde{S}_k(A) = \frac{1}{1-k} \log\left(  \frac{\sum_i \tr(\rho_{i,A}^k)}{\sum_i \tr(\rho_{i})^k}  \right)
\end{equation}
where $\rho_{i,A}$ is the reduced state of $\rho_i$ on subregion $A$.
\end{definition}
Virtually identical quantities were also considered in two other very recent works \cite{bao2019theory,jian2019measurement}. Notably, in the $k \rightarrow 1$ limit, these quantities approach the expected von Neumann  entropy $\langle S(A) \rangle$ achieved when state $\rho_i$ is drawn from $\mathcal{E}$ with probability proportional to $\tr(\rho_i)$:
\begin{equation}
\lim_{k\rightarrow 1} \tilde{S}_k(A) = \langle S(A) \rangle = \sum_i \left(\frac{\tr(\rho_i)}{\sum_j \tr(\rho_j)}\right) S(A)_{\rho_i}
\end{equation}
Although the quasi-entropies $\tilde{S}_k$ are not the entropic quantites that directly relate to the runtime of our algorithms, we study them because the stat mech mapping permits for an analytical handle on $\tilde{S}_k$ for integer $k \geq 2$, and the calculations become especially tractable for $k=2$.  Essentially, changing the qudit dimension $q$ of the random circuit model corresponds to changing the interaction strengths in the associated stat mech model. Phase transitions in the classical stat mech model are accompanied by phase transitions in quasi-entropies. While the efficiency of our algorithms is related to different entropic quantities, which are hard to directly analyze, the phase transition in quasi-entropies  provides analytical evidence in favor of our conjectures, as we outline below.

\paragraph{Area-law to volume-law transition for $\tilde{S}_2$ in brickwork architecture.}
Since we can only analytically access quasi-entropies, our argument falls short of a rigorous proof that \texttt{SEBD} is efficient and experiences a computational phase transition, but contributes the following conclusions about the entanglement in the effective 1D dynamics.
\begin{enumerate}
    \item For the brickwork architecture, $\tilde{S}_2$ for the collection of pure states encountered by the \texttt{SEBD} algorithm satisfies an area law when qudits have local dimension $q=2$ (qubits) or $q=3$ (qutrits).
    \item For brickwork architecture, $\tilde{S}_2$ transitions to a volume-law phase once qudit local dimension $q$ becomes sufficiently large. The critical point separating the two phases is estimated to be roughly $q_c\approx 6$ and could be precisely computed with standard Monte Carlo techniques.
\end{enumerate}
The fact that $\tilde{S}_2$ obeys an area law for brickwork architecture with $q=2$ corroborates our other evidence that \texttt{SEBD} is efficient in this regime.

\paragraph{Phase transitions for $\tilde{S}_2$ in arbitrary architectures.}
The stat mech mapping can also be used to understand the behavior of $\tilde{S}_2$ during $\texttt{SEBD}$ for more general architectures. In particular, for any architecture, the mapping implies an area-law-to-volume-law phase transition in $\tilde{S}_2$ as $q$ is increased, contributing further evidence in favor of a computational phase transition driven by $q$. We also present a heuristic argument for why the existence of a computational phase transition as a function of $q$ should always be accompanied by the existence of a phase transition as a function of the number of layers of the random circuit. Together with prior work on phase transitions in unitary-and-measurement models as measurement strength and qudit dimension are changed, this provides further evidence of computational phase transitions as stated in \Cref{con:transition}.
}

\paragraph{\texttt{Patching} algorithm and transitions in quasi-conditional mutual information.}
We can use the stat mech mapping to study average entropic properties of the classical output distribution of the quantum circuit. In particular, we define a ``quasi-$k$ conditional mutual information'' for the distribution over classical output distributions which is parameterized by the real number $k$ and approaches the average conditional mutual information (CMI) in the $k \rightarrow 1$ limit. We argue that if a random circuit has an associated stat mech model that is disordered, then the quasi-2 CMI of the distribution over classical output distributions is exponentially decaying in the sense that $\tilde{I}_2 (A:C|B) \leq \poly(n) e^{-\Omega(\text{dist}(A,C))}$ for any lattice subregions $A$, $B$, and $C$, where $\text{dist}(A,C)$ is the distance between subregions $A$ and $C$.

Taking this as evidence that the average CMI obeys the same exponential decay condition, we show that a very different circuit simulation algorithm which we call \texttt{Patching} can also be used to efficiently simulate the random circuit with high probability. This simulation algorithm, based on \cite{brandao2019finite} but improving on the runtime of that algorithm (from quasi-polynomial to polynomial) in our setting of shallow circuits, works by exactly simulating disconnected subregions before applying recovery maps to ``stitch''  the regions together. We obtain rigorous bounds on the performance of \texttt{Patching} in terms of the rate of decay of CMI of the output distribution and give conditions in which it is asymptotically efficient.

On the other hand, when the corresponding stat mech model transitions to an ordered phase as $q$ is increased, the quasi-CMI does not decay to zero as $\dist(A,C)$ is increased, providing evidence against CMI decay and against the efficiency of  \texttt{Patching}. In fact, in the limit of infinitely large qudit dimension $q\rightarrow \infty$, by formally evaluating all quasi-$k$ CMIs and taking the $k\rightarrow 1$ limit  we show that the expected CMI of the classical output distribution between three regions forming a tripartition of the lattice becomes equal to a constant: $\langle I(A:C|B) \rangle =  (1-\gamma)/ \ln 2 \approx 0.61$ where $\gamma$ is the Euler constant. Unfortunately, performing the analytic continuation and hence exactly evaluating von Neumann entropies is difficult outside of the $q\rightarrow \infty$ limit.

\subsection{Future work and open questions}

Our work yields several natural follow-up questions and places for potential future work. We list some here.

\begin{enumerate}

    \item Can ideas from our work also be used to simulate \emph{noisy} 2D quantum circuits? Roughly, we expect that increasing noise in the circuit corresponds to decreasing the interaction strength in the corresponding stat mech model, pushing the model closer toward the disordered phase, which is (heuristically) associated with efficiency of our algorithms. We therefore suspect that if noise is incorporated, there will be a 3-dimensional phase diagram depending on circuit depth, qudit dimension, and noise strength. As the noise is increased, our algorithms may therefore be able to simulate larger depths and qudit dimensions than in the noiseless case.

    \item Can one approximately simulate random 2D circuits of arbitrary depth? This is the relevant case for Google's quantum computational supremacy experiment \cite{googleSupremacy}. Assuming \Cref{con:transition}, our algorithms are not efficient once the depth exceeds some constant, but it is not clear if this difference in apparent complexity for shallow vs.~deep circuits is simply an artifact of our simulation method, or if it is inherent to the problem itself.

    \item Our algorithms are well-defined for all 2D circuits, not only random 2D circuits. Are they also efficient for other kinds of unitary evolution at shallow depths, for example evolution by a fixed local 2D Hamiltonian for a short amount of time?

    \item Can we rigorously prove \Cref{con:efficient}? One way to make progress on this goal would be to find a worst-case-hard uniform circuit family for which it would be possible to perform the analytic continuation of quasi-entropies $\tilde{S}_k$ in the $k \rightarrow 1$ limit using the mapping to stat mech models.

    \item Can we give numerical evidence for \Cref{con:transition}, which claims that our algorithms undergo computational phase transitions? This would require numerically simulating our algorithms for circuit families with increasing local Hilbert space dimension and increasing depth and finding evidence that the algorithms eventually become inefficient.

   \item How precisely does the stat mech mapping inform the  efficiency of our algorithms? Is the correlation length of the stat mech model associated with the runtime of our simulation algorithms? How well does the phase transition point in the stat mech model (and accompanying phase transition in quasi-entropies) predict the computational phase transition point in the simulation algorithms? If such questions are answered, it may be possible to predict the efficiency and runtime of the simulation algorithms for an arbitrary (and possibly noisy) random circuit distribution via Monte Carlo studies of the associated stat mech model. In this way, the performance of the algorithms could be studied even when direct numerical simulation is not feasible.

    \item In the regime where \texttt{SEBD} is inefficient, i.e., when the effective 1D dynamics it simulates are on the volume-law side of the entanglement phase transition, is \texttt{SEBD} still better than previously known exponential-time methods? Intuitively, we expect this to be the case close to the transition point.

\end{enumerate}

\subsection{Outline for remainder of paper}
We now outline the material of the remaining sections. The logical dependencies between subsections are described in the following figure. The paper may be read linearly without issue, but readers interested in only one specific aspect of our results may skip subsections outside the relevant chain of dependencies illustrated by the diagram.
\begin{figure}[ht]
\centering
\includegraphics[width=0.3\textwidth]{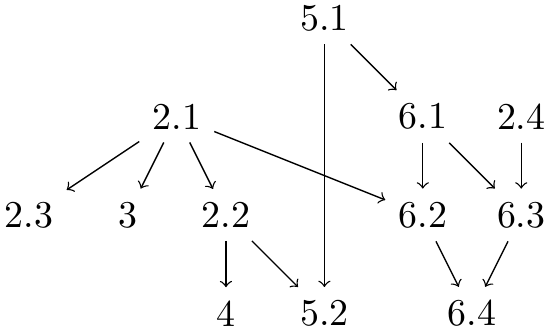}
\end{figure}

\begin{itemize}
	\item In \Cref{se:algorithms}, we specify and analyze our two algorithms for simulating 2D circuits, obtaining rigorous bounds on the runtime in terms of various parameters. In \Cref{sec:SEBD}, we specify and analyze the \texttt{SEBD} algorithm. In \Cref{se:CHR}, we give an explicit example of how the efficiency of the algorithm is related to the simulability of associated unitary-and-measurement processes. In \Cref{sec:conjectured}, we study a toy model for unitary-and-measurement dynamics in an area-law phase, which (along with supporting numerics) is the basis for \Cref{con:aggressive}. In \Cref{sec:patching}, we specify and analyze the \texttt{Patching} algorithm.

	\item In \Cref{se:proof}, we prove the complexity separation for a particular architecture. To understand this section, one need only read \Cref{sec:SEBD}.

    \item In \Cref{se:numerics}, we present numerical evidence supporting \Cref{con:efficient} and \Cref{con:aggressive}.

    \item In \Cref{se:mapping}, we introduce the mapping from random circuits to stat mech models. We first review the mapping technique developed in prior works before, in \Cref{se:weakmeasurementmapping}, applying the mapping to analyze a particular family of 1D circuits with weak measurements that correspond to the effective 1D dynamics of the cluster state model we study in \Cref{se:CHR}.
    \item In \Cref{se:2dcircuitmapping}, we apply the stat mech mapping to shallow random 2D circuits. In \Cref{se:statmechsebd} and \Cref{se:statmechpatching}, we show how the stat mech mapping supports \Cref{con:efficient} and \Cref{con:transition} with respect to \texttt{SEBD} and \texttt{Patching}, respectively. In \Cref{sec:brickwork}, we study the stat mech mapping in more detail for the ``brickwork'' architecture, strengthening the evidence for \Cref{con:efficient} and \Cref{con:transition}. To understand the contents of \Cref{se:statmech2D}, one must read \Cref{se:mapping} but \Cref{se:weakmeasurementmapping} may be skipped.
\end{itemize}

\section{Algorithms}\label{se:algorithms}
We now propose and analyze two algorithms for sampling from the output distributions of shallow 2D random quantum circuits. The first algorithm is based on a reduction from the 2D simulation problem to a 1D simulation problem, which is then simulated with the time-evolving block decimation (\texttt{TEBD}) algorithm \cite{vidal2004efficient}. We therefore refer to this algorithm as \emph{space-evolving block decimation} (\texttt{SEBD}). Essentially, the resulting algorithm is efficient if the 1D state in the corresponding effective 1D dynamics can be approximately represented as an MPS of polynomially bounded bond dimension at all times. We also discuss how a variation of \texttt{SEBD} can be used to compute output probabilities with small additive error.

We analyze the behavior of \texttt{SEBD} applied to a simple class of 2D shallow random circuits with a uniform circuit architecture that is universal for measurement-based quantum computation and is therefore believed to be hard to simulate in the worst case -- namely, the 2D cluster state with Haar-random single qubit measurements. We believe that this simple class of random 2D circuits qualitatively captures the behavior of many other random shallow circuit architectures. The benefit of studying this model is that we can obtain an exact, closed-form description of the behavior of our simulation algorithm for this problem. In particular, we show that the algorithm can be understood as simulating a 1D process involving alternating layers of random unitary gates and measurements. Such processes have been the subject of a number of recent works \cite{li2018quantum, chan2018weak, skinner2019measurement, li2019measurement, szyniszewski2019entanglement, choi2019quantum, gullans2019dynamical, zabalo2019critical},  which find evidence for the existence of an entanglement phase transition driven by the frequency and strength of the measurements from an ``area-law'' phase characterized by low entanglement to a ``volume-law'' phase characterized by large entanglement.

We introduce a toy model that intuitively captures how the entanglement spectrum of such a unitary-and-measurement process might scale in the area-law phase. The toy model predicts a superpolynomial decay of Schmidt values across any cut, a sufficient condition for efficient MPS representation. Later in \Cref{se:numerics}, we numerically observe the effective 1D dynamics of the uniform architectures we simulate to be in the area-law phase, with a decay of Schmidt values consistent with that predicted by the toy model. This physical picture provides strong evidence that our algorithm is efficient for the uniform architectures we considered. Later in \Cref{se:statmech2D}, we provide additional analytical evidence for this indeed being the case.

The second algorithm, which we call \texttt{Patching}, involves first sampling from the output distribution of small disconnected patches of the lattice, and then stitching them together to obtain a global sample. This algorithm is  efficient if the conditional mutual information (CMI) of the output distribution of the circuit is exponentially decaying in a sense that we make precise below.  The latter algorithm is essentially an adaptation of an algorithm for preparing Gibbs states with finite correlation length \cite{brandao2019finite}. However, by exploiting the fact that the distribution we want to sample from is classical and arises from a constant-depth local circuit, we are able to improve on a na\"{i}ve application of that scheme, obtaining a polynomial-time algorithm instead of the quasipolynomial-time algorithm obtained in \cite{brandao2019finite} if the conditional mutual information is exponentially decaying.

While the criteria required by these two algorithms for efficiency superficially appear unrelated, we find evidence that they are indeed related. Namely, in \Cref{se:statmech2D} we relate the efficiency criteria of both algorithms to phases of a statistical mechanical model associated with the random circuit family.  A stat mech model in the ordered phase suggests that the criteria for both algorithms is met, whereas a model in the disordered phase suggests that the criteria for both is not met.  It therefore is plausible that \texttt{SEBD} can efficiently simulate some random circuit family if and only if \texttt{Patching} can.

We assume the reader is familiar with standard tensor work methods, particularly algorithms for manipulating matrix product states (see e.g.~\cite{orus2014practical, bridgeman2017hand} for reviews).

\subsection{Space-evolving block decimation (\texttt{SEBD})}\label{sec:SEBD}
In this section, we introduce the \texttt{SEBD} algorithm for simulating a shallow 2D random circuit. In fact, the algorithm is well-defined for \emph{any} 2D circuit, but we present evidence that it is efficient for sufficiently shallow random circuit families with uniform architectures, and prove it is efficient for certain depth-3 random circuit families with non-uniform architectures in \Cref{se:proof}. We first give our algorithm for sampling and error bounds, before describing a modified version of the algorithm for computing output probabilities and error bounds associated with this version of the algorithm.

\subsubsection{Specification of algorithm}\label{specification}

\begin{figure}[tbp]
    \centering
    \includegraphics[width=0.9\columnwidth]{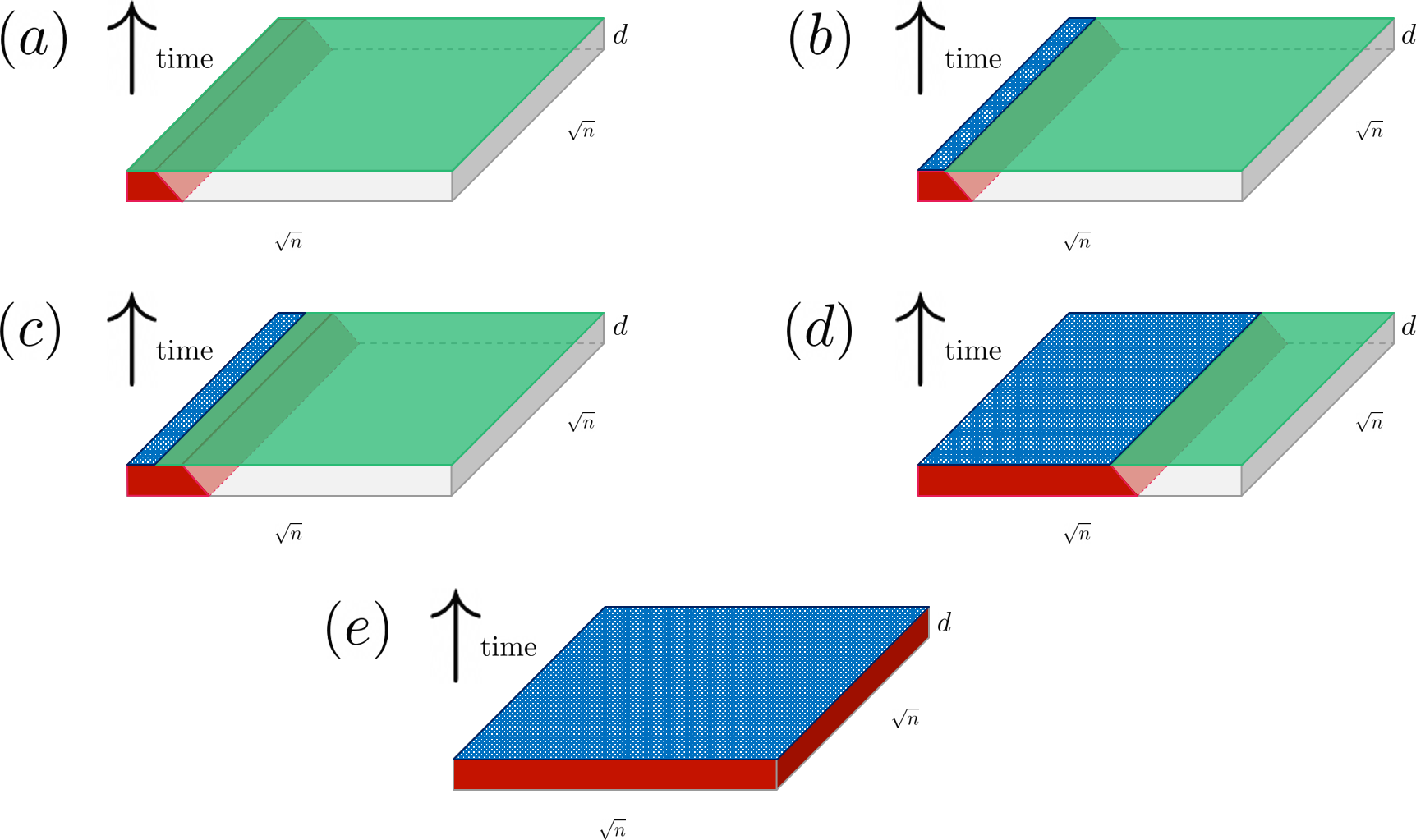}
    \caption{ \label{fg:sebd} Schematic depiction of \texttt{SEBD} acting on a square lattice with circuit depth $d$. In all figures, the 2D circuit is depicted as a spacetime volume, with time flowing upwards. The green (respectively dotted blue) region denotes unmeasured (measured) qudits. In (a), we apply all gates in the lightcone of \texttt{column 1}, namely, those gates intersecting the spacetime volume shaded red. In (b), we simulate the computational basis measurement of \texttt{column 1}.  In (c), we apply all gates in the lightcone of \texttt{column 2} that were previously unperformed. Figure (d) depicts the algorithm at an intermediate stage of the simulation, after the measurements of about half of the qudits have been simulated. The algorithm stores the current state as an MPS at all times, which may be periodically compressed to improve efficiency. Figure (e) depicts the algorithm at completion: the measurements of all $n$ of the qudits have been simulated.}
\end{figure}

\begin{figure}[htbp]\label{iteration}
    \centering
    \def\svgwidth{.7\columnwidth}
    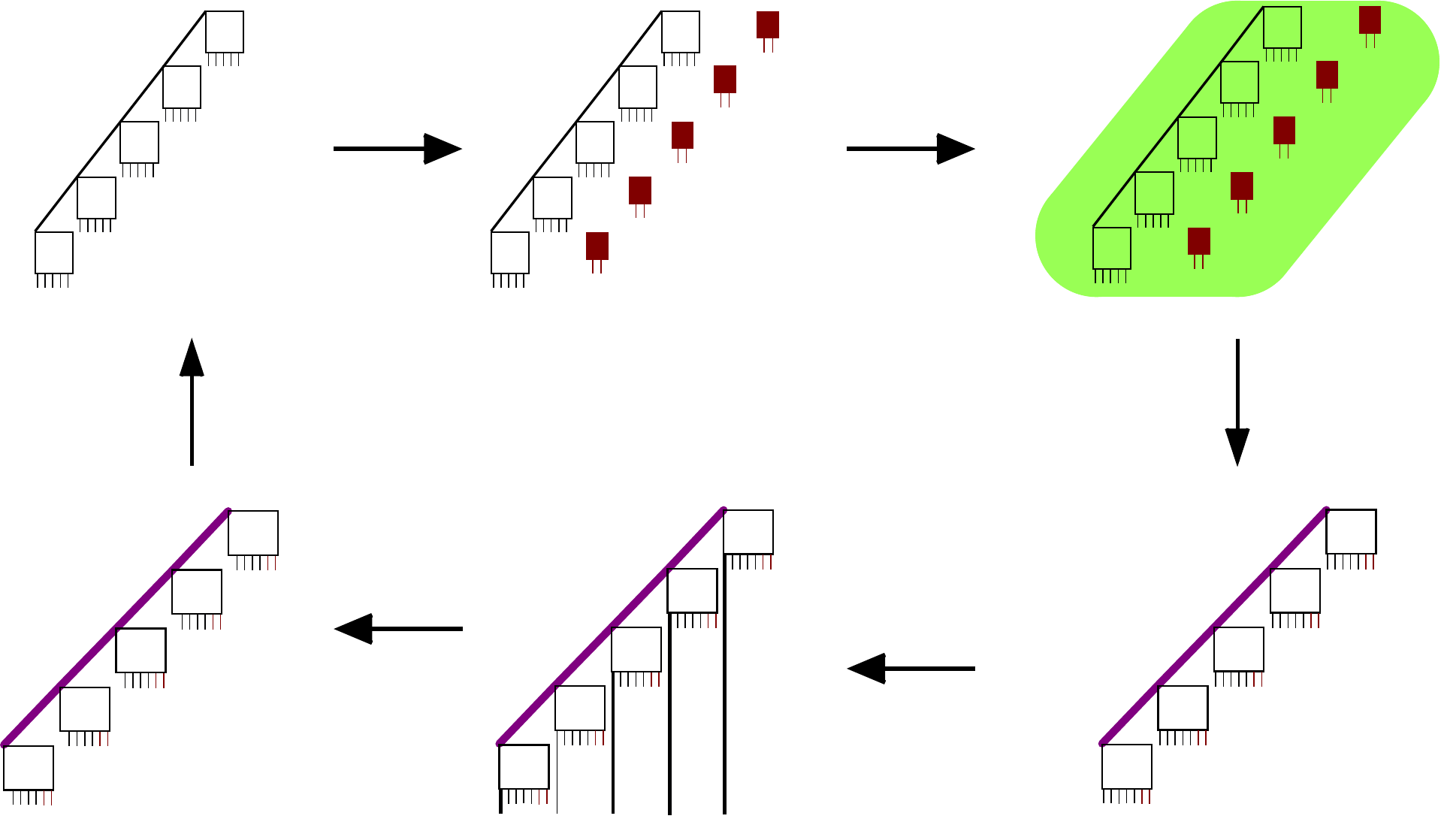
    \caption{ \label{fg:iteration} Iteration of \texttt{SEBD}. In (a), we begin with an MPS describing the current state $\rho_j$. In (b), the MPS is compressed via truncation of small Schmidt values. This will generally decrease the bond dimension of the MPS, depicted by the thin black line rather than thick purple line. In (c), qudits acted on by $V_j$ that are not already incorporated into the current state are added to the MPS (increasing the physical bond dimension of the MPS) and initialized in $\ket{0}$ states. In (d), the unitary gates associated with $V_j$ are applied. Figure (e) depicts the MPS after the application of $V_j$; the thick purple lines schematically illustrate the fact that the bond dimension may increase in this step. In (f), the measurement of \texttt{column j} is performed, and the outcome \texttt{01101} is obtained. Subsequently \texttt{column j} is projected onto \texttt{01101}, removing the physical legs associated with these sites from the MPS. The resulting state is $\rho_{j+1}$. }
\end{figure}

For concreteness, we consider a rectangular grid of qudits with local Hilbert space dimension $q$, although the algorithm could be similarly defined for different lattices. Assume WLOG that the grid consists of $n=L_1 \times L_2$ qudits, where $L_1$ is the number of rows, $L_2$ is the number of columns, and $L_1 \leq L_2$. For each qudit, let $\ket{i}, i \in [q]$ label a set of basis states which together form the computational basis.  Assume all gates act on one site or two neighboring sites, and the starting state is $\ket{1}^{\otimes n}$. Let $d$ denote the circuit depth, which should be regarded as a constant. For a fixed circuit instance $C$, the goal is to sample from a distribution close to $\mathcal{D}_C$, defined to be the distribution of the output of $C$ upon measuring all qudits in the computational basis. For an output string $\vb{x} \in [q]^n$, we let $\mathcal{D}_C(\vb{x})$ denote the probability of the circuit outputting $\vb{x}$ after measurement. The high-level behavior of the algorithm is illustrated in \Cref{fg:sebd}.  Recall that $C$ can always be \emph{exactly} simulated in time $L_2 q^{\Theta(d L_1)}$ using standard tensor network algorithms \Cite{markov2008simulating}.

Since all of the single-qudit measurements commute, we can measure the qudits in any order. In particular, we can first measure all of the sites in \texttt{column 1}, then those in \texttt{column 2}, and iterate until we have measured all $L_2$ columns.  This is the measurement order we will take. Now, consider the first step in which we measure \texttt{column 1}. Instead of applying all of the gates of the circuit and then measuring, we may instead apply only the gates in the \emph{lightcone} of \texttt{column 1}, that is, the gates that are causally connected to the measurements in \texttt{column 1}. We may ignore qudits that are outside the lightcone, by which we mean qudits that are outside the support of all gates in the lightcone.

Let $\rho_1 = \dyad{1}^{\otimes L_1}$ denote the trivial starting state that is a tensor product of $\ket{1}$ states in \texttt{column 1}, which the algorithm represents as an MPS. Let $V_1$ denote the isometry corresponding to applying all gates in the lightcone of this column. The algorithm simulates the application of $V_1$ by adding qudits in the lightcone of \texttt{column 1} as necessary and applying the associated unitary gates, maintaining the description of the state as an MPS of length $L_1$ as illustrated  in \Cref{fg:iteration}. Since there are up to $d+1$ columns in the lightcone of \texttt{column 1}, each tensor of the MPS after the application of $V_1$ has up to $d+1$ dangling legs corresponding  to physical indices, for a total physical dimension of at most  $q^{d+1}$. Since in the application of $V_1$, there are up to  $O(d^2)$ gates that act between any two neighboring rows, the (virtual) bond dimension of the updated MPS is at most $q^{O(d^2)}$.

We now simulate the computational basis measurement of \texttt{column 1}. More precisely, we measure the qudits of \texttt{column 1} one by one. We first compute the respective probabilities $p_1, p_2, \dots, p_q$ of the $q$ possible measurement outcomes for the first qudit. This involves contracting the MPS encoding $V_1 \rho_1 V_1^\dagger$. We now use these probabilities to classically sample an outcome $i \in [q]$, and update the MPS to condition on this outcome. That is, if (say) we obtain outcome \texttt{1} for site $i$, we apply the projector $\dyad{1}$ to site $i$ of the state and subsequently renormalize. After doing this for every qudit in the column, we have exactly sampled an output string $\vb{x}_1 \in [q]^{L_1}$ from the marginal distribution on \texttt{column 1}, and are left with an MPS description of the pure, normalized, post-measurement state $\rho_2$  proportional to $\tr_{\texttt{column 1}} \qty(\Pi_1^{\vb{x}} V_1 \rho_1 V_1^\dagger \Pi_1^{\vb{x}})$, where $\Pi_1^{\vb{x}}$ denotes the projection of \texttt{column 1} onto the sampled output string $\vb{x} = \vb{x}_1$. Using standard tensor network algorithms, the time complexity of these steps is $L_1 q^{O(d^2)}$.

We next consider \texttt{column 2}. At this point, we add the qudits and apply the gates that are in the lightcone of \texttt{column 2} but were not applied previously. Denote this isometry by $V_2$. It is straightforward to see that this step respects causality. That is, if some gate $U$ is in the lightcone of \texttt{column 1}, then any gate $W$ that is in the lightcone of \texttt{column 2} but not \texttt{column 1} cannot be required to be applied before $U$, because if it were, then it would be in the lightcone of \texttt{column 1}. Hence, when we apply gates in this step, we never apply a gate that was required to be applied before some gate that was applied in the first step. After this step, we have applied all gates in the lightcone of \texttt{columns (1, 2)}, and we have also projected \texttt{column 1} onto the measurement outcomes we observed.

By simulating the measurements of \texttt{column 2} in a similar way to those of \texttt{column 1}, we sample a string $\vb{x}_2$ from the marginal distribution on \texttt{column 2}, conditioned on the previously observed outcomes from \texttt{column 1}.   Each time an isometry $V_j$ is applied, the bond dimension of the MPS representation of the current state will in general increase by a multiplicative factor. In particular, if we iterate this procedure to simulate the entire lattice,  we will eventually encounter a maximal bond dimension of up to $q^{O(d L_1)}$ and will obtain a sample $\vb{x} = (\vb{x}_1, \vb{x}_2, \dots, \vb{x}_{L_2}) \in [q]^n$ from the true output distribution.

To improve the efficiency at the expense of accuracy, we may compress the MPS in each iteration to one with smaller bond dimension using standard MPS compression algorithms. In particular, in each iteration $j$ before we apply the corresponding isometry $V_j$, we first discard as many of the smallest singular values (i.e.~Schmidt values) associated with each cut of the MPS as possible up to a total \emph{truncation error}  per bond of $\epsilon$, defined as the sum of the squares of the discarded singular values. The bond dimension across any cut is reduced by the number of discarded values. This truncation introduces some error that we quantify below.

If the maximal bond dimension of this truncated version of the simulation algorithm is $D$, the total runtime of the full algorithm to obtain a sample is bounded by (taking $q$ and $d$ to be constants) $O(nD^3)$ using standard MPS compression algorithms. %

We assume that for a specified maximal bond dimension $D$ and truncation error per bond $\epsilon$, if a bond dimension ever exceeds $D$ then the algorithm terminates and outputs a failure flag \textsc{fail}. Hence, the runtime of the algorithm when simulating some circuit $C$ with parameters $\epsilon$ and $D$ is bounded by $O(nD^3)$, and the algorithm has some probability of failure $p_{f,C}$.  We summarize the \texttt{SEBD} algorithm in \Cref{SEBD}.

\begin{algorithm}
\caption{\texttt{SEBD}}\label{SEBD}
\hspace*{\algorithmicindent} \textbf{Input}: circuit instance $C$, truncation error $\epsilon$, bond dimension cutoff $D$ \\
\hspace*{\algorithmicindent} \textbf{Output}: string $\vb{x} \in [q]^n$ or \textsc{fail} \\
\hspace*{\algorithmicindent} \textbf{Runtime}: $O(nD^3)$ [$q$ and $d$ assumed to be constants]
\begin{algorithmic}[1]
\State initialize an MPS in the state $\dyad{1}^{\otimes L_1}$, corresponding to \texttt{column 1}
\For {$t=1\dots L_2$}
    \State compress MPS describing state by truncating small singular values, up to error $\epsilon$ per bond
     \State apply $V_t$, corresponding to gates in the lightcone of column $t$ not yet applied
     \State if some bond dimension is $> D$, terminate and output \textsc{fail}
    \State simulate measurement of all qudits in column $t$ via MPS contraction and sampling
    \State apply $\Pi_t^{\vb{x}_t}$ to condition on measurement string $\vb{x}_t$ observed for that column
\EndFor
\Return $(\vb{x}_1, \dots, \vb{x}_{L_2}) \in [q]^{n}$
\end{algorithmic}
\end{algorithm}
The untruncated version of the algorithm presented above samples from the true distribution $\mathcal{D}_C$ of the measurement outcomes of the original 2D circuit $C$. However, due to the MPS compression which we perform in each iteration and the possibility of failure, the algorithm incurs some error which causes it to instead sample from some distribution $\mathcal{D}_C^\prime$. Here, we bound the total variation distance between these distributions, defined by
\begin{equation}
    \frac{1}{2}\| \mathcal{D}'_C-\mathcal{D}_C \|_1 = \frac{1}{2}\sum_{\vb{x}} \lvert \mathcal{D}'_C(\vb{x})-\mathcal{D}_C(\vb{x}) \rvert + \frac{1}{2}p_{f,C},
\end{equation}
where the sum runs over the $q^n$ possible output strings (not including \textsc{fail}) in terms of the truncation error made by the algorithm.

We first obtain a very general bound on the error made by \texttt{SEBD} with no bond dimension cutoff in terms of the truncation error. Note that the truncation error may  depend on the (random) measurement outcomes, and is itself therefore a random variable.  See \Cref{appendix:error} for a proof.

\begin{restatable}{lemma}{samplingError}\label{lem:sampling_error_bound}
Let $\epsilon_i$ denote the sum of the squares of all singular values discarded in the compression during iteration $i$ of the simulation of a circuit $C$ with output distribution $\mathcal{D}_C$ by \texttt{SEBD} with no bond dimension cutoff, and let $\Lambda$ denote the sum of all singular values discarded over the course of the algorithm. Then the distribution $\mathcal{D}^\prime_C$ sampled from by \texttt{SEBD} satisfies
\begin{equation}\frac{1}{2} \norm{\mathcal{D}'_C - \mathcal{D}_C}_1 \leq \E \sum_{i=1}^{L_2} \sqrt{2 \epsilon_i} \leq   \sqrt{2} \E \Lambda,
\end{equation}
where the expectations are over the random measurement outcomes.
\end{restatable}

From \Cref{lem:sampling_error_bound} we immediately obtain two corollaries. The first is useful for empirically bounding the sampling error in total variation distance made by \texttt{SEBD} when the algorithm also has a bond dimension cutoff. The second is a useful asymptotic statement. The corollaries follow straightforwardly from the coupling formulation of variational distance, Markov's inequality, and the triangle inequality.

\begin{corollary}\label{cor:SEBD_error_bound}
	Let $\mathcal{A}$ denote a \texttt{SEBD} algorithm with truncation error parameter $\epsilon$ and bond dimension cutoff $D$. Consider a fixed circuit $C$, and suppose that $\mathcal{A}$ applied to this circuit fails with probability $p_{f,C}$.  Then $\mathcal{A}$ samples from the output distribution of $C$ with total variation distance error bounded by  $L_2 \sqrt{2 \epsilon L_1} + p_{f,C}$.

	If the failure probability of $\mathcal{A}$ averaged over random choice of circuit instance and measurement outcome is $p_f$, then for any $\delta$, on at least $1-\delta$ fraction of circuit instances, $\mathcal{A}$ samples from the true output distribution with total variation distance error bounded by $L_2 \sqrt{2 \epsilon L_1} + p_f/\delta$.
\end{corollary}

In practice, the variational distance error of \texttt{SEBD} with truncation error $\epsilon$ applied to the simulation of some circuit $C$ can be bounded by constructing a confidence interval for $p_{f,C}$ and applying the above bound.

\begin{corollary}\label{cor:SEBD_asymptotic_efficiency}
	Let $\mathcal{A}$ denote a \texttt{SEBD} algorithm with truncation error parameter $\epsilon$ and no bond dimension cutoff. Suppose that, for some random circuit family with $q=O(1)$ and $d=O(1)$, the expected bond dimension across any cut is bounded by $\poly(n, 1/\epsilon)$. Then, \texttt{SEBD} with some choice of $\epsilon = 1/\poly(n)$ and $D = \poly(n)$ runs in time $\poly(n,1/\varepsilon, 1/\delta)$ and, with probability at least $1-\delta$ over the choice of circuit instance $C$, samples from the output distribution of $C$ with variational distance error less than $\varepsilon$.
\end{corollary}

Thus, to prove the part of \Cref{con:efficient} about sampling up to total variation distance error $\varepsilon$ for uniform random circuit families, it would suffice to show that there is a 2D constant-depth uniform random quantum circuit family with the worst-case-hard property for which the expected bond dimension across any cut while running \texttt{SEBD} with truncation parameter $\epsilon$ is bounded by $\poly(n, 1/\epsilon)$. Later, we will introduce two candidate circuit families for which we can give numerical and analytical evidence that this criterion is indeed met.

In the next subsection, we show how the other part of \Cref{con:efficient}, regarding computing output probabilities, would also follow from a $\poly(n,1/\epsilon)$ bound on the bond dimension of states encountered by \texttt{SEBD} on uniform worst-case-hard circuit families.

\subsubsection{Computing output probabilities}
\label{se:computing_output_probs}
In the previous section, we described how a \texttt{SEBD} algorithm with a truncation error parameter $\epsilon$ and a  bond dimension cutoff $D$ applied to a circuit $C$ samples from a distribution $\mathcal{D}'_C$ satisfying $ \| \mathcal{D}'_C - \mathcal{D}_C \|_1 \leq 2  L_2 \sqrt{2 \epsilon L_1} + 2 p_{f,C}$ where $p_{f,C}$ is the probability that some bond dimension exceeds $D$ and the algorithm terminates and indicates failure. Expanding the expression for the 1-norm and rearranging, we have
\begin{equation}
	\frac{1}{q^n} \sum_{\vb{x}} |\mathcal{D}'_C(\vb{x}) - \mathcal{D}_C(\vb{x}) | \leq \frac{2 L_2 \sqrt{2 \epsilon L_1} + p_{f,C}}{q^n}.
\end{equation}
\texttt{SEBD} with bond dimension cutoff $D$ can be used to compute $\mathcal{D}'_C(\vb{x})$ for any output string $\vb{x}$ in time $O(n D^3)$ (taking $q$ and $d$ to be constants). To do this, for a fixed output string $\vb{x}$, \texttt{SEBD} proceeds similarly to the case in which it's being used for sampling, but rather than sampling from the output distribution of some column, it simply projects that column onto the outcome specified by the string $\vb{x}$, and computes the conditional probability of that outcome via contraction of the MPS.  That is, at iteration $t$, the algorithm computes the conditional probability of measuring the string $\vb{x}_t \in [q]^{L_1}$ in column $t$ , $\mathcal{D}'_C(\vb{x}_t | \vb{x}_1, \dots, \vb{x}_{t-1})$, by projecting column $t$ onto the relevant string via the projector $\Pi^{\vb{x}_t}_t$ and then contracting the relevant MPS. If the bond dimension ever exceeds $D$, then it must hold that $\mathcal{D}'_C(\vb{x}) = 0$, and so the algorithm outputs zero and terminates.  Otherwise, the algorithm outputs $\mathcal{D}'_C(\vb{x}) = \prod_{t=1}^{L_2} \mathcal{D}_C'(\vb{x}_t | \vb{x}_1, \dots, \vb{x}_{t-1})$. We summarize this procedure in \Cref{SEBD_probability}.

\begin{algorithm}[h]
\caption{\texttt{SEBD} for computing output probabilities}\label{SEBD_probability}
\hspace*{\algorithmicindent} \textbf{Input}: circuit instance $C$, truncation error $\epsilon$, bond dimension cutoff $D$, string $\vb{x} \in [q]^n$ \\
\hspace*{\algorithmicindent} \textbf{Output}: $\mathcal{D}'_C(\vb{x})$ \\
\hspace*{\algorithmicindent} \textbf{Runtime}: $O(nD^3)$ [$q$ and $d$ assumed to be constants]
\begin{algorithmic}[1]
\State initialize an MPS in the state $\dyad{1}^{\otimes L_1}$, corresponding to \texttt{column 1}
\For {$t=1\dots L_2$}
    \State compress MPS describing state by truncating small singular values, up to error $\epsilon$ per bond
     \State apply $V_t$, corresponding to gates in the lightcone of column $t$ not yet applied
     \State if some bond dimension is $> D$, terminate and output zero
    \State apply $\Pi_t^{\vb{x}_t}$ to condition on string $\vb{x}_t$
    \State compute $\mathcal{D}'_C (\vb{x}_t|\vb{x}_1, \dots, \vb{x}_{t-1})$ via MPS contraction
\EndFor
\Return $\mathcal{D}'_C(\vb{x}) = \prod_{t=1}^{L_2} \mathcal{D}'_C(\vb{x}_t | \vb{x}_1, \dots, \vb{x}_{t-1})$
\end{algorithmic}
\end{algorithm}

We have therefore shown the following.

\begin{lemma}
Let $p_{f,C}$ be the failure probability of \texttt{SEBD} when used to simulate a circuit instance $C$ with truncation error parameter $\epsilon$ and bond dimension cutoff $D$. Suppose $\vb{x} \in [q]^n$ is an output string drawn uniformly at random. Then \Cref{SEBD_probability} outputs a number $\mathcal{D}'_C(\vb{x})$ satisfying
\begin{equation}\label{eq:errorUniformStrings}
	\E_{\vb{x}} |\mathcal{D}'_C(\vb{x}) - \mathcal{D}_C(\vb{x})| \leq \frac{2 L_2 \sqrt{2 \epsilon L_1} + p_{f,C}}{q^n}.
\end{equation}
\end{lemma}

The above lemma bounds the expected error incurred while estimating a uniformly random output probability for a fixed circuit instance $C$.  We may use this lemma to straightforwardly bound the expected error incurred while estimating the  probability of a fixed output string over a distribution of random circuit instances. The corollary is applicable if the distribution of circuit instances has the property of being invariant under an application of a final layer of arbitrary single-qudit gates. This includes circuits in which all gates are Haar-random (as long as every qudit is acted on by some gate), but is more general. In particular, any circuit distribution in which the final gate to act on any given qudit is Haar-random satisfies this property. This fact will be relevant in subsequent sections.

\begin{corollary}\label{cor:SEBD_probability_bound}
Let $p_{f}$ be the failure probability of \texttt{SEBD} when used to simulate a random circuit instance $C$ with truncation error parameter $\epsilon$ and bond dimension cutoff $D$, where $C$ is drawn from a distribution that is invariant under application of a final layer of arbitrary single-qudit gates. Then for any fixed string $\vb{x} \in [q]^n$ the output of \Cref{SEBD_probability} satisfies
\begin{equation}
	\E_C |\mathcal{D}_C'(\vb{x}) - \mathcal{D}_C(\vb{x}) | \leq \frac{2L_2 \sqrt{2 \epsilon L_1} + p_{f}}{q^n}.
\end{equation}
\end{corollary}
\begin{proof}
Averaging the bound of \Cref{eq:errorUniformStrings} over random circuit instances, we have
\begin{equation}
	\E_{\vb{y}} \E_{C} |\mathcal{D}'_C(\vb{y}) - \mathcal{D}_C(\vb{y})| \leq \frac{2 L_2 \sqrt{2 \epsilon L_1} + p_{f}}{q^n}.
\end{equation}
Let $L_{\vb{y}}$ denote a layer of single-qudit gates with the property that $L_{\vb{y}}\ket{\vb{x}} = \ket{\vb{y}}$. By assumption, $C$ is distributed identically to the composition of $C$ with $L_{\vb{y}}$, denoted $L_{\vb{y}} \circ C$. Together with the observation that $\mathcal{D}_{L_{\vb{y}} \circ C}(\vb{y}) = \mathcal{D}_{C}(\vb{x})$, we have
\begin{equation}
\E_{\vb{y}} \E_{C} |\mathcal{D}'_C(\vb{y}) - \mathcal{D}_C(\vb{y})| = \E_{\vb{y}} \E_{C} |\mathcal{D}'_{L_{\vb{y}} \circ C}(\vb{y}) - \mathcal{D}_{L_{\vb{y}} \circ C}(\vb{y})| =  \E_{C} |\mathcal{D}'_C(\vb{x}) - \mathcal{D}_C(\vb{x})|,
\end{equation}
from which the result follows.
\end{proof}
The following asymptotic statement follows straightforwardly.
\begin{corollary}\label{cor:SEBD_postselected_asymptotic_bound}
Let $\mathcal{A}$ denote a \texttt{SEBD} algorithm with truncation error parameter $\epsilon$ and no bond dimension cutoff. Suppose that, for some random circuit family with $q=O(1)$ and $d=O(1)$, the expected bond dimension across any cut is bounded by $\poly(n, 1/\epsilon)$. Then, \texttt{SEBD} with some choice of $\epsilon = 1/\poly(n)$ and $D = \poly(n)$ runs in time $\poly(n,1/\varepsilon, 1/\delta)$ and, with probability at least $1-\delta$ over the choice of circuit instance $C$, estimates $\mathcal{D}_C(\vb{x})$ for some fixed $\vb{x} \in [q]^n$ up to additive error bounded by $\varepsilon/q^n$.
\end{corollary}

\Cref{cor:SEBD_postselected_asymptotic_bound} shows how the part of \Cref{con:efficient} about computing arbitrary output probabilities to error $\varepsilon/q^n$ would follow from a bound on the bond dimension across any cut when \texttt{SEBD} runs on a uniform worst-case-hard circuit family.

\subsection{\texttt{SEBD} applied to cluster state with Haar-random measurements (\texttt{CHR})}\label{se:CHR}
We now study the \texttt{SEBD} algorithm in more detail for a simple uniform family of 2D random circuits that possesses the worst-case-hard property required by \Cref{con:efficient}. The model we consider is the following: start with a 2D cluster state of $n$ qubits arranged in a $\sqrt{n}\times \sqrt{n}$ grid, apply a single-qubit Haar-random gate to each qubit, and then measure all qubits in the computational basis. Recall that a cluster state may be created by starting with the product state $\ket{+}^{\otimes n}$ before applying CZ gates between all adjacent sites. An equivalent formulation which we will find convenient in the subsequent section is to measure each qubit of the cluster state in a Haar-random basis. We refer to this model as \texttt{CHR}, for ``cluster state with Haar-random measurements''.

 It is straightforward to show, following \cite{bremner2010classical}, that sampling from the output distribution of \texttt{CHR} is classically \emph{worst-case hard} assuming the polynomial hierarchy (\textsf{PH}) does not collapse to the third level.
\begin{lemma}\label{HardnessOfCluster}
Suppose that there exists a polynomial-time classical algorithm for \texttt{CHR} that, for any circuit realization, samples from the classical output distribution. Then \textsf{PH} collapses to the third level.
\end{lemma}
\begin{proof}
Recall that the cluster state is a resource state for universal measurement-based quantum computation (MBQC) \cite{raussendorf2001one}. Hence, it is \textsf{PostBQP}-hard to sample from the conditional output distribution of an arbitrary instance of \texttt{CHR} with some outcomes postselected on zero. Therefore, if there is some efficient classical sampling algorithm, then \textsf{PostBPP} = \textsf{PostBQP} which implies that \textsf{PH} collapses to the third level \cite{bremner2010classical}.
\end{proof}

It can also be shown, following \cite{movassagh2019}, that near-exactly computing output probabilities of \texttt{CHR} is \SharpP-hard in the average case.

\begin{lemma}[Follows from \cite{movassagh2019}]\label{lem:ramis}
Suppose there exists an algorithm $\mathcal{A}$ that, given a random instance $C$ of \texttt{CHR} and fixed string $\vb{x}$, with probability $1-1/p(n)$ outputs $\mathcal{D}_C(\vb{x}) = |\langle \vb{x} | C | 0\rangle ^{\otimes n}|^2$ up to additive error $2^{-\tilde{\Theta}(n^2)}$, where $p(n)$ is a sufficiently large polynomial. Then  $\mathcal{A}$ can be used to compute a \SharpP-complete function with high probability in polynomial time.
\end{lemma}

Under standard complexity theoretic assumptions, \Cref{HardnessOfCluster} rules out the existence of a classical sampling algorithm for \texttt{CHR} that succeeds for all instances, and  \Cref{lem:ramis} rules out the existence of an algorithm for efficiently computing most output probabilities of \texttt{CHR}. A natural question is then whether efficient approximate average-case versions of these algorithms   may exist. We formalize these questions as the problems $\texttt{CHR}_{\pm}^{\texttt{samp/prob}}$.

\begin{problem}[$\texttt{CHR}_{\pm}^{\texttt{samp/prob}}$]
Given as input a random instance $C$ of \texttt{CHR} (specified by a sidelength $\sqrt{n}$ and a set of $n$ single-qubit Haar-random gates applied to the $\sqrt{n} \times \sqrt{n}$ cluster state) and error parameters $\varepsilon$ and $\delta$, perform the following computational task in time $\poly(n,1/\varepsilon,1/\delta)$.
\begin{itemize}
\item $\texttt{CHR}_{\pm}^{\texttt{samp}}$. Sample from a distribution $\mathcal{D}'_C$ that is $\varepsilon$-close in total variation distance to the true output distribution $\mathcal{D}_C$ of circuit $C$, with probability of success at least $1-\delta$ over the choice of measurement bases. \\
\item $\texttt{CHR}_{\pm}^{\texttt{prob}}$. Estimate $\mathcal{D}_C(\vb{0})$, the probability of obtaining the all-zeros string upon measuring the output state of $C$ in the computational basis, up to additive error at most $\varepsilon/2^n$, with probability of success at least $1-\delta$ over the choice of measurement bases.
\end{itemize}
\end{problem}

In the next section, we show that \texttt{SEBD} solves $\texttt{CHR}_{\pm}^{\texttt{samp/prob}}$ if a certain form of 1D dynamics involving local unitary gates and measurements is classically simulable.

\subsubsection{\texttt{SEBD} applied to \texttt{CHR}}
We first consider the sampling variant of \texttt{SEBD}.  Specializing to the \texttt{CHR} model, the algorithm takes on a particularly simple form due to the fact that the cluster state is built by applying CZ gates between all neighboring pairs of qubits, which are initialized in $\ket{+}$ states. Due to this structure, the radius of the lightcone for this model is simply one. In particular, the only gates in the lightcone of \texttt{columns 1-j} are the Haar-random single-qubit gates acting on qubits in these columns, as well as CZ gates that act on at least one qubit within these columns. This permits a simple prescription for \texttt{SEBD} applied to this problem.

Initialize the simulation algorithm in the state $\rho_1 = \dyad{+}^{\otimes \sqrt{n}}$ corresponding to \texttt{column 1}. To implement the isometry $V_1$, initialize the qubits of \texttt{column 2} in the state $\dyad{+}^{\otimes \sqrt{n}}$ and apply CZ gates between adjacent qubits that are both in \texttt{column 1} and between adjacent qubits in separate columns. Now, measure the qubits of \texttt{column 1} in the specified Haar-random bases (equivalently, apply the specified Haar-random gates and measure in the computational basis), inducing a pure state $\rho_2$ with support in \texttt{column 2}. Iterating this process, we progress through a random sequence of 1D states on $\sqrt{n}$ qubits $\rho_1 \rightarrow \rho_2 \rightarrow \dots \rightarrow \rho_{\sqrt{n}}$ which we will see can be equivalently understood as arising from a 1D dynamical process consisting of alternating layers of random unitary gates and weak measurements.

It will be helpful to introduce notation. Define $\ket{\theta, \phi} := \cos\qty(\frac{\theta}{2})\ket{0} + e^{i\phi}\sin\qty(\frac{\theta}{2})\ket{1}$. In other words, let $\ket{\theta, \phi}$ denote the single-qubit pure state with polar angle $\theta$ and azimuthal angle $\phi$ on the Bloch sphere. Let $\theta^{(t)}_i$ and $\phi^{(t)}_i$ specify the measurement basis of the qubit in row $i$ and column $t$; that is, the projective measurement on the qubit in row $i$ and column $t$ is $\{\Pi_{\theta^{(t)}_i,\phi^{(t)}_i}^0, \Pi_{\theta^{(t)}_i,\phi^{(t)}_i}^1\}$ with $\Pi_{\theta^{(t)}_i, \phi^{(t)}_i}^0 := \dyad{\theta^{(t)}_i, \phi^{(t)}_i}$ and   $\Pi_{\theta^{(t)}_i, \phi^{(t)}_i}^1 := I- \Pi_{\theta^{(t)}_i, \phi^{(t)}_i}^0$. We also define
\begin{subequations}\label{eq:M0M1}
  \begin{align}
M_0(\theta, \phi) &:= \mqty(\cos(\theta/2) & 0 \\ 0 & e^{-i \phi} \sin(\theta/2))  \\
M_1(\theta, \phi) &:= \mqty(\sin(\theta/2) & 0 \\ 0 & e^{i \phi} \cos(\theta/2)).
  \end{align}
  \end{subequations}
Note that $\{M_0(\theta,\phi), M_1 (\theta, \phi)\}$ defines a weak single-qubit measurement. We now describe, in \Cref{effective}, a 1D  process which we claim produces a sequence of states identical to that encountered by \texttt{SEBD} for the same choice of measurement bases and measurement outcomes, and also has the same measurement statistics.

\begin{algorithm}
\caption{Effective 1D dynamics of a fixed instance of \texttt{CHR}}\label{effective}
\begin{algorithmic}[1]
\State $\varphi_1 \gets \dyad{+}^{\otimes \sqrt{n}}$.
\For {$t =1\dots \sqrt{n}-1$}
	\State apply a CZ gate between every adjacent pair of qubits
	\State measure $\{ M_0(\theta^{(t)}_i,\phi^{(t)}_i), M_1(\theta^{(t)}_i,\phi^{(t)}_i) \}$ on qubit $i$, obtaining $X^{(t)}_i$, for $i \in [\sqrt{n}]$
	\State apply a Hadamard transform
	\State $\varphi_{t+1} \gets $ resulting state
\EndFor
\State measure $\qty{\Pi_{\theta^{(\sqrt{n})}_i,\phi^{(\sqrt{n})}_i}^0, \Pi_{\theta^{(\sqrt{n})}_i,\phi^{(\sqrt{n})}_i}^1 }$ on qubit $i$,  obtaining $X_i^{(\sqrt{n})}$, for $i \in [\sqrt{n}]$
\end{algorithmic}
\end{algorithm}

\begin{lemma}
For a fixed choice of $\{\theta^{(t)}_i, \phi^{(t)}_i\}$ parameters, the joint distribution of outcomes $\{ X_i^{(t)} \}_{i,t}$ is identical to that of $\{ Y_i^{(t)} \}_{i,t}$, where $\{Y_{i}^{(t)}\}_{i,t}$ are the measurement outcomes obtained upon measuring all qubits of a $\sqrt{n} \times \sqrt{n}$ cluster state, with the measurement on the qubit in row $i$ and column $t$ being $\{\Pi_{\theta^{(t)}_i,\phi^{(t)}_i}^0, \Pi_{\theta^{(t)}_i,\phi^{(t)}_i}^1\}$. Furthermore, for any fixed choice of measurement outcomes, $\varphi_j = \rho_j$ for all $j \in [\sqrt{n}]$, where $\rho_j$ is the state at the beginning of iteration $j$ of the \texttt{SEBD} algorithm.
\end{lemma}
\begin{proof}
The lemma follows from the above description of the behavior of \texttt{SEBD} applied to \texttt{CHR}, as well as the following identities holding for any single-qubit state $\ket{\xi}$ which may be verified by straightforward calculation:
\begin{align}
(\Pi_{\theta,\phi}^0 \otimes I) CZ (\ket{\xi}\otimes \ket{+}) &= \ket{\theta, \phi} \otimes H M_0(\theta,\phi) \ket{\xi} \\
(\Pi_{\theta,\phi}^1 \otimes I) CZ (\ket{\xi}\otimes \ket{+}) &=\ket{\pi - \theta, -\phi} \otimes H  M_1(\theta,\phi) \ket{\xi}.
\end{align} \end{proof}

\sloppy We have seen that, for a fixed choice of single-qubit measurement bases $\{\theta^{(t)}_j, \phi^{(t)}_j \}_{t,j}$ associated with an instance $C$, we can define an associated 1D process consisting of alternating layers of single-qubit weak measurements and local unitary gates, such that simulating this 1D process is sufficient for sampling from $\mathcal{D}_C$.

Now, recall that in the context of simulating \texttt{CHR}, each single-qubit measurement basis is chosen randomly according to the Haar measure. That is, the Bloch sphere angles $(\theta^{(t)}_i, \phi^{(t)}_i)$ are Haar-distributed. If we define $x^{(t)}_i \equiv \cos \theta^{(t)}_i$, we find that $x^{(t)}_i$ is uniformly distributed on the interval $[-1,1]$. The parameters $\phi^{(t)}_i$ are uniformly distributed on $[0,2\pi]$. Using these observations, as well as the observation that the outcome probabilities of the measurement of qubit $i$ in iteration $t$ are independent of the azimuthal angle $\phi^{(t)}_i$ when $t < \sqrt{n}$, we may derive effective dynamics of a random instance.

Define the operators
$$ N(x) := \mqty(\sqrt{\frac{1+x}{2}} & 0 \\ 0 & \sqrt{\frac{1-x}{2}}), \; \; x \in [-1,1] .$$
Note that $\{ N(x), N(-x) \}$ defines a weak measurement. Also, define the phase gate
$$ P(\phi) := \mqty(1 & 0 \\ 0 & e^{i\phi}), \; \; \phi \in [0,2\pi]. $$

By randomizing each single-qubit measurement basis according to the Haar distribution, one finds that the dynamics of \Cref{effective} (which applies for a fixed choice of measurement bases) may be written as \Cref{RandomEffective} below, where the notation $x\in_U [-1,1]$ means that $x$ is a random variable uniformly distributed on $[-1,1]$. That is, the distribution of random sequences $\varphi_1 \rightarrow \varphi_2 \rightarrow \dots \rightarrow \varphi_{\sqrt{n}}$ and distribution of output statistics  produced by \Cref{RandomEffective} is identical to that produced by \texttt{SEBD} applied to \texttt{CHR}.

\begin{algorithm}
\caption{Effective 1D dynamics of \texttt{CHR}}\label{RandomEffective}
\begin{algorithmic}[1]
\State $\varphi_1 \gets \dyad{+}^{\otimes \sqrt{n}}$.
\For {$t =1\dots \sqrt{n}-1$}
	\State apply a CZ gate between every adjacent pair of qubits
	\For {$i=1 \dots \sqrt{n}$}
		\State measure $\{ N(x), N(-x) \}$ on qubit $i$ with $x \in_U [-1,1]$
		\State apply the gate $P(\phi)$ with $\phi\in_U [0,2\pi]$ to qubit $i$
	\EndFor
	\State apply a Hadamard transform
	\State $\varphi_{t+1} \gets $ resulting state
\EndFor
	\State perform a projective measurement on each qubit in a Haar-random basis
\end{algorithmic}
\end{algorithm}

Hence, if \texttt{TEBD} can efficiently simulate the process of \Cref{RandomEffective} with high probability, then \texttt{SEBD} can solve $\texttt{CHR}_{\pm}^{\texttt{samp}}$ and $\texttt{CHR}_{\pm}^{\texttt{prob}}$. We formalize this in the following lemma.

\begin{lemma}
Suppose that \texttt{TEBD} can efficiently simulate the process described in  \Cref{RandomEffective} in the sense that the expected bond dimension across any cut is bounded by $\poly(n,1/\epsilon)$ where $\epsilon$ is the truncation error parameter. Then \texttt{SEBD} can be used to solve $\texttt{CHR}_{\pm}^{\texttt{samp}}$ and $\texttt{CHR}_{\pm}^{\texttt{prob}}$.
\end{lemma}
\begin{proof}
Follows from \Cref{cor:SEBD_asymptotic_efficiency}, \Cref{cor:SEBD_postselected_asymptotic_bound}, and the equivalence to \Cref{RandomEffective} discussed above.
\end{proof}

We have shown how \texttt{SEBD} applied to \texttt{CHR} can be reinterpreted as \texttt{TEBD} applied to a 1D dynamical process involving alternating layers of random unitaries and weak measurements. Up until this point, there has been little reason to expect that \texttt{SEBD} is efficient for the simulation of \texttt{CHR}. In particular, with no truncation, the bond dimension of the MPS stored by the algorithm grows exponentially as the algorithm sweeps across the lattice.

We now invoke the findings of a number of related recent works \cite{li2018quantum, chan2018weak, skinner2019measurement, li2019measurement, szyniszewski2019entanglement, choi2019quantum, gullans2019dynamical,bao2019theory,jian2019measurement,gullans2019scalable, zabalo2019critical}  to motivate the possibility that \texttt{TEBD} can efficiently simulate the effective 1D dynamics. These works study various 1D dynamical processes involving alternating layers of measurements and random local unitaries. In some cases, the measurements are considered to be projective and only occur with some probability $p$. In other cases, similarly to \Cref{RandomEffective}, weak measurements are applied to each site with probability one. The common finding of these papers is that such models appear to exhibit an entanglement phase transition driven by measurement probability $p$ (in the former case), or measurement strength (in the latter case). On one side of the transition, the entanglement entropy obeys an area law, scaling as $O(1)$ with the length $L$. On the other side, it obeys a volume law, scaling as $O(L)$.

Based on these works, one expects the entanglement dynamics to saturate to an area-law or volume-law phase. And in fact, our numerical  studies (presented in \Cref{se:numerics})  suggest that these dynamics saturate to an area-law phase.   The common intuition that 1D quantum systems obeying an area law for the von Neumann entropy are easy to simulate with matrix product states therefore suggests that \texttt{SEBD} applied to this problem is efficient. While counterexamples to this common intuition are known \cite{schuch2008entropy}, they are contrived and do not present an obvious obstruction for our algorithm. To better understand the relationship between maximal bond dimension and truncation error when the effective dynamics is in the area-law phase as well as rule out such counterexamples, in the following section we describe a toy model for a unitary-and-measurement process in the area-law phase, which predicts a superpolynomial decay of Schmidt values across any cut and therefore predicts that a polynomial runtime is sufficient to perform the simulation to $1/\poly(n)$ error. Our numerical results (presented in  \Cref{se:numerics}) suggest that the effective dynamics of the random circuit architectures we consider are indeed in the area-law phase, with entanglement spectra consistent with those predicted by the toy model dynamics. Further analytical evidence for efficiency is given in \Cref{se:statmech}.

Note that, although we explicitly derived the effective 1D dynamics for the \texttt{CHR} model and observed it to be a simple unitary-and-measurement process, the interpretation of the effective 1D dynamics as a unitary-and-measurement process is not specific to \texttt{CHR} and is in fact general. In the general case, \texttt{SEBD} tracks $O(r)$ columns simultaneously where $r$ is the radius of the lightcone corresponding to the circuit. In each iteration, new qudits that have come into the lightcone are added, unitary gates that have come into the lightcone are performed, and finally projective measurements are performed on a single column of qudits. Similarly to the case of \texttt{CHR}, this entire procedure can be viewed as an application of unitary gates followed by weak measurements on a 1D chain of qudits of dimension $q^{O(r)}$. Intuitively, increasing the circuit depth corresponds both to increasing the local dimension in the effective 1D dynamics and decreasing the measurement strength. The former is due to the fact that in general the lightcone radius $r$ will increase as depth is increased, and the local dimension of the effective dynamics is $q^{O(r)}$. The latter is due to the fact that as $r$ increases, the number of tracked columns increases but the number of measured qudits in a single round stays constant. Hence the fraction of measured qudits decreases, and intuitively we expect this to correspond to a decrease in effective measurement strength. This intuition together with the findings of prior works on unitary-and-measurement dynamics suggests that the effective dynamics experiences an entanglement phase transition from an area-law to volume-law phase as $q$ or $d$ is increased, and therefore \texttt{SEBD} experiences a computational phase transition, supporting \Cref{con:transition}. While this analogy is not perfect, we provide further analytical evidence in \Cref{se:statmech2D} that the effective 1D dynamics indeed undergoes such a phase transition.

\subsection{Conjectured entanglement spectrum of unitary-and-measurement dynamics in an area-law phase}\label{sec:conjectured}

Numerical (\Cref{se:numerics}) and analytical (\Cref{se:statmech2D}) evidence suggests that the effective 1D dynamics corresponding to the uniform 2D shallow random circuit families we consider are in the area-law phase, making efficient simulation via \texttt{SEBD} very plausible. However, it is desirable to have clear predictions for the scaling of the entanglement spectra for states of the effective 1D dynamics, as this allows us to make concrete predictions for error scaling of \texttt{SEBD} and rule out (contrived) examples of states \cite{schuch2008entropy} which cannot be efficiently represented via MPS despite obeying an area law for the von Neumann entanglement entropy.

To this end, we study a simple toy model of how entanglement might scale in the area-law phase of a unitary-and-measurement circuit.  Consider a chain of $n$ qubits where we are interested in the entanglement across the cut between $1,\ldots,n/2$ and $n/2+1,\ldots,n$ (assume $n$ is even).  We model the dynamics as follows.  In each time step we perform the following three steps:
\begin{enumerate}
    \item   Set the state of sites $n/2$ and $n/2+1$ to be an EPR pair $\ket{\Phi}=(\ket{00}+\ket{11})/\sqrt{2}$.
    \item Perform the cyclic permutations $n/2,n/2-1,\ldots,1,n/2$ and $n/2+1,n/2+2,\ldots,n,n/2+1$.  That is, move each qubit one step away from the central cut, except for qubits 1 and $n$, which are moved to $n/2$ and $n/2+1$ respectively.
    \item Perform a weak measurement on each qubit with Kraus elements $M_0(\theta)=\cos(\theta/2)\dyad 0 + \sin (\theta/2)\dyad 1$ and
    $M_1(\theta)=\sin(\theta/2)\dyad 0 + \cos (\theta/2)\dyad 1$.  This is based on \Cref{eq:M0M1}, but the phases will not matter here so we have dropped them for simplicity.
\end{enumerate}

Without the measurements this would create one EPR pair in each time step until the system had $n/2$ EPR pairs across the cut after time $n/2$.  However, the measurements have the effect of reducing the entanglement. For this process, we derive the functional form of the asymptotic scaling of half-chain Schmidt coefficients $\lambda_1 \geq \lambda_2 \geq \cdots$. Moreover, bounds on the scaling of the entanglement spectrum allows us to derive a relation between the truncation error (sum of squares of discarded Schmidt values) $\epsilon$ incurred upon discarding small Schmidt values, and the rank $r$ of the post-truncation state. The bounds are given in the following lemma, which is proved in \Cref{appendix:error}.
\begin{restatable}{lemma}{toyModel}\label{lem:toy_model}
Let $\lambda_1 \geq \lambda_2 \geq \cdots$ denote the half-chain Schmidt values after at least $n/2$ iterations of the toy model process. Then with probability at least $1-\delta$ the half-chain Schmidt values indexed by $i \geq i^* = \exp(\Theta(\sqrt{\log(n/\delta)}))$ obey the asymptotic scaling
\be \lambda_i \propto \exp(-\Theta(\log^2(i))). \ee
Furthermore, upon truncating the smallest Schmidt coefficients up to a truncation error of $\epsilon$, with probability at least $1-\delta$, the half-chain Schmidt rank $r$ of the post-truncation state obeys the scaling
\be r \leq \exp(\Theta\qty(\sqrt{\log(n/\epsilon \delta)})). \ee
\end{restatable}

This is the basis for our \Cref{con:aggressive}. More precisely, we take this analysis as evidence that the bond dimension $D$, truncation error $\epsilon$, and system size $n$ obey the scaling $D \leq  \exp(\Theta\qty(\sqrt{\log(n/\epsilon \delta)}))$ with probability $1-\delta$ over random circuit instance and random measurement outcomes when \texttt{SEBD} simulates a random constant-depth 2D circuit whose effective 1D dynamics lie in the area-law phase.  Recalling that the runtime of \texttt{SEBD} scales like $O(n D^3)$ for a maximal bond dimension of $D$ and using the relationship between truncation error, failure probability, variational distance error, and simulable circuit fraction given in \Cref{cor:SEBD_error_bound}, we conclude that \texttt{SEBD} with a maximal bond dimension cutoff scaling as  $\exp(\Theta\qty(\sqrt{\log(n/\epsilon \delta)}))$  runs in time $n^{1+o(1)} \exp(\Theta\qty(\sqrt{\log(1/\varepsilon \delta)}))$ and simulates $1-\delta$ fraction of random circuit instances up to variational distance error $\varepsilon$.

It is important to note what this heuristic argument leaves out.  While a 1D unitary-and-measurement circuit will indeed create $O(1)$ ebits across any given cut in each round, these will not remain in the form of distinct pairs of qubits.  The unitary dynamics {\em within} each side of the cut will have the effect of transforming the Schmidt bases into entangled ones.  This will make the measurements less effective at reducing the entanglement, for reasons that can be understood in terms of quantum state merging~\cite{HOW07,choi2019quantum}.  Another simplification of the toy model is that the measurement angle $\theta$ is taken to be a fixed constant rather than random. Finally, in the toy model we assume for simplicity that the EPR pairs move cyclically. We expect that, if this effect is significant, it is more likely to make the toy model overly pessimistic compared with the real situation. Despite these simplifications, we believe this model is qualitatively accurate in the area-law phase. Indeed, the scaling of Schmidt values predicted by our toy model analysis is consistent with the scaling we find numerically in \Cref{fig:spectrum}.

\subsection{\texttt{Patching}}\label{sec:patching}
We now describe a second algorithm for sampling from the output distributions and computing output probabilities of 2D quantum circuits acting on qudits of local dimension $q$. While the \texttt{SEBD} algorithm described in the previous section is efficient if the corresponding effective 1D dynamics can be efficiently simulated with \texttt{TEBD}, the algorithm of this section is efficient if the circuit depth $d$ and local dimension $q$ are constant and the conditional mutual information (CMI) of the classical output distribution is exponentially decaying in a sense that we make precise below. In \Cref{se:statmech2D} we will give evidence that the output distribution of sufficiently shallow random 2D circuits acting on qudits of sufficiently small dimension satisfies such a property with high probability, and the property is not satisfied if the circuit depth or local dimension exceeds some critical constant value.

The algorithm we describe is an adaptation and simplification of the Gibbs state preparation algorithm of \cite{brandao2019finite}. In that paper, the authors essentially showed that a quantum Gibbs state defined on a lattice can be prepared by a quasipolynomial time quantum algorithm, if the Gibbs state satisfies two properties: (1) exponential decay of correlations and (2) exponentially decaying quantum conditional mutual information for shielded regions. Our situation is simpler than the one considered in that paper, due to the fact that sufficiently separated regions of the lattice are causally disconnected as a result of the fact that the circuit inducing the distribution is constant-depth and therefore has a constant-radius lightcone. The structure of our algorithm is very similar to theirs, except we can make some simplifications and substantial improvements as a result of the constant-radius lightcone and the fact that we are sampling from a classical distribution rather than a quantum Gibbs state.

Before we describe the algorithm, we set some notation. Let $\Lambda$ denote the set of all qudits of a $L_1 \times L_2$ rectangular grid (assume $L_1 \leq L_2 \leq \poly(L_1)$).  If $A$ and $B$ are two subsets of qudits of $\Lambda$, we define $\dist(A,B) := \min_{i \in A, j\in B} \dist(i,j)$, where $\dist(i,j)$ is the distance between sites $i$ and $j$ as measured by the $\infty$-norm. There are two primary facts that our algorithm relies on. First, if the circuit has depth $d$, any two sets of qudits separated by a distance greater than $2d$ have non-overlapping lightcones. Hence, if $A$ and $B$ are two lattice regions separated by distance at least $2d$, and $\rho$ is the quantum state output by the circuit (before measurement), it holds that $\rho_{AB} = \rho_A \otimes \rho_B$ and therefore $\mathcal{D}_{AB} = \mathcal{D}_A \otimes \mathcal{D}_B$ if $\mathcal{D} = \sum_{\vb{x}} \mathcal{D}(\vb{x}) \dyad{\vb{x}}$ is the classical output distribution of the circuit and (for example) $\mathcal{D}_{A}$ denotes the marginal of $\mathcal{D}$ on subregion $A$. (Note that our notation is slightly different in this section -- we now use subscripts on $\mathcal{D}$ to denote marginals, and the dependence of $\mathcal{D}$ on the circuit instance is left implicit.) Second, if the classical CMI $I(X:Z|Y)_p$ of three random variables with joint distribution $p_{XYZ}$ is small, then $p_{XYZ}$ is close to the distribution $p_{X|Y}p_Y p_{Z|Y}$ corresponding to a Markov chain $X - Y - Z$. We state this more formally as the following lemma, which follows from the Pinsker inequality.

\begin{lemma}[see e.g. \cite{cover1991elements}]\label{cmi}
Let $X,Y,Z$ be discrete random variables, and let $p_{XYZ}$ denote their joint distribution. Then
$$ I(X:Z|Y)_p \geq \frac{1}{2\ln 2} \| p_{XYZ} - p_{X|Y}p_Y p_{Z|Y} \|_1^2. $$
\end{lemma}

Following \cite{brandao2019finite}, we also formally define a notion of CMI decay.

\begin{definition}[Markov property]\label{markovProperty}
Let $p$ denote a probability distribution supported on $\Lambda$. Then $p$ is said to satisfy the \emph{$\delta(l)$-Markov condition} if, for any tripartition of a subregion $X$ of the lattice into subregions $X=A\cup B\cup C$ such that $\dist(A,C)\geq l$, we have
\begin{equation}
    I(A:C|B)_p \leq  \delta(l).
\end{equation}
\end{definition}

\begin{figure}
    \centering
    \def\svgwidth{0.6\columnwidth}
    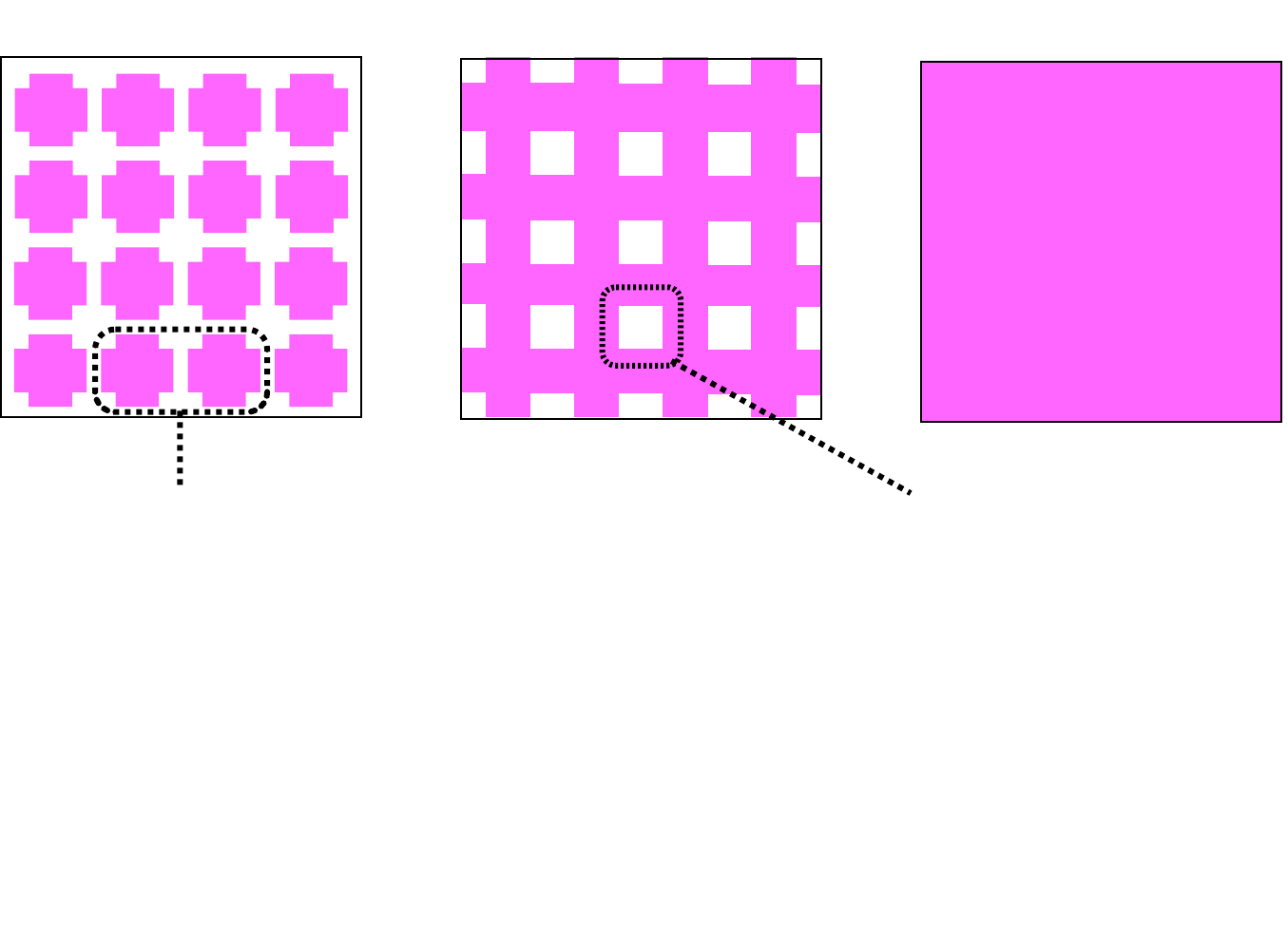
    \caption{\texttt{Patching}. Pink represents marginals of the output distribution that have been approximately sampled, while white represents unsampled regions. In $(a)$, the algorithm has sampled from disconnected patches. Figure $(b)$ depicts how the algorithm transitions from configuration $(a)$ to $(c)$.  Namely, the algorithm generates a sample from the conditional distribution on $A$, conditioned on the configuration of region $B$. Similarly, figure $(d)$ depicts how the ``holes'' of configuration $(c)$ are filled in. The end result is shown in $(e)$, an approximate sample from the global distribution on the full lattice.}\label{fg:patching}
\end{figure}

Intuitively, our algorithm works by first sampling from the marginal distributions of spatially separated patches on the lattice, and then stitching the patches together to approximately obtain a sample from the global distribution. For a $O(1)$-depth circuit whose output distribution has exponentially decaying CMI, the efficiency of this procedure is guaranteed by the two facts above. We now show this more formally.
\begin{theorem}
Suppose $C$ is a 2-local quantum circuit of depth $d$ defined on a 2D rectangular grid $\Lambda$ of $n = L_1 \times L_2$ qudits, and let $\mathcal{D}(\vb{x}) := |\bra{\vb{x}}C\ket{1}^{\otimes n}|^2$ denote its output distribution. Then if $\mathcal{D}$ satisfies the $\delta(l)$-Markov condition, for any integer $l > 2d$ \texttt{Patching} with a length-scale parameter $l$  runs in time $n q^{O(d l)}$ and samples from some distribution $\mathcal{D}'$ that satisfies $\| \mathcal{D}' - \mathcal{D} \|_1 \leq O(1) (n / l^2)\sqrt{ \delta(l)} $.

In particular, if $d = O(1)$, $q=O(1)$, and $\mathcal{D}$ is $\poly(n)e^{-\Omega(l)}$-Markov, then for any polynomial $r(n)$, for some choice of lengthscale parameter \texttt{Patching} runs in time $\poly(n)$ and samples from a distribution that is $1/r(n)$-close to $\mathcal{D}$ in total variation distance.
\end{theorem}
\begin{proof}
The algorithm proceeds in three steps, illustrated in \Cref{fg:patching}. First, for each square subregion $R_i$ shaded in \Cref{fg:patching}(a) with $i \in [O(n/l^2)]$, sample from $\mathcal{D}_{R_i}$, the marginal distribution of $\mathcal{D}$ on subregion $R_i$. To do this, first restrict to the qudits and gates in the lightcone of $R_i$. Sampling from the output distribution on $R_i$ produced by this restricted version of the circuit is equivalent to sampling from the marginal on $R_i$ of the true distribution produced by the full circuit.  Since $l > 2d$, this restriction of the circuit is contained in a sublattice of dimensions $O(l)\times O(l)$. Using standard tensor network methods \cite{markov2008simulating}, sampling from the output distribution of this restricted circuit on $R_i$ can be performed in time $q^{O(d l)}$. Since there are $O(n/l^2)$ patches, this step can be performed in time $n q^{O( d l )}$. After performing this step, we have prepared the state $\mathcal{D}_{R_1}\otimes \cdots \otimes \mathcal{D}_{R_k} = \mathcal{D}_{R_1, \dots, R_k}$ where the equality holds because the patches are separated by $l > 2d$ and are therefore mutually independent.

In the second step, we apply ``recovery maps'' to approximately prepare a sample from the larger, connected lattice subregion $S$ shaded in \Cref{fg:patching}(c). The prescription for these recovery maps is given in \Cref{fg:patching}(b). Referring to this figure, a recovery map $\mathcal{R}_{B\rightarrow AB}$ is applied  to generate a sample from subregion $A$, conditioned on the state of region $B$. Explicitly, the mapping is given by linearly extending the map $\mathcal{R}_{B\rightarrow AB}(\dyad{b}_B) = \sum_a \mathcal{D}_{A|B}(a|b) \dyad{a}_A \otimes \dyad{b}_B$. Note that, for a tripartite distribution $\mathcal{D}_{ABC}$, $\mathcal{R}_{B\rightarrow AB}(\mathcal{D}_{BC}) = \mathcal{D}_{A|B}\mathcal{D}_B \mathcal{D}_{C|B}$. To implement this recovery map, one can again restrict to gates in the lightcone of region $AB$ and utilize standard tensor network simulation algorithms to generate a sample from the marginal distribution on $A$, conditioned on the (previously sampled) state of $B$. The time complexity for this step is again $q^{O(d l)}$. After applying this and $O(n/l^2)$ similar recovery maps, we obtain a sample from a distribution $\mathcal{D}'_S$. By \Cref{cmi}, the triangle inequality, and \Cref{markovProperty}, the error of this step is bounded as
\begin{equation}
    \| \mathcal{D}'_S - \mathcal{D}_S \|_1 \leq O(1) (n/l^2) \sqrt{\delta(l)} = O(1) (n / l^2)\sqrt{\delta(l)}.
\end{equation}

Note that the fact that the errors caused by recovery maps acting on disjoint regions accumulate at most linearly has been referred to previously \cite{brandao2019finite} as the ``union property'' for recovery maps.  The final step is very similar to the previous step. We now apply recovery maps, described by \Cref{fg:patching}(d), to fill in the ``holes'' of the subregion $S$ and approximately obtain a sample from the full distribution $\mathcal{D} = \mathcal{D}_\Lambda$. By a similar analysis, we find that the error incurred in this step is again $O(1) (n / l^2) \sqrt{\delta(l)}$, and therefore the  procedure samples from a  distribution $\mathcal{D}'_\Lambda$ for which $\| \mathcal{D}'_\Lambda - \mathcal{D}_\Lambda \|_1 \leq O(1) (n / l^2)\sqrt{\delta(l)} $.

The second paragraph of the theorem follows immediately by choosing a suitable $l = \Theta(\log n)$.
\end{proof}

A straightforward application of Markov's inequality implies that a polynomial-time algorithm for sampling with error $1/\poly(n)$ succeeds with high probability over a random circuit instance if the output distribution CMI is exponentially decaying in expectation. We formalize this as the following corollary.

\begin{corollary}
Let $\mathcal{C}$ be a random circuit distribution. Define $\mathcal{C}$ to be $\delta(l)$-Markov if, for any tripartition of a subregion $X$ of the lattice into subregions $X = A \cup B \cup C$ such that $\text{dist}(A,C) \geq l$, we have
\begin{equation}
\langle I(A:C|B)_{\mathcal{D}} \rangle \leq \delta(l)
\end{equation}
where the angle brackets denote an average over circuit realizations and $\mathcal{D}$ is the associated classical output distribution. Then if $d = O(1), q=O(1)$, and $\mathcal{C}$ is $\poly(n) e^{-\Omega(l)}$-Markov, then for any polynomials $r(n)$ and $s(n)$, \texttt{Patching} can run in time $\poly(n)$ and, with probability $1-1/s(n)$ over the random circuit realization, sample from a distribution that is $1/r(n)$-close to the true output distribution in variational distance.
\end{corollary}

Thus, proving that some uniform worst-case-hard circuit family $\mathcal{C}$ is $\poly(n)e^{-\Omega(l)}$-Markov provides another route to proving the part of \Cref{con:efficient} about sampling with small total variation distance error. In \Cref{se:statmech2D}, we will give analytical evidence that if $\mathcal{C}$ is a random circuit distribution of sufficiently low depth and small qudit dimension, then $\mathcal{C}$ is indeed $\poly(n) e^{-\Omega(l)}$-Markov, and if the depth or qudit dimension becomes sufficiently large, then $\mathcal{C}$ is not $\poly(n)f(l)$-Markov for any $f(l) = o(1)$, supporting \Cref{con:transition}, which states that our algorithms exhibit computational phase transitions.

Finally, we note that \texttt{Patching} can also be used to estimate specific output probabilities of a random circuit instance $C$ with high probability if  $C$ is drawn from a distribution $\mathcal{C}$  that is  $\poly(n) e^{-\Omega(l)}$-Markov. This shows that the Markov condition could also be used to prove the second part of \Cref{con:efficient} regarding computing output probabilities with small error.  This is similar to how \texttt{SEBD} can also be used to compute  output probabilities, as discussed in \Cref{se:computing_output_probs}.

\begin{lemma}
Let $\mathcal{C}$ be a circuit distribution over constant depth $d$ and constant qudit dimension $q$ 2D circuits on $n$ qudits which is $\poly(n) e^{-\Omega(l)}$-Markov and invariant under application of a final layer of arbitrary single-qudit gates. Then for a  circuit instance $C$ drawn from $\mathcal{C}$ and a fixed $\vb{x} \in [q]^n$, a variant of  \texttt{Patching} can be used to output a number $\mathcal{D}'(\vb{x})$ in time $\poly(n)$ that satisfies
\begin{equation}
| \mathcal{D}'(\vb{x}) - \mathcal{D}(\vb{x}) | \leq q^{-n}/r(n)
\end{equation}
with probability $1-1/s(n)$ for any polynomials $r(n)$ and $s(n)$, where $\mathcal{D}$ is the output distribution associated with $C$.
\end{lemma}
\begin{proof}
With probability $1-1/\poly(n)$ over the circuit instance $C$, \texttt{Patching} with some choice of lengthscale $l = \Theta(\log n)$ efficiently samples from a distribution $\mathcal{D}'_C$ that is $1/\poly(n)$-close in variational distance to $\mathcal{D}_C$ for any choice of polynomials. Hence, for an output probability $\vb{y}$ chosen uniformly at random and a circuit $C$ drawn from $\mathcal{C}$, it holds that
\begin{equation}
\E_{\vb{y}} \E_C |\mathcal{D}'(\vb{y}) - \mathcal{D}(\vb{y}) | \leq q^{-n}/\poly(n)
\end{equation}
if $l = c \log n$ and $c$ is a sufficiently large constant. By a nearly identical argument to that used in the proof of \Cref{cor:SEBD_probability_bound}, due to the invariance of $\mathcal{C}$ under application of a final layer of single qudit gates, for some fixed $\vb{x} \in [q]^n$ we also have
\begin{equation}
\E_C |\mathcal{D}'(\vb{x}) - \mathcal{D}(\vb{x})| \leq q^{-n}/\poly(n)
\end{equation}
for any choice of polynomial. Finally, it is straightforward to see that an instance of \texttt{Patching} that samples from $\mathcal{D}'$ can also be used to exactly compute $\mathcal{D}'(\vb{x})$ for any $\vb{x}$. (To do this, the algorithm computes conditional probabilities via tensor network contractions as before, except instead of using these conditional probabilities to sample, it simply multiplies them together similarly to how \texttt{SEBD} can be used to compute output probabilities.) Applying Markov's inequality completes the proof.
\end{proof}

\section{Rigorous complexity separation for the ``extended brickwork architecture''}\label{se:proof}
\begin{figure}
\centering
\includegraphics[width=\textwidth]{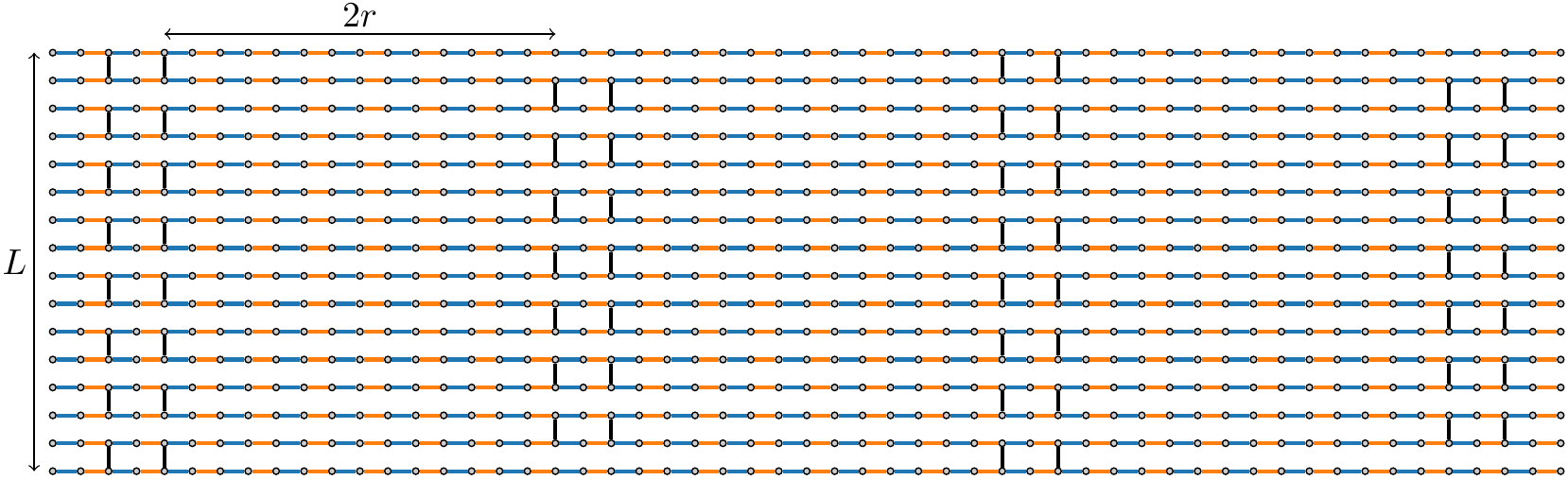}
\caption{Extended brickwork architecture with $n$ qubits. Here, circles represent qubits initialized in the state $\ket{0}^{\otimes n}$, blue lines represent the first layer of gates to act, orange lines represent the second layer, and black lines represent the third and final layer. All gates are chosen Haar-randomly. We let \texttt{Brickwork}$(L,r,v)$ denote the corresponding random circuit with circuit layout depicted in the figure above with vertical sidelength $L$, ``extension parameter'' $2r$ (which gives the distance between vertical gates acting on adjacent pairs of rows), and number of pairs of columns of vertical gates $v$. In the above example, $r=7$ and $v=4$.  The standard brickwork architecture corresponds to $r=1$.  Note that $n = \Theta(L r v)$.}\label{fig:extendedBrickwork}
\end{figure}

In this section, we show that \texttt{SEBD} is provably efficient for certain random circuit families that are worst-case hard (similar facts could also be shown for \texttt{Patching}, but for brevity we restrict our focus to the former algorithm). We define the circuit architecture in \Cref{fig:extendedBrickwork}. It follows immediately from prior works that exactly sampling from the output distribution of this random circuit family for arbitrary circuit instances or near-exactly computing a specific output probability with high probability is classically hard under standard complexity theoretic assumptions. We summarize these observations in the following lemma.

\begin{lemma}
Let $r(L)$ and $v(L)$ be any polynomially bounded functions, with $v(L) \geq L^{a}$ for some $a > 0$. Suppose that there exists a classical algorithm that runs in time $\poly(n)$ and samples from the output distribution of an arbitrary realization of $\texttt{Brickwork}(L, r(L), v(L))$, as defined in \Cref{fig:extendedBrickwork}. Then the polynomial hierarchy collapses to the third level. Suppose there exists a classical algorithm that runs in time $\poly(n)$ and, for an arbitrary fixed output string $\vb{x}$, with probability at least $1-1/\poly(n)$ over choice of random instance, computes the output probability of $\vb{x}$ up to additive error $2^{-\tilde{\Theta}(n^2)}$.  Then there exists a probabilistic polynomial-time algorithm for computing a $\SharpP$-hard function.
\end{lemma}
\begin{proof}
We first note that $\texttt{Brickwork}(L, r(L), v(L))$ supports universal MBQC, in the sense that a specific choice of gates can create a resource state that is universal for MBQC. This is an immediate consequence of the proof of universality of the ``standard'' brickwork architecture (corresponding to $r=1$) proved in \cite{broadbent2009universal}. Indeed, when using the extended brickwork architecture for MBQC, measurements on the long 1D stretches of length $2r$ may be chosen such that the effective state is simply teleported to the end when computing from left to right (i.e., measurements may be chosen such that the long 1D segments simply amount to applications of identity  gates on the effective state). The scaling $v \geq L^{a}$ ensures that MBQC with an extended brickwork resource state suffices to simulate any \textsf{BQP} computation with polynomial overhead. Since a worst-case choice of gates creates a resource state for universal MBQC, an algorithm that can simulate an arbitrary circuit realization can be used to simulate arbitrary single-qubit measurements on a resource state universal for MBQC. Under post-selection, such an algorithm can therefore simulate \textsf{PostBQP} \cite{raussendorf2001one} and hence cannot be efficiently simulated classically unless the polynomial hierarchy collapses to the third level \cite{bremner2010classical}.

Similarly, for some subsets of instances, it is $\SharpP$-hard to compute the output probability of an arbitrary string, since (by choosing gates to create a resource state for universal MBQC) this would allow one to compute output probabilities of universal polynomial-size quantum circuit families which is known to be $\SharpP$-hard.  The result of  \cite{movassagh2019} is then applicable, which implies that if the gates are chosen Haar-randomly, efficiently computing the output probability of some fixed string with probability $1-1/\poly(n)$ over the choice of instance up to additive error bounded by $2^{-\tilde{\Theta}(n^2)}$ implies the ability to efficiently compute a $\SharpP$-hard function with high probability.
\end{proof}

Our goal is to prove that \texttt{SEBD} can efficiently approximately simulate the extended brickwork architecture in the average case for  choices of extension parameters for which the above hardness results apply. To this end, we first show a technical lemma which describes how measurements destroy entanglement in 1D shallow random circuits. In particular, given a 1D state generated by a depth-2 Haar-random circuit acting on qubits, after measuring some contiguous region of spins $B$, the expected entanglement entropy of the resulting post-measurement pure state across a cut going through $B$ is exponentially small in the length of $B$. We defer the proof to \Cref{appendix:error}. (We expect the result to remain true for general constant-depth random circuits in 1D acting on qudits, but we will only need the depth-2 case with qubits.)
\begin{restatable}{lemma}{entanglementDecay}\label{lem:entanglementDecay}
Suppose a 1D random circuit $C$ is applied to qubits $\{1,\dots, n\}$ consisting of a layer of 2-qubit Haar-random gates acting on qubits $(k,k+1)$ for odd $k \in \{1,\dots, n-1\}$, followed by a layer of 2-qubit Haar-random gates acting on qubits $(k,k+1)$ for even $k \in \{1,\dots, n-1\}$. Suppose the qubits of region $B := \{i, i+1, \dots, j\}$ for $j \geq i$ are measured in the computational basis, and the outcome $b$ is obtained. Then, letting $\ket{\psi_b}$ denote the post-measurement pure state on the unmeasured qubits, and letting $A := \{1, 2, \dots, i-1\}$ denote the qubits to the left of $B$,
\begin{equation}
\E S(A)_{\psi_b} \leq c^{|B|}
\end{equation}
for some universal constant $c < 1$, where the expectation is over measurement outcomes and choice of random circuit $C$.
\end{restatable}

We now outline the argument for why \texttt{SEBD} should be efficient for the extended brickwork architecture for sufficiently large extension parameters, before showing this more formally. During the evolution of \texttt{SEBD} as it sweeps from left to right across the lattice, it periodically encounters long stretches of length $2r$ in which no vertical gates are applied. We call these ``1-local regions'' since the gates applied in the corresponding effective 1D dynamics are 1-local when the algorithm is in such a region. Hence, in the effective 1D dynamics, no 2-qubit gates are applied and therefore the bond dimension of the associated MPS cannot increase during these stretches. It turns out that in 1-local regions, not only does the bond dimension needed to represent the state not increase, but it in fact rapidly decays in expectation. If $r$ is sufficiently large, then the effective 1D state at the end of the 1-local region is very close to a product state with high probability, regardless of how entangled the state was before the region. Hence, when  \texttt{SEBD} compresses the MPS describing the effective state at the end of the region, it may compress the bond dimension of the MPS to some fixed constant with very small incurred error. The two-qubit gates that are performed in-between 1-local regions only increase the bond dimension by a constant factor. Hence, with high probability, \texttt{SEBD} can use a $O(1)$ maximal bond dimension cutoff and simulate a random circuit with extended brickwork architecture with high probability.  More precisely, it turns out that the scaling $r \geq \Theta(\log(n))$ is sufficient to guarantee efficient simulation with this argument. We now give the argument in more detail before stating some implications as a corollary.
\begin{figure}
\centering
\includegraphics[width=\textwidth]{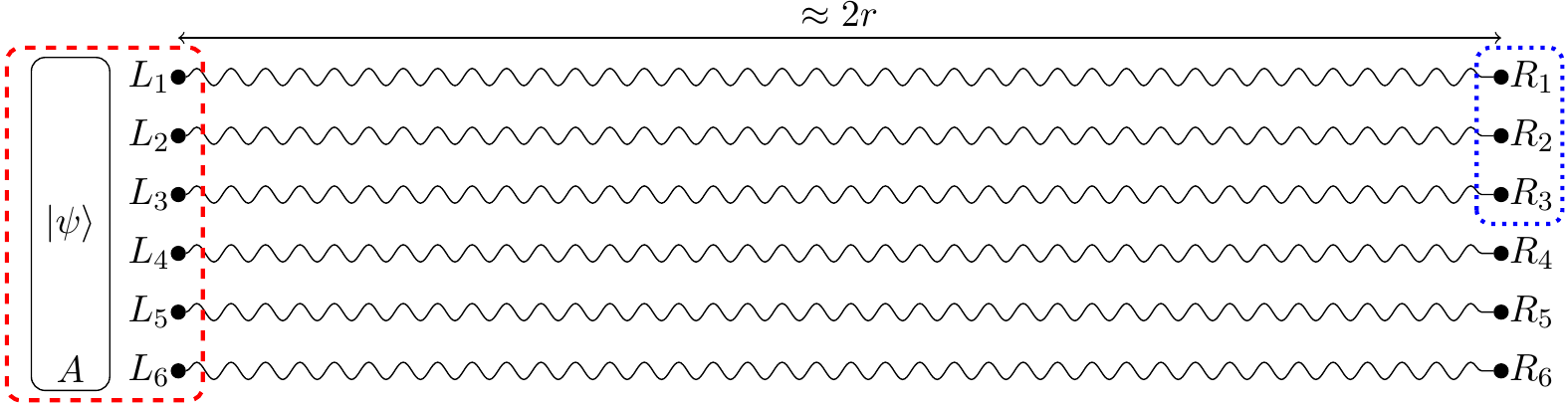}
\caption{Illustration of the state after the qubits of columns $i,i+1,\dots,j$ have been measured, but before gates in the lightcone of registers $A$ and $L$ have been performed.  In each row $i$, we are left with a post-measurement bipartite state $\ket{\phi_i}_{L_i R_i}$ depicted by a wavy line. The expected entanglement entropy $S(L_i)_{\phi_i}$ decays exponentially in $r$. The final state of interest $\ket{\psi'}$ is obtained by applying local unitaries to the qubits in the dashed red box before measuring all of these qubits in the computational basis, inducing the final state $\ket{\psi'}$ on $R = R_1 \cup \cdots \cup R_L$. By concavity of the von Neumann entropy, the expected entanglement entropy of $\ket{\psi'}$ across the cut defined by the dotted blue box is upper bounded by the entanglement entropy across this cut before the unitaries and measurements in the dashed red box are performed.}\label{fig:entangledPairs}
\end{figure}
\begin{lemma}
Let $C$ be an instance of \texttt{Brickwork}$(L,r,v)$. Then, with probability at least $1-2^{-\Theta(r)}$ over the circuit instance, \texttt{SEBD} running with maximal bond dimension cutoff $D = \Theta(1)$ and truncation error parameter $\epsilon = 2^{-\Theta(r)}$ can be used to (1) sample from the output distribution of $C$ up to error $n  2^{-\Theta(r)}$ in variational distance and (2) compute the output probability of an arbitrary output string up to additive error $n  2^{-\Theta(r)}/2^n$ in runtime $\Theta(n)$.
\end{lemma}
\begin{proof}
Suppose the state stored by \texttt{SEBD} immediately before entering into a 1-local region is $\ket{\psi}_A$, defined on register $A$. After another $O(r)$ iterations of \texttt{SEBD}, just before the end of the 1-local region, denote the new one-dimensional state stored by \texttt{SEBD} as $\ket{\psi'}$. Note that $\ket{\psi'}$ is a random state, depending on both the random choices of gates  in the 1-local region and the random measurement outcomes. We now bound the expected entanglement entropy of $\ket{\psi'}$ across an arbitrary cut.

To this end, we observe that the random final state $\ket{\psi'}$ may be equivalently generated as follows. Instead of iterating \texttt{SEBD} as usual for $O(r)$ iterations, we first introduce a contiguous block of qubits that lie in the 1-local region. In particular, for all rows, we introduce all qubits that lie in columns $\{i, i+1, \dots, j\}$. Here, $i$ is chosen to be the leftmost column such that the lightcone of column $i$ does not contain qubits in register $A$. Similarly, $j$ is chosen to be the rightmost column such that the lightcone of qubits in column $j$ does not contain vertical gates. Note that $|i-j| = \Theta(r)$.

We next apply all gates in the lightcone of the qubits of columns $\{i, i+1, \dots, j\}$, before measuring these qubits in the computational basis. Note that in this step, we are effectively performing a set of $L$ one-dimensional depth-2 Haar-random circuits, and then measuring $\Theta(r)$ intermediate qubits for each of the $L$ instances. For each instance, we are left with a (generically entangled) pure state between a ``left'' and ``right'' subsystem, as illustrated in \Cref{fig:entangledPairs}. Let $L_i$ ($R_i$) denote the left (right) subsystem associated with row $i$, and let $\ket{\phi_i}_{L_i R_i}$ denote the associated post-measurement pure state on these subsystems. By \Cref{lem:entanglementDecay}, it follows that the expected entanglement entropy for any 1D instance obeys $\E S(L_i)_{\phi_i} \leq 2^{-\Theta(r)}$ where the expectation is over random circuit instance and measurement outcomes.

The next step is to apply all gates in the lightcone of the qubits of registers $A$ and $L := \cup_i L_i$ before measuring these registers, inducing a (random) 1D post-measurement pure state on subsystem $R := \cup_i R_i$. It is straightforward to verify that the distribution of the random 1D pure state $\ket{\psi'}_{R}$ obtained via this procedure is identical to that obtained from repeatedly iterating \texttt{SEBD} through column $j$\footnote{Strictly speaking, we are actually studying a version of \texttt{SEBD} that only performs the MPS compression step at the end of a 1-local region. Since 1-local operations cannot increase the bond dimension of the associated MPS, the algorithm can forego the compression steps during the 1-local regions without incurring a bond dimension increase.}. Indeed, the procedures are identical up to performing commuting gates and commuting measurements in different orders, which does not affect the measurement statistics or post-measurement states.

Our goal is now to bound the entanglement entropy $S(R_1 R_2 \dots  R_k)_{\psi'}$ in expectation across an arbitrary cut of the post-measurement 1D state. Such a bound follows from the concavity of the von Neumann entropy. Let $\rho_{R_1, \dots, R_k}$ denote the reduced density matrix on these subsystems before the measurements on $A$ and $L$ are performed. Let $\rho^x_{R_1, \dots, R_k}$ denote the reduced density matrix on these subsystems after the measurements on $A$ and $L$ are performed and the outcome $x$ is obtained; note that the final state $\psi'$ implicitly depends on $x$.  Now, letting $\Pr[x]$ denote the probability of obtaining outcome $x$, we have the relation $\sum_x\Pr[x] \rho^x_{R_1 \cdots R_k} = \rho_{R_1 \cdots R_k}$. To see this, observe that for any set of measurement operators $\{M^x\}_x$ satisfying $\sum_x M^{x \dagger} M^x = I$, we have $\rho_{R_1\cdots R_k} = \tr_{\setminus R_1\cdots R_k}\qty(\dyad{\psi'}) = \sum_x \tr_{\setminus R_1\cdots R_k}\qty(M^x \dyad{\psi'} M^{x\dagger}) = \sum_x \Pr[x] \frac{\tr_{\setminus R_1\cdots R_k}\qty(M^x \dyad{\psi'} M^{x\dagger})}{\tr\qty(M^x \dyad{\psi'} M^{x\dagger})} = \sum_x \Pr[x] \rho^x_{R_1 \cdots R_k}$. Now,
\begin{align}
\sum_x \Pr[x] S(R_1 \dots R_k)_{\psi'} &= \sum_x \Pr[x] S(\rho^x_{R_1, \dots, R_k}) \\
&\leq S\qty(\sum_x \Pr[x] \rho^x_{R_1, \dots, R_k}) \\
&= S(\rho_{R_1, \dots, R_k}) \\
&= \sum_{i=1}^k S(R_i)_{\phi_i}
\end{align}
where the first  line follows by definition, the second line follows from concavity of the von Neumann entropy, the third line uses the relation we discussed previously, and in the final line we used the fact that $\rho_{R_1, \dots, R_k}$ is a product state. Hence, we see that for any fixed set of gates and for any outcomes of the measurements of qubits in columns $i, i+1, \dots, j$, the expected entanglement entropy of the final 1D state $\psi'$ on $R$ across any cut is bounded by the entropy across that cut before the measurements on subregions $A$ and $L$. Taking the expectations of both sides of this result with respect to the random gates and measurement outcomes of the qubits in columns $i, i+1, \dots, j$, we finally obtain
\begin{equation}
\E S(R_1 \dots R_k)_{\psi'} \leq L 2^{-\Theta(r)}
\end{equation}
where we used the fact that $k < L$ and $\E S(R_i)_{\phi_i} \leq 2^{-\Theta(r)}$. We now use the fact that the largest eigenvalue $\lambda_{\max}(R_1\cdots R_k)$ of the reduced density matrix is lower bounded as $\lambda_{\max}(R_1\cdots R_k)_{\psi'} \geq 2^{-S(R_1\cdots R_k)_{\psi'}}$; this follows from the fact that Shannon entropy upper bounds min-entropy. Using this inequality as well as Jensen's inequality, we have the bound
\begin{equation}
\E \lambda_{\max}(R_1\cdots R_k) \geq \E 2^{-S(R_1\cdots R_k)_{\psi'}} \geq 2^{-\E S(R_1\cdots R_k)_{\psi'}}  \geq 2^{-L 2^{-\Theta(r)}}  \geq 1 - L 2^{-\Theta(r)}.
\end{equation}
Therefore, if we truncate all but the largest Schmidt coefficient across the $R_k \, : \, R_{k+1}$ cut, we incur an expected truncation error upper bounded by $L 2^{-\Theta(r)}$. Hence, by Markov's inequality, we incur a truncation error upper bounded by $L 2^{-\Theta(r)}$ with probability at least $1-2^{-\Theta(r)}$.

Therefore, if we run \texttt{SEBD} using a \emph{per bond} truncation error of $\epsilon = L 2^{-\Theta(r)}$ and a maximum bond dimension cutoff of $D = O(1)$, the failure probability will be upper bounded by $L v 2^{-\Theta(r)}$; here we used the union bound to upper bound the  probability that any of the $O(Lv)$ bonds over the course of the algorithm becomes larger than the cutoff $D$. Hence, by \Cref{cor:SEBD_error_bound}, for at least $1-2^{-\Theta(r)}$ fraction of random circuit instances, \texttt{SEBD} can sample from the output distribution with variational distance error $L v 2^{-\Theta(r)} < n 2^{-\Theta(r)}$. Similarly, by \Cref{cor:SEBD_probability_bound}, for at least $1-2^{-\Theta(r)}$ fraction of circuit instances, \texttt{SEBD} can compute the probability of the all-zeros output string up to additive error $n 2^{-\Theta(r)} / 2^n$.

Since the runtime of \texttt{SEBD} is $O(n D^3)$ when acting on qubits as discussed previously, and $D$ is chosen to be constant for the version of the algorithm used here, the runtime is $O(n)$.
\end{proof}

With an appropriate choice of $r = \Theta(\log(L))$, the above result implies that \texttt{SEBD} can perform the simulation with error $1/\poly(n)$ for at least $1-1/\poly(n)$ fraction of instances. Similarly, choosing $r$ to be a sufficiently large polynomial in $L$, \texttt{SEBD} can perform the simulation with error $2^{-n^{1-\delta}}$ for $1-2^{-n^{1-\delta}}$ fraction of instances, for any constant $\delta > 0$. We summarize these observations as the following corollary.

\begin{corollary}
For any choice of polynomially bounded $v, p_1, p_2$, for any sufficiently large constant $c$ \texttt{SEBD}  can simulate $1-1/p_1(n)$ fraction of instances of \texttt{Brickwork}$(L, \lceil c \log(L) \rceil, v(L))$ up to error $\varepsilon \leq 1/p_2(n)$ in time $O(n)$. For any choice of $\delta > 0$ and $v(L) \leq \poly(L)$, for any sufficiently large constant $c$  \texttt{SEBD}  can simulate $1-2^{-n^{1-\delta}}$ fraction of instances of \texttt{Brickwork}$(L,\lceil L^c \rceil, v(L))$ up to error $\varepsilon \leq 2^{-n^{1-\delta}}$ in time $O(n)$. Here, ``simulate with error $\varepsilon$'' implies the ability to sample with variational distance error $\varepsilon$ and compute the output probability of some fixed string $\vb{x}$ with additive error $\varepsilon/2^n$.
\end{corollary}

\section{Numerical results}\label{se:numerics}
\begin{figure}[htbp]
	\centering
	\subfloat[\texttt{CHR}]{{\includegraphics[width=0.46\textwidth]{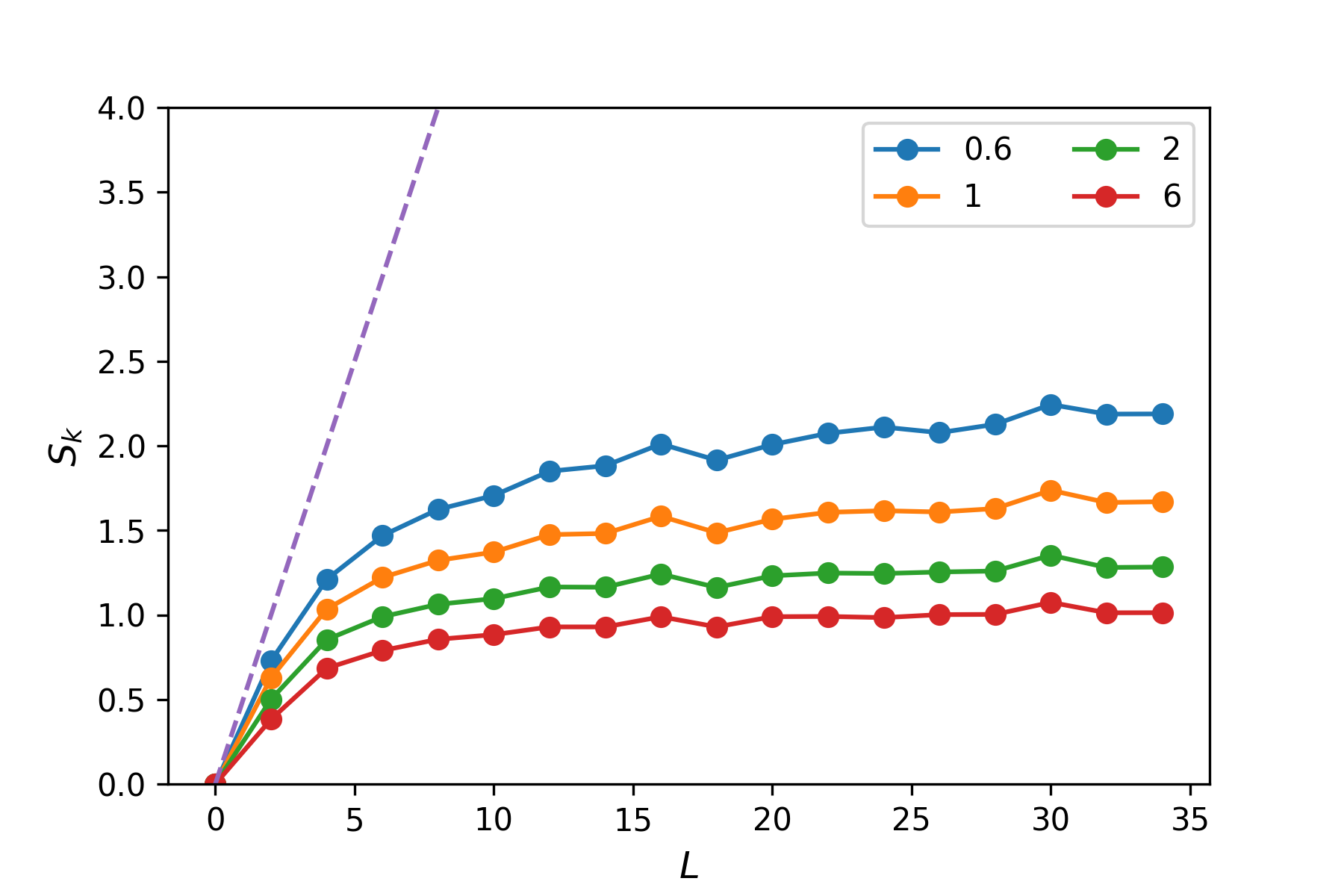} }}%
	\qquad
	\subfloat[Brickwork]{{\includegraphics[width=0.46\textwidth]{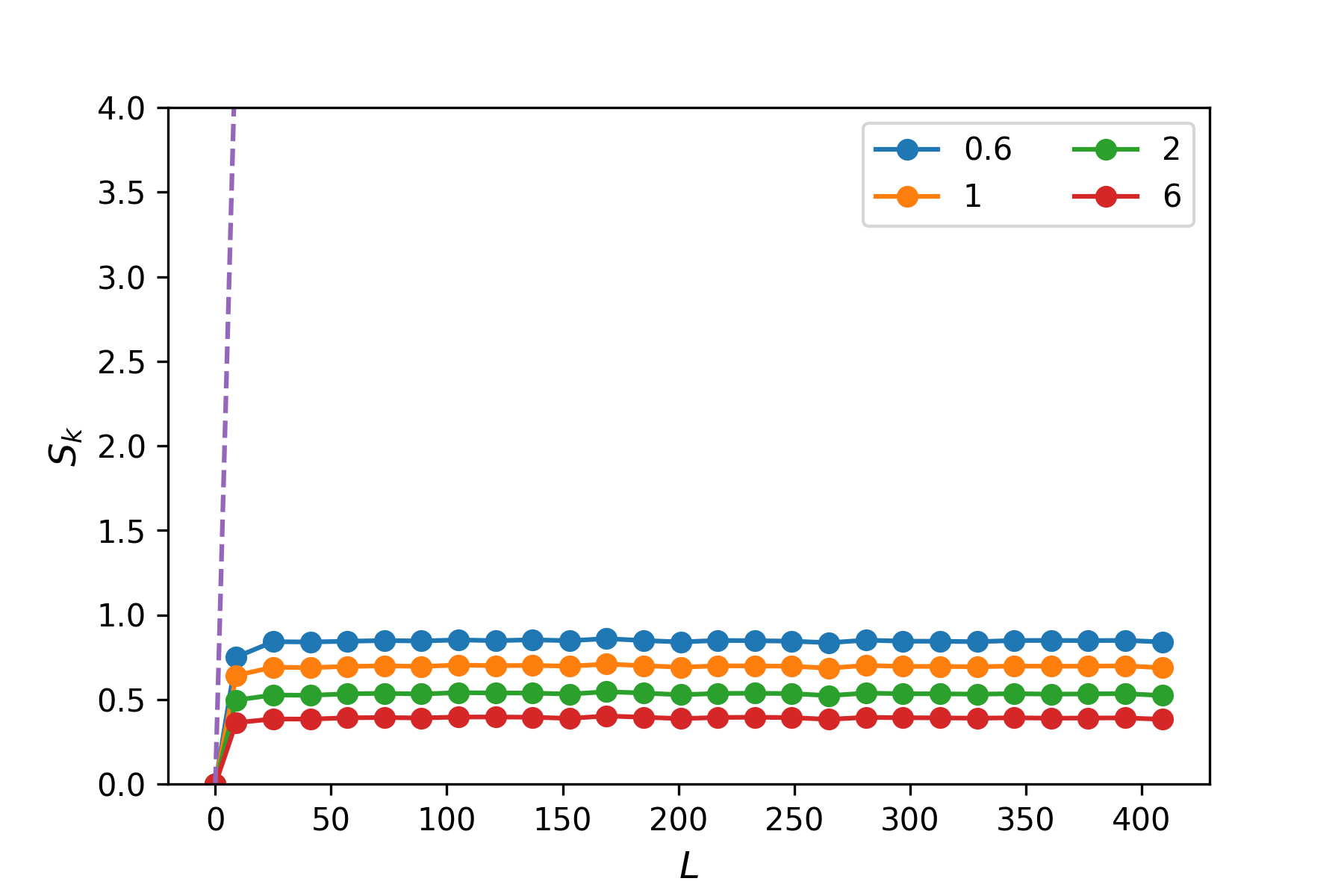} }}%
	\caption{R\'{e}nyi half-chain entanglement entropies $S_k$ versus sidelength $L$ in the effective 1D dynamics for the \texttt{CHR} and brickwork models, after 80 (resp.~550) iterations. Each point represents the entanglement entropy averaged over 50 random circuit instances, and over the final 10 (resp.~50) iterations for the \texttt{CHR} (resp.~brickwork) model. Dashed lines depict the half-chain entanglement entropy scaling of a maximally entangled state, which can be created with a ``worst-case'' choice of gates for both architectures. The maximal truncation error per bond $\epsilon$ was $10^{-10}$ for \texttt{CHR} and $10^{-14}$ for the brickwork model.}\label{fig:entanglement}
\end{figure}

\begin{figure}[htbp]
	\centering
	\includegraphics[width=0.65\textwidth]{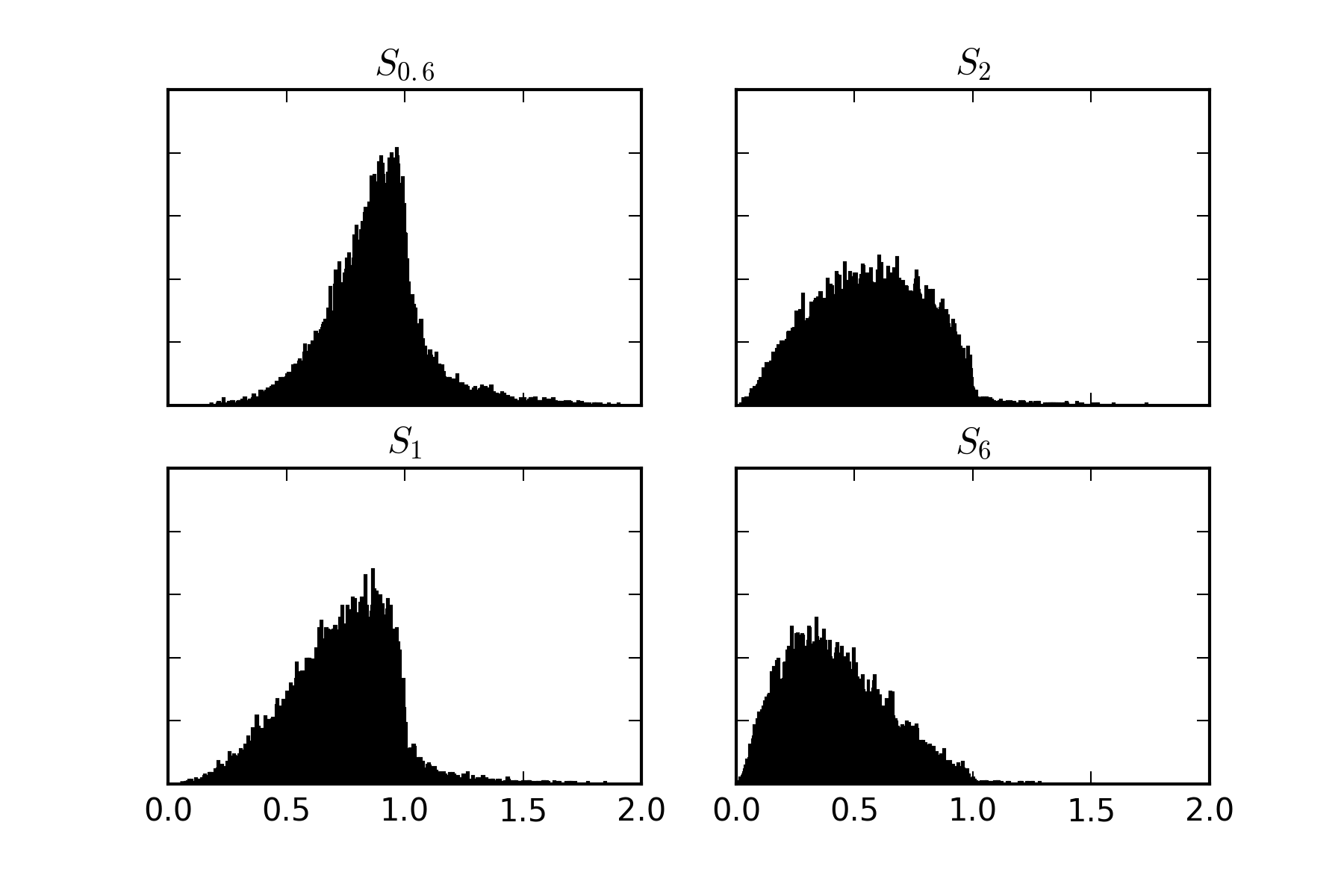}
	\caption{Histograms of observed half-chain R\'{e}nyi entanglement entropies after 49 iterations of the effective 1D dynamics of the brickwork architecture with sidelength $L=49$. Histograms are based on approximately 11,000 trials. Each value is generated from a distinct random circuit realization.}\label{fig:histogram}
\end{figure}

\begin{figure}[htbp]
	\centering
	\subfloat{{\includegraphics[width=0.46\textwidth]{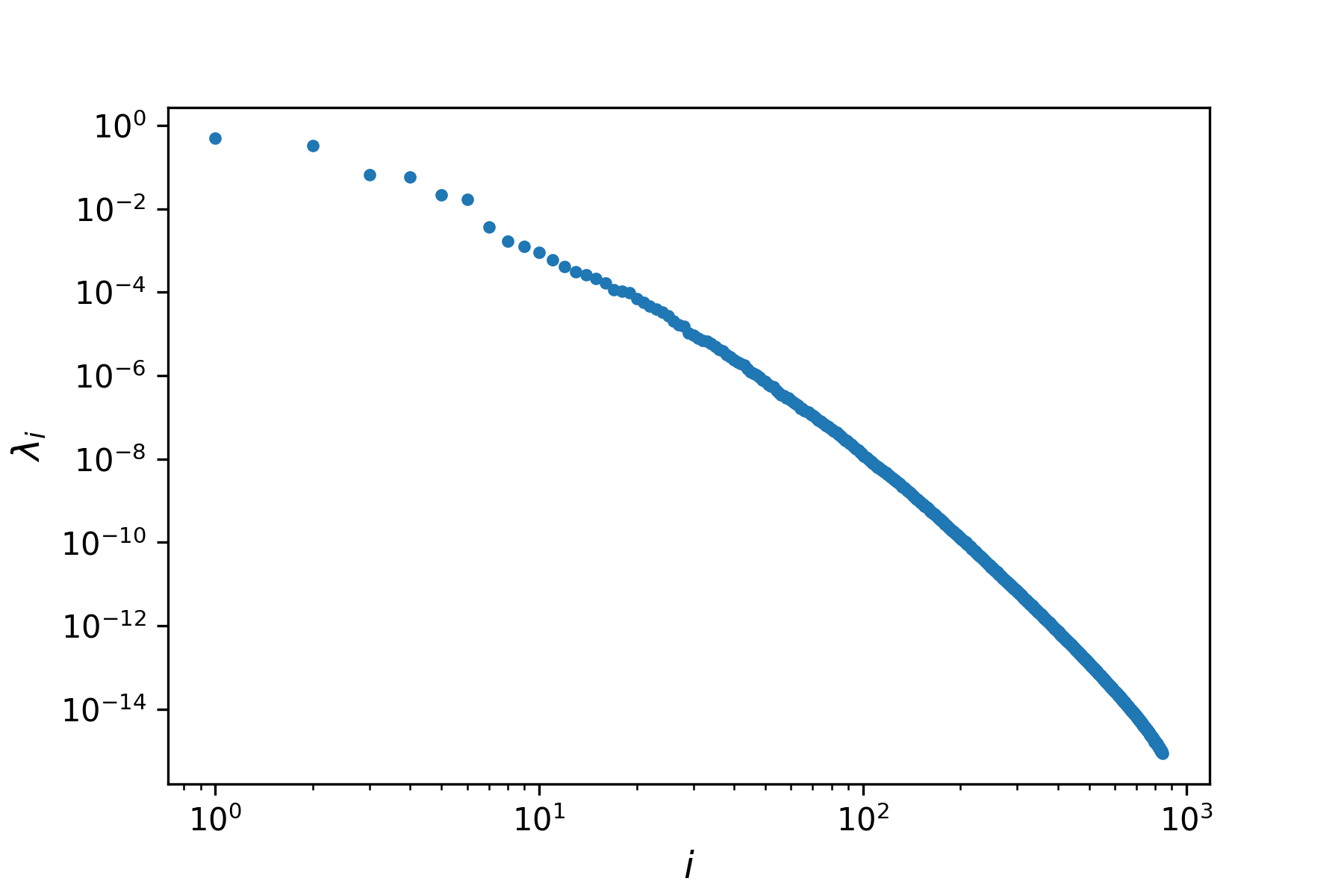} }}%
	\qquad
	\subfloat{{\includegraphics[width=0.46\textwidth]{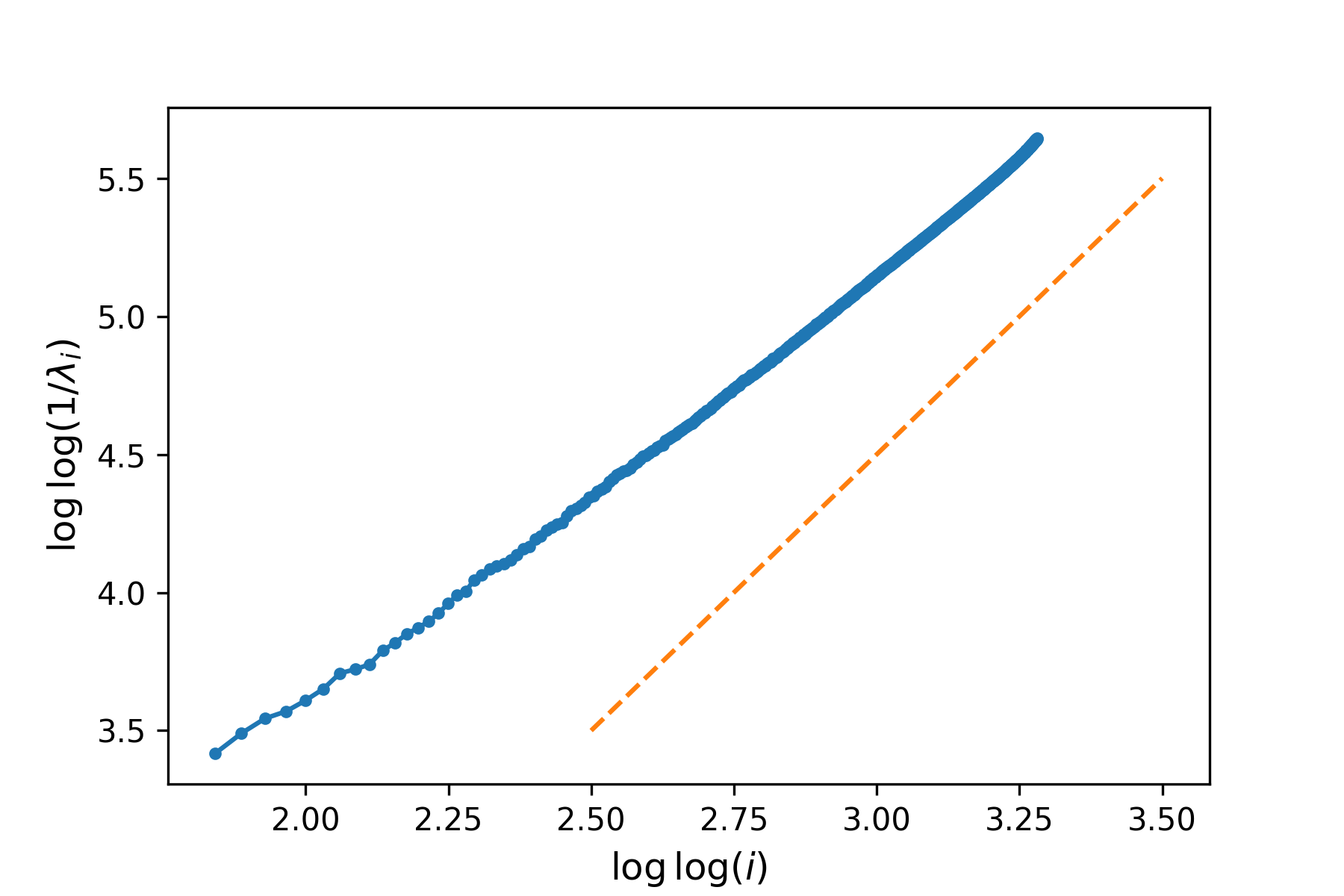} }}%
	\caption{Typical half-chain entanglement spectra $\lambda_1 \geq \lambda_2 \geq \dots$ observed during the effective 1D dynamics of \texttt{CHR}. These plots were generated from an instance with sidelength $L=44$ after running for $44$ iterations, with squared Schmidt values smaller than approximately $10^{-15}$ truncated. The left figure shows the spectra of half-chain eigenvalues.  The downward curvature in the log-log scale indicates superpolynomial decay. The right figure displays the same data (minus the few largest values) on a loglog-loglog scale. The toy model predicts that the blue curve asymptotes  to a  straight line with slope two in the right figure, illustrated by the dashed orange line, corresponding to scaling like $\lambda_i \sim 2^{-\Theta(\log^2(i))}$.  The plot is qualitatively consistent with this prediction. The spectrum for the brickwork model decays too quickly to obtain as useful statistics without going to much higher numerical precision.}\label{fig:spectrum}
\end{figure}

We implemented\footnote{code available at \href{https://github.com/random-shallow-2d/random-shallow-2d}{https://github.com/random-shallow-2d/random-shallow-2d}} \texttt{SEBD} on two families of random circuits: one consisting of depth-3 random circuits defined on a ``brickwork architecture'' consisting of three layers of two-qubit Haar-random gates (\Cref{fig:extendedBrickwork} with parameter $r=1$), and  the other being the random circuit family obtained by applying single-qubit Haar-random gates to all sites of a cluster state -- we referred to this problem as $\texttt{CHR}$ previously. Note that the former architecture has depth three (not including the measurement layer) and the latter has depth four, and both architectures support universal measurement-based quantum computation \cite{broadbent2009universal}, meaning they have the worst-case-hard property relevant for Conjecture \ref{con:efficient}. We did not implement \texttt{Patching}, due to its larger overhead.

Implementing \texttt{SEBD} on a standard laptop, we could simulate typical instances of the $409 \times 409$ brickwork model  with truncation error $10^{-14}$ per bond with a runtime on the order of one minute per sample, and typical instances of the $34\times 34$ \texttt{CHR} model with truncation error $10^{-10}$ per bond with a runtime on the order of five minutes per sample (these truncation error settings correspond to sampling errors of less than $0.01$ in variational distance as derived previously in \Cref{sec:SEBD}).  We in fact simulated instances of \texttt{CHR} with grid sizes as large as $50\times 50$, although due to the significantly longer runtime for such instances we did not perform large numbers of trials for these cases. In the case of the $409\times 409$ brickwork model, performing over $3000$ trials (consisting of generating a random circuit instance and generating a sample from its output distribution using a truncation error of $10^{-14}$) and finding no instances we could not simulate, then with $95\%$ confidence, we may conclude that the probability of successfully simulating a random instance with this truncation error is over $0.999$. Using the bound derived in \Cref{sec:SEBD}, we can therefore conclude that, with $95\%$ confidence, we can sample from the output distribution of at least $0.9$ fraction of $409\times 409$ circuit instances with variational distance error at most $0.01$.  Intuitively, we expect the true simulable fraction to be much larger than this statistical guarantee, as it appears that the entanglement in the effective 1D dynamics only grows extensively for highly structured instances. Note that for both models, the runtime for a fixed truncation error was qualitatively highly concentrated around the mean. We expect that substantially larger instances of both random circuit families could be quickly simulated with more computing power, although we previously argued in \Cref{se:overview} that the $409 \times 409$ simulation of the brickwork architecture is already far beyond what could have been achieved by previous simulation methods that we are aware of.

The discrepancy between maximal lattice sizes achieved for the two architectures is a result of the fact that the two have very different effective 1D dynamics. In particular, the entanglement of the effective dynamics for the brickwork architecture saturates to a significantly smaller value than that of the cluster state architecture. And even more directly relevant for prospects of fast simulation, the  typical spectrum of Schmidt values across some cut of the effective 1D dynamics for the brickwork architecture decays far more rapidly than that of the 1D dynamics for \texttt{CHR}.  For this reason, the slower-decaying eigenvalue spectrum of \texttt{CHR} was significantly more costly for the runtime of the algorithm. (In fact, the eigenvalue spectrum of the brickwork model decayed sufficiently quickly that we were primarily limited not by the runtime of our algorithm, but by our numerical precision, which could in principle be increased.) But while the slower decay of the spectrum for the \texttt{CHR} model necessitated a longer runtime for a given sidelength, it allowed us to study the functional form of the spectrum and in particular compare against the predictions of the toy model of \Cref{sec:conjectured} as we discuss below.

While we were computationally limited to probing low-depth and small-size models, our numerical results point toward \texttt{SEBD} having an asymptotic running time for both models bounded by $\poly(n,1/\varepsilon,1/\delta)$ in order to sample with variational distance $\varepsilon$ or compute output probabilities with additive error $\varepsilon/q^n$ with probability $1-\delta$, suggesting that \Cref{con:efficient} is true.  Our numerical evidence for this is as follows.

\begin{enumerate}
	\item We find that the effective 1D dynamics associated with these random circuit families appear to be in area-law phases, as displayed in \Cref{fig:entanglement}. That is, the entanglement does not grow extensively with the sidelength $L$, but rather saturates to some constant (which appears to be concentrated, as illustrated in \Cref{fig:histogram}). We furthermore observe qualitatively identical behavior for some R\'{e}nyi entropies $S_\alpha$ with $\alpha < 1$. It is known \cite{schuch2008entropy} that this latter condition is sufficient to imply that a 1D state may be efficiently described by an MPS, indicating that \texttt{SEBD} is efficient for these circuit families and that \Cref{con:efficient} is true.

	\item  For further evidence of efficiency, we study the functional form of the entanglement spectra of the effective 1D dynamics. For the effective 1D dynamics corresponding to \texttt{CHR}, we observe superpolynomial decay of eigenvalues (i.e.~squared Schmidt values) associated with some cut, displayed in \Cref{fig:spectrum}, indicating that choosing a maximal bond dimension of $D = \poly(1/\epsilon)$ is more than  sufficient to incur less than $\epsilon$ truncation error per bond. The observed spectrum tends toward a scaling which is consistent with the asymptotic scaling of $\lambda_i \sim 2^{-\Theta(\log^2 (i))}$ predicted by the toy model of \Cref{sec:conjectured} and consistent with our \Cref{con:aggressive}. Note that this actually suggests that the required bond dimension of \texttt{SEBD} may be even smaller than $\poly(1/\epsilon)$, scaling like $D = 2^{\Theta(\sqrt{\log(1/\epsilon)})}$.
\end{enumerate}

While these numerical results may be surprising given the prevalence of average-case hardness conjectures for quantum simulation, they are not surprising from the perspective of the recent works (discussed in previous sections) that find strong evidence for an entanglement phase transition from an area-law to volume-law phase for 1D unitary-and-measurement processes driven by measurement strengths. Since the effective dynamics of the 2D random shallow circuits we study are exactly such processes, our numerics simply point out that these systems are likely on the area-law side of the transition. (However, no formal universality theorems are known, so the various  models of unitary-and-measurement circuits that have been studied are generally not known to be equivalent to each other.) In the case of the brickwork architecture, we are also able to provide independent analytical evidence (Section \ref{sec:brickwork}) that this is the case by showing the ``quasi-entropy'' $\tilde{S}_2$ for the 1D process is in the area-law phase. We leave the problem of numerically studying the precise relationship between circuit depth, qudit dimension, properties of the associated stat mech models (including ``quasi-entropies'') as discussed in subsequent sections, and the performance of \texttt{SEBD} for future work. In particular, simulations of larger depth and larger qudit local dimension could be used to provide numerical support for \Cref{con:transition}, which claims that as these parameters are increased the circuit architectures eventually transition to a regime where our algorithms are no longer efficient.

\section{Mapping random circuits to statistical mechanical models}\label{se:statmech}

The close relationship between general tensor networks and statistical mechanics has been explored in previous work (see, e.g.~\cite{robeva2018duality}). The relationship becomes especially useful when the tensors are drawn from random ensembles. For example, a mapping from random tensor networks to stat mech models was used to study holographic duality in \cite{hayden2016holographic, vasseur2018entanglement}. The same idea was applied to tensor networks where the tensors are Haar-random unitary matrices, i.e.~random quantum circuits, to study quantum chaos and quantify entanglement growth under random unitary dynamics \cite{nahum2018operator,von2018operator,zhou2019emergent,hunter2019unitary}. Very recently, \cite{bao2019theory} and \cite{jian2019measurement} used this mapping to argue that the phase transition from area-law to volume-law scaling of entanglement entropy numerically observed in circuits consisting of Haar-random local gates mixed with projective measurements is related to the disorder-to-order phase transition in Ising-like statistical mechanical models.

We present the details of this mapping for unitary quantum circuits in Section \ref{se:mapping} and then in Section \ref{se:weakmeasurementmapping}, we show how it is applied to 1D circuits interspersed with weak measurements, as was done in \cite{bao2019theory,jian2019measurement} (however, we analyze a different weak measurement). This provides background and context for our application of the mapping directly to 2D circuits, the relevant case for our algorithms, in Section \ref{se:statmech2D}.

\subsection{Mapping procedure} \label{se:mapping}

\paragraph{Setup.}
{
Let our system consist of $n$ qudits of local dimension $q$. The circuits we consider are specified by a sequence of pairs of qudits (indicating where unitary gates are applied) and single-qudit weak measurements; this sequence can be assembled into a quantum circuit diagram. The single-qudit measurements are each described by a set $\mathcal{M}$ of measurement operators along with a probability distribution $\mu$ over the set $\mathcal{M}$. These sets are normalized such that $\tr(M^\dagger M)$ is constant for all $M \in \mathcal{M}$ and $\E_{M\leftarrow\mu} M^\dagger M = \mathbb{I}_q$ where $\mathbb{I}_q$ is the $q \times q$ identity matrix. Thus we have $\tr(M^\dagger M)=q$ for all $M$. The introduction of a probability measure over $\mathcal{M}$ in our notation, which was also used in \cite{jian2019measurement}, is not conventional, but it is equivalent to the standard formulation and will be important for later definitions.

When a measurement is performed, if the state of the system at the time of measurement is $\sigma$, the probability of measuring the outcome associated with operator $M$ is $\mu(M)\tr( M\sigma M^\dagger)$ (Born rule for quantum measurements). For a fixed outcome $M$, the quantity $\tr( M\sigma M^\dagger)$ is a function of $\sigma$ that we refer to as the \textit{relative likelihood} of obtaining the outcome $M$ on the state $\sigma$, since it gives the ratio of the probability of obtaining outcome $M$ in the state $\sigma$ to the probability of obtaining outcome $M$ in the maximally mixed state $\frac{1}{q}\mathbb{I}_q$.  After obtaining outcome $M$, the state is updated by the rule $\sigma \rightarrow M\sigma M^\dagger/\tr(M\sigma M^\dagger)$.
Thus a pure initial state remains pure throughout the evolution. For notational convenience and without loss of generality, we will assume that for each $u$, the $u$th unitary is immediately followed by single-qudit measurements $(\mathcal{M}_u, \mu_u)$ and $(\mathcal{M}'_u,\mu'_u)$ on the qudits $a_u,a'_u \in [n]$ that are acted upon by the unitary, respectively; in the case no measurement is performed, we may simply take $\mathcal{M}_u$ to consist solely of the identity operator, and in the case that more than one measurement is performed, we may multiply together the sets of measurement operators and their corresponding probability distributions to form a single set describing the overall weak measurement.

Thus, the (non-normalized) output state of the circuit with $l$ unitaries acting on the initial state $\ket{1\ldots 1}$ can be expressed as
\begin{equation}
    \rho = {M'}_l M_l U_l\ldots M'_1M_1U_1\ket{1\ldots 1}\bra{1\ldots 1} U_1^\dagger M_1^\dagger
    (M'_1)^\dagger\ldots U_l^\dagger M_l^\dagger (M'_l)^\dagger\label{eq:rhocircuitoutput}
\end{equation}
where each unitary $U_u$ is chosen from the Haar measure over unitaries acting on qudits $a_u$ and $a'_u$, while $M_u$ and $M'_u$ are the measurement operators associated with the measurement outcome obtained upon performing a measurement on qudits $a_u$ and $a'_u$, respectively, following application of unitary $U_u$.
}

\paragraph{Goal.}
{
The objective of the stat mech mapping is to learn something about the entanglement entropy for the output state $\rho$ on some subset $A$ of the qudits. The $k$-R\'{e}nyi entanglement entropy for the state $\rho$ on region $A$ is defined as
\begin{equation}
 S_k(A)_\rho = \frac{1}{1-k}\log\left(\frac{Z_{k,A}}{Z_{k,\emptyset}}\right)
\end{equation}
where
\begin{align}
    Z_{k,\emptyset}&= \tr(\rho)^k \label{eq:Zk0def}\\
    Z_{k,A}&= \tr(\rho_A^k). \label{eq:ZkAdef}
\end{align}
and $\rho_A$ is the reduced density matrix of $\rho$ on region $A$. The von Neumann entropy $S(A)_\rho = -\tr\left(\frac{\rho_A}{\tr(\rho)} \log\left(\frac{\rho_A}{\tr(\rho)}\right)\right)$ represents the $k \rightarrow 1$ limit.

For the purposes of assessing the efficiency of our algorithms, we would like to be able to calculate the average value of the $k$-R\'{e}nyi entropy over the random choice of the unitaries in the circuit and for measurement outcomes drawn randomly according to the Born rule. Mathematically, we let the notation $\E_U(Q)$ represent the expectation of a quantity $Q$ when the unitaries of the circuit are drawn uniformly at random from the Haar measure and the measurement outcomes are drawn at random from the distribution over their respective sets of measurement operators
\begin{equation}\label{eq:ECQ}
    \E_U(Q) = \E_{M_1 \leftarrow \mu_1}\E_{M'_1 \leftarrow \mu'_1}\ldots\E_{M_l \leftarrow \mu_1}\E_{M'_l \leftarrow \mu'_1}\int_{U(q^2)} dU_1\ldots \int_{U(q^2)}dU_l \; Q
\end{equation}
Here $\int_{U(q^2)}$ denotes integration over the Haar measure of the unitary group with dimension $q^2$.
To take into account the Born rule, a certain choice of unitaries and measurement outcomes leading to the output state $\rho$, as in Eq.~\eqref{eq:rhocircuitoutput}, should be weighted by $\tr(\rho)$, i.e.~the product of the relative likelihoods of all the measurement outcomes. Thus, the relevant average $k$-R\'{e}nyi entropy values are given by
\begin{equation}\label{eq:entropydesired}
    \langle S_k(A)_\rho \rangle := \frac{\E_U( \tr(\rho) S_{k}(A)_\rho)}{\E_U(\tr(\rho))} = \frac{1}{1-k}\frac{\E_U\left( \tr(\rho)\log\frac{Z_{k,A}}{Z_{k,\emptyset}}\right)}{\E_U(\tr(\rho))}.
\end{equation}
As we will see, the quantity naturally computed by the stat mech model is not $\langle S_k(A)_\rho \rangle$, but rather the ``quasi-entropy" given by
\begin{equation}\label{eq:freeenergydifference}
 \tilde{S}_k(A) := \frac{1}{1-k}\log\left(\frac{\E_U(\tr(\rho)^k \frac{Z_{k,A}}{Z_{k,\emptyset}})}{\E_U(\tr(\rho)^k)}\right) = \frac{1}{1-k}\log\left(\frac{\E_U(Z_{k,A})}{\E_U(Z_{k,\emptyset})}\right) = \frac{F_{k,\emptyset}-F_{k,A}}{1-k}
\end{equation}
where $F_{k,\emptyset/A} := -\log(\E_U(Z_{k,\emptyset/A}))$ will be associated with the ``free energy'' of the classical stat mech model that the circuit maps to. When clear, we abbreviate $\tilde{S}_k(A)$ by $\tilde{S}_k$ and $\langle S_k(A)_\rho \rangle$ by $\langle S_k \rangle$. It is apparent that $\langle S_k\rangle$ is not equal to the quasi-entropy $\tilde{S}_k$: the definition of $\langle S_k \rangle$ weights circuit instances by $\tr(\rho)$ and takes the $\log$ before the expectation, while the definition of $\tilde{S}_k$ weights by $\tr(\rho)^k$ and takes the $\log$ afterward. Indeed, it is possible for $\tilde{S}_k$ to be smaller than some constant independent of $L$ (area law), while $\langle S_k \rangle$ scales extensively with $L$ (volume law) due to fluctuations of the random variable $Z_{k,A}$ away from its average value toward 0.

Importantly, though, $\tilde{S}_k \rightarrow \langle S_k\rangle$ as $k \rightarrow 1$; in this limit both quantities approach the expected observed von Neumann entropy $\langle S \rangle$ of the circuit output. This conclusion is justified by L'Hospital's rule and noting that $Z_{k,A}/Z_{k,\emptyset} \rightarrow 1$ as $k\rightarrow 1$. We will see that the stat mech model can only be applied for integers $k \geq 2$, so unfortunately it does not allow for direct access to formula $\tilde{S}_k$ in this limit. Nonetheless, the fact that the quasi-entropy $\tilde{S}_k$ is intimately related to the actual expectation of the von Neumann entropy lends some justification to the use of $\tilde{S}_k$ as a proxy for $\langle S_k \rangle$ even when $k\neq1$. This is further justified by previous work studying random 1D circuits without measurements; in \cite{von2018operator}, the growth rate of $\tilde{S}_2$ in random 1D circuits was calculated using the stat mech mapping (note that when there are no measurements, we have $\tr(\rho) = 1$ and thus $Z_{k,\emptyset} = 1$ for any $k$), and no significant difference was found with numerical calculations of $\langle S_2 \rangle$. Moreover, \cite{zhou2019emergent} used the so-called \textit{replica trick} to directly compute $\langle S_2 \rangle$ as a series in $1/q$, where $q$ is the qudit local dimension, and the leading term of this expansion agrees with $\tilde{S}_2$. While the replica trick is not completely rigorous, this gives a clear indication that $\tilde{S}_2$ is a valid substitute for $\langle S_2 \rangle$ in the $q\rightarrow\infty$ limit and it is a good approximation for finite $q$ (when there are no measurements).
}

\paragraph{Mapping.}
{
We now describe the procedure for mapping a quantum circuit diagram to a classical statistical mechanical model, such that quantities $\E_U(Z_{k,\emptyset})$ and $\E_U(Z_{k,A})$ for integers $k\geq 2$ are given by partition functions of the stat mech model. This follows work in \cite{nahum2018operator,von2018operator,zhou2019emergent,hunter2019unitary,bao2019theory,jian2019measurement}, although our presentation is for the most part self-contained. Here we present merely how to perform the mapping, leaving the details of its justification to Appendix \ref{app:statmechdetails}.

To define the stat mech model we must specify two ingredients: first, the nodes and edges that form the interaction graph on which the model lives, and second, the details of the interactions between nodes that share an edge. The graph, which is the same for all $k$, is formed from the circuit diagram as follows. First, we replace each Haar-random unitary (labeled by integer $u$) in the circuit diagram with a pair of nodes, which we refer to as the \textit{incoming} node $t_u$ and \textit{outgoing} node $s_u$ for that unitary, and we connect nodes $t_u$ and $s_u$ by an edge. Then, we add edges between the outgoing node $s_{u_1}$ of unitary $u_1$ and the incoming node $t_{u_2}$ of another unitary $u_2$ when $u_2$ acts immediately after $u_1$ on the same qudit (ignoring for now the presence of any measurements made in between). Finally, we introduce a single auxiliary node $x_a$ for each qudit $a \in [n]$ and add a single edge connecting $x_a$ to the outgoing node $s_u$ for unitary $u$ if $u$ is the final unitary of the circuit to act on qudit $a$. Thus, all of the outgoing nodes have degree equal to three. The incoming node for each unitary has degree three minus the number of qudits (0, 1, or 2) for which the unitary is the first in the circuit to act on that qudit. We provide an example of this mapping in Figure \ref{fig:statmechexample}.
\begin{figure}[ht]
   \centering
   \includegraphics[width=0.9\columnwidth]{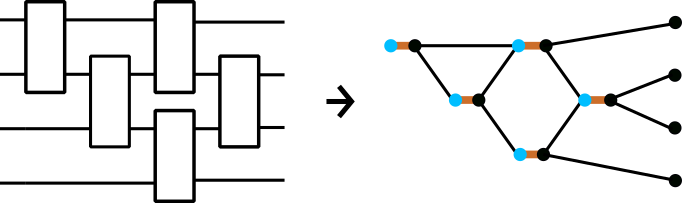}
    \caption{ \label{fig:statmechexample} Example of stat mech mapping applied to a circuit diagram with 4 qudits and 5 Haar-random gates. Thick orange edges carry Weingarten weight. Black edges carry measurement-dependent weight.}
\end{figure}
Each node in the graph may now be viewed as a spin that takes on one of $k!$ values, corresponding to an element of the symmetric group $S_k$. A spin configuration is given by an assignment $(\sigma_u,\tau_u) \in S_k\times S_k$ for each pair of nodes $(s_u,t_u)$, as well as an assignment $\chi_a \in S_k$ to each auxiliary node $x_a$. The main utility of the stat mech mapping is then given by the following equation, expressing the quantities $\E_U(Z_{k,\emptyset/A})$ as a sum over spin configurations on this graph

\begin{equation}\label{eq:partitionfunction}
    \E_U(Z_{k,\emptyset/A}) = \sum_{\{\sigma_u\}_u,\{\tau_u\}_u} \prod_{u}\text{weight}(\langle s_u t_u \rangle) \prod_{\langle s_{u_1} t_{u_2}\rangle}\text{weight}(\langle s_{u_1} t_{u_2} \rangle)\prod_{\langle s_u x_a \rangle}\text{weight}(\langle s_u x_{a}\rangle)
\end{equation}
This is a partition function --- a weighted sum over spin configurations where the weight of each term is given by a product of factors that depend only on the spin value of a pair of nodes $(s,t)$ connected by an edge, denoted $\langle s t \rangle$. In this case, the sum runs only over the values $\sigma_u$ and  $\tau_u$, of the incoming and outgoing nodes; the values $\chi_a$ of the auxiliary nodes are fixed across all the terms and encode the boundary conditions that differ between $\E_U(Z_{k,\emptyset})$ and $\E_U(Z_{k,A})$. We define the free energy to be the negative logarithm of this partition function (see Eq.~\eqref{eq:freeenergydifference}), mirroring the standard relationship $F=-k_B T \log(Z)$ between the free energy and the partition function from statistical mechanics, with $k_B T$ set to 1.

We now specify the details of the interaction by defining the weight function for different edges. There are only two different kinds of interactions. Edges $\langle s_ut_u \rangle$ between incoming and outgoing nodes of the same unitary have
\begin{equation}
    \text{weight}(\langle s_ut_u \rangle) = \Wg(\tau_{u}\sigma_{u}^{-1},q^2) \label{eq:weightsutu}
\end{equation}
where $\Wg(\pi,q^2)$ is the Weingarten function. The Weingarten function arises from performing the integrals over the Haar measure in Eq.~\eqref{eq:ECQ}, and one formula for it is given in the appendix in \Cref{eq:weingarten}. Note that there exist permutations $\pi$ for which $\Wg(\pi,q^2)<0$, so the overall weight of a configuration can be negative and our stat mech model would only correspond to a physical model with complex-valued energy.

Meanwhile, edges $\langle s_{u_1} t_{u_2}\rangle$ connecting nodes of successive unitaries $u_1$ and $u_2$ (resp.~edges $\langle s_u x_a \rangle$ connecting outgoing nodes to auxiliary nodes) have weight that depends both on the setting of variables $\sigma_{u_1}$ and $\tau_{u_2}$ (resp.~$\sigma_u$ and $\chi_a$) as well as on the distribution $\mu$ over some set $\mathcal{M}$ of single-qudit measurement operators that act on the qudit between the application of unitaries $u_1$ and $u_2$ (resp.~after unitary $u$). This weight is given by
\begin{align}
    \text{weight}(\langle s_{u_1} t_{u_2}\rangle)&=\E_{M \leftarrow \mu} \tr((M^\dagger M)^{\otimes k} W_{\sigma_{u_1}\tau_{u_2}^{-1}}) \label{eq:weightsu1tu2}\\
    \text{weight}(\langle s_u x_{a}\rangle)&=\E_{M \leftarrow \mu}\tr((M^\dagger M)^{\otimes k} W_{\sigma_u\chi_{a}^{-1}}) \label{eq:weightsuxa}
\end{align}
where $W_\pi$ is the operator acting on a $k$-fold tensor product space that performs the permutation $\pi$ of the registers. Later, in Section \ref{se:statmechpatching}, we will be interested in expressing entropies of the \textit{classical} output distribution of the circuit in terms of partition functions and to handle this case we will update Eq.~\eqref{eq:weightsuxa}. Note that the quantity $\tr(X^{\otimes k} W_\pi)$ is equal for all $\pi$ with the same cycle structure, which corresponds to some partition $\lambda = (\lambda_1, \ldots, \lambda_r)$ of $k$, where $\sum_i\lambda_i = k$ and $\lambda_1\geq\ldots \geq \lambda_r > 0$. Then we have
\begin{equation}
    \tr(X^{\otimes k} W_\pi) = \prod_{i=1}^r\tr(X^{\lambda_i})
\end{equation}
This formula allows us to simplify the weight formulas \eqref{eq:weightsu1tu2} and \eqref{eq:weightsuxa} in a few special cases. If no measurement is made, then $\mathcal{M}=\{I\}$ and $\text{weight}(\langle s_{u_1} t_{u_2}\rangle) = q^{C(\sigma_{u_1}\tau_{u_2}^{-1})}$, where $C(\pi)$ is the number of cycles $r$ in the permutation $\pi$. On the other hand, if a projective measurement onto one of the $q$ basis states is made, then $\mathcal{M} = \{\sqrt{q}\Pi_m\}_{m=1}^{q}$ and $\mu$ is the uniform distribution, where $\Pi_m = \ket{m}\bra{m}$. Since in this case $\tr((M^\dagger M)^w) = q^w$ for any power $w$ and any $M \in \mathcal{M}$, we have $\text{weight}(\langle s_{u_1} t_{u_2}\rangle) = q^{k-1}$ for any pair $\sigma_{u_1}, \tau_{u_2}$.

The final piece of this prescription is setting the value $\chi_{a}$ for each of the auxiliary nodes $x_a$ at the end of the circuit, which can be seen as fixing the boundary conditions for the stat mech model. These nodes are fixed to the same value for each term in the sum and depend on whether we are calculating $\E_U(Z_{k,\emptyset})$ or $\E_U(Z_{k,A})$, and whether the qudit $a$ is in the region $A$. For $\E_U(Z_{k,\emptyset})$, the value $\chi_{a}$ is fixed to the identity permutation $e$ for every $a$. Meanwhile, for $\E_U(Z_{k,\emptyset})$, we ``twist'' the boundary conditions and change $\chi_a$ to be the $k$-cycle $(1\ldots k)$ if $a$ is in $A$, leaving $\chi_a = e$ if $a$ is in the complement of $A$.

}

\paragraph{Special case of $k=2$.}
{
When $k=2$, the symmetric group $S_k$ has only 2 elements, identity $e$ and swap $(12)$, so the quantities $\E_U(Z_{2,\emptyset})$ and $\E_U(Z_{2,A})$ map to partition functions of Ising-like classical stat mech models where each node takes on one of two values. Furthermore, in the $k=2$ case with no measurements, it was shown in \cite{nahum2018operator,von2018operator} (see also \cite{zhou2019emergent,hunter2019unitary}) that one can get rid of all negative terms in the partition function by decimating half of the nodes, i.e.~explicitly performing the sum over the values of the incoming nodes $\{\tau_u\}_u$ in Eq.~\eqref{eq:partitionfunction}.  This continues to be true even when there are measurements in between unitaries in the circuit. However, the decimation causes the two-body interactions to become three-body interactions between any three nodes $s_{u_1},s_{u_2},s_{u_3}$ when unitary $u_3$ succeeds unitaries $u_1$ and $u_2$ and shares a qudit with each. The lack of negative weights for $k=2$ is convenient because it allows one to view the system as a classical spin model at a real temperature and can therefore be analyzed with well-studied numerical techniques like Monte Carlo sampling.
}

\subsection{Mapping applied to 1D circuits with weak measurements}\label{se:weakmeasurementmapping}

In Section \ref{se:CHR}, we discussed the connection between the effective 1D dynamics of our \texttt{SEBD} algorithm and previous work \cite{li2018quantum,chan2018weak,skinner2019measurement,li2019measurement,szyniszewski2019entanglement,choi2019quantum,gullans2019dynamical,gullans2019scalable,zabalo2019critical} on 1D Haar-random circuits with some form of measurements in between each layer of unitaries.

In this subsection, we apply the stat mech mapping to the 1D with weak measurement model and explain the connection between the area-law-to-volume-law transition that has been observed in numerical simulations and the disorder-to-order thermal transition in the classical stat mech model, which occurs at a non-zero critical temperature $T_c$. This analysis was first performed in \cite{bao2019theory} and independently in \cite{jian2019measurement}. The results presented in this section are a reproduction of their analysis but for a different weak measurement, chosen to be relevant for the dynamics of the \texttt{SEBD} algorithm acting on the \texttt{CHR} problem. We include this analysis for two purposes: first, to shed light on the behavior of \texttt{SEBD} acting on \texttt{CHR}, and second, to serve as a more complete example of the stat mech mapping in action, prior to our application of the mapping directly to 2D circuits in Section \ref{se:statmech2D}, where the analysis is somewhat more complicated.

\begin{figure}[ht]
    \centering
        \def\svgwidth{0.9\columnwidth}
    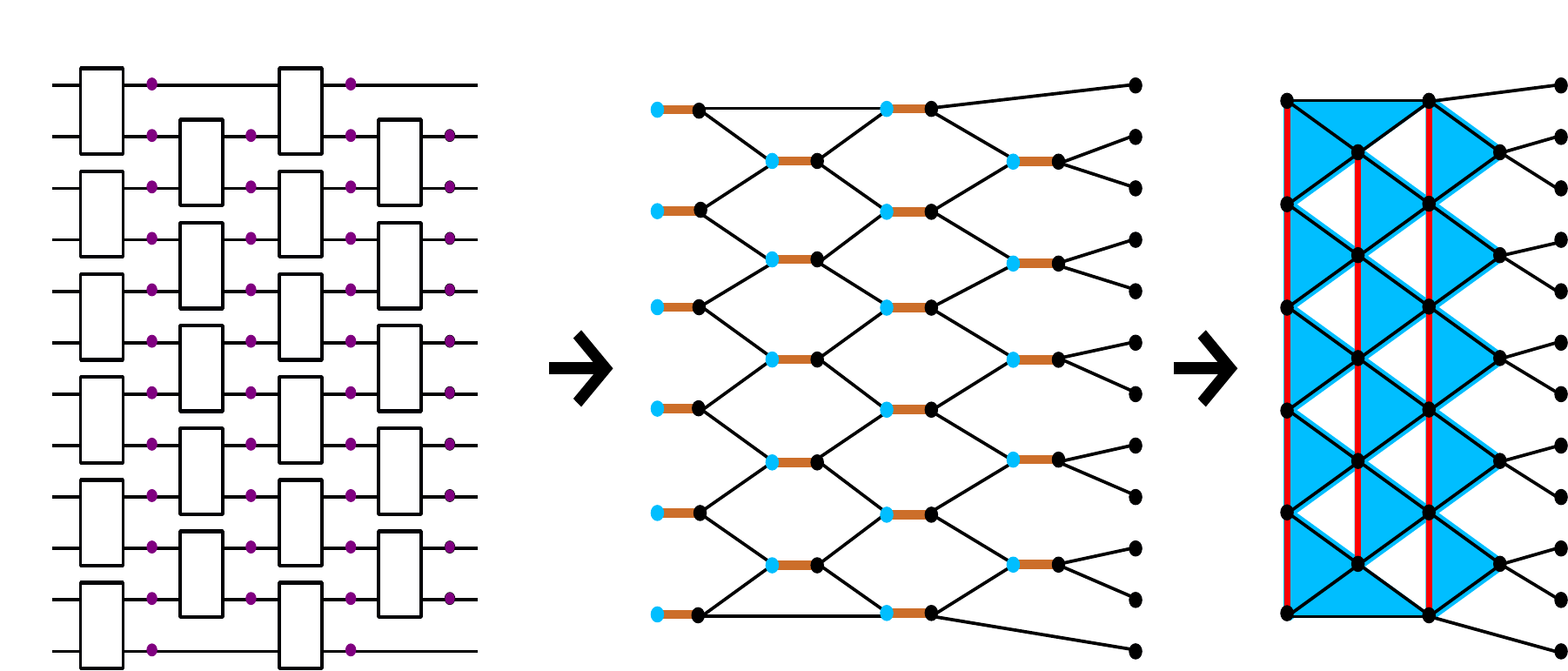
    \caption{ \label{fig:weakmeasurementmapping} Summary of series of maps for Haar-random 1D circuits with weak measurements. (a) The quantum circuit diagram for the unitary plus weak measurement model consists of layers of Haar-random two-qudit gates followed by layers of weak measurements on every qudit, indicated by purple dots. (b) The stat mech mapping results in a model on the honeycomb lattice, where horizontal orange links have weight given by the Weingarten function and diagonal black links have weight that depends on the weak measurement. Blue dots and black dots represent incoming and outgoing nodes, respectively. (c) By decimating the incoming (blue) nodes in the honeycomb lattice, we reduce the number of nodes by half and generate a model with three-body interactions living on rightward-pointing triangles, shaded in blue. When $k=2$ the weights are all positive, and the three-body interaction can be decomposed into an anti-ferromagnetic interaction along vertical (red) links and ferromagnetic interactions along diagonal (dark blue) links.}
\end{figure}

\paragraph{Mapping to the honeycomb lattice.}
{
Let us assume our circuit has $n$ qudits of local dimension $q$ arranged on a line with open boundary conditions. A circuit of depth $d$ acts on the qudits where each layer consists of nearest-neighbor two-qudit Haar-random unitaries. In between each layer of unitaries, a weak measurement is performed on every qudit, described by the set $\mathcal{M}$ of measurement operators and a probability distribution $\mu$ over $\mathcal{M}$. The first step of the stat mech mapping is to replace each Haar-random unitary with a pair of nodes and connect these nodes according to the order of the unitaries acting on the qudits. The second step is to introduce a new auxiliary node for each qudit and connect each outgoing node within the final layer of unitaries to the corresponding pair of auxiliary nodes. The resulting graph is the honeycomb lattice, as shown in Figure \ref{fig:weakmeasurementmapping}(b). We now review what the interactions are on this graph. The horizontal links in Figure \ref{fig:weakmeasurementmapping}(b) host interactions that contribute a weight equal to the Weingarten function. When $k=2$, the interaction depends only on if the pair of nodes agree ($\sigma_u\tau_u^{-1} = e$) or if they disagree ($\sigma_u\tau_u^{-1} = (12)$). In this case the interactions are given explicitly by
\begin{equation}
\text{weight}(\langle s_u t_u \rangle) = \Wg(\sigma_u \tau_u^{-1},q^2)=
\begin{cases}
    \frac{1}{q^4-1} & \text{if }\sigma_u \tau_u^{-1}=e\\
    -\frac{1}{q^2(q^4-1)} &\text{if }\sigma_u \tau_u^{-1} = (12)
\end{cases}
\end{equation}
Meanwhile, the diagonally oriented links in Figure \ref{fig:weakmeasurementmapping}(b) host interactions that depend on the details of the weak measurement being applied in between each layer of unitaries, which we now define.
}

\paragraph{Weak measurement and diagonal weights.}
{
The weak measurement we choose is given as follows. For a fixed $q \times q$ unitary matrix $U$, define
\begin{equation}
    M^{(m)}_U := \sqrt{q} \cdot \text{diag}(U_{m,\cdot})
\end{equation}
that is, the $q \times q$ matrix whose diagonal entries are given by the $m$th row of $U$, scaled by a factor of $\sqrt{q}$, and whose off-diagonal entries are 0. Define the probability distribution $\mu_U$ be the uniform distribution over the set $\mathcal{M}_U = \{M^{(m)}_U\}_{m=1}^{q}$. We can see that ($\mathcal{M}_U,\mu_U$) forms a valid weak measurement since
\begin{equation}
    \sum_{m=1}^{q} \mu_U(m) (M^{(m)}_U)^\dagger M^{(m)}_U = \sum_{m=1}^{q} \text{diag}(\lvert U_{m,\cdot}\rvert^2) = \mathbb{I}_q
\end{equation}
where the last equality follows from the fact that the sum of the squared norms of the entries within a column of a unitary matrix is 1. When $U=\mathbb{I}_q$, the measurement operator $M^{(m)}_U$ is a projector onto the $m$th basis state (scaled by a factor of $\sqrt{q}$), and the weak measurement is simply a projective measurement onto the computational basis.

The weak measurement that we consider here will be a mixture of the weak measurement $(\mathcal{M}_U, \mu_U)$ for different $U$. Formally, we take $\mathcal{M} = \cup_{U \in U(q)} \mathcal{M}_U$. We let the distribution $\mu$ over $\mathcal{M}$ be the distribution resulting from drawing $U$ according to the Haar measure, and then drawing $M$ from $\mathcal{M}_U$ uniformly at random.

This weak measurement is seen to exactly reproduce the weak measurement of \texttt{SEBD} acting on \texttt{CHR} in \Cref{RandomEffective} when $q=2$, where the measurement operators were the diagonal matrices
\begin{subequations}\label{eq:M0M1statmech}
  \begin{align}
M^{(1)} &:= \mqty(\cos(\theta/2) & 0 \\ 0 & e^{-i \phi} \sin(\theta/2))  \\
M^{(2)} &:= \mqty(\sin(\theta/2) & 0 \\ 0 & e^{i \phi} \cos(\theta/2)).
  \end{align}
\end{subequations}
with angles $(\theta,\phi)$ drawn according to the Haar measure on the sphere. Indeed, even for $q\neq2$, this weak measurement arises from a natural generalization of the $\texttt{CHR}$ problem, where one makes Haar-random measurements on a cluster state of higher local dimension, which is created by applying a generalized Hadamard gate to each qudit followed by a generalized $CZ$ gate on each pair of neighboring qudits on the 2D lattice.

To compute the weights on the edges of the stat mech model for $k=2$, we apply the formula in Eqs.~\eqref{eq:weightsu1tu2} and \eqref{eq:weightsuxa}.
\begin{align}
\text{weight}(\langle s_{u_1}t_{u_2} \rangle) &= \int_{U(q)}dU \sum_{m=1}^{q}\frac{1}{q}\tr\left({\left((M^{(m)}_{U})^\dagger M^{(m)}_{U}\right)}^{\otimes 2} W_{\sigma_{u_1}\tau_{u_2}^{-1}}\right) \\
&=\int_{U(q)}dU q\sum_{m=1}^{q}
\begin{cases}
    \text{tr}\left(\text{diag}(\lvert U_{m,\cdot} \rvert^2)\right)^2 & \text{if } \sigma_{u_1}\tau_{u_2}^{-1} = e\\
    \text{tr}(\text{diag}\left(\lvert U_{m,\cdot} \rvert^4)\right) & \text{if } \sigma_{u_1}\tau_{u_2}^{-1} = (12)
\end{cases}
\\
&=
\begin{cases}
    q^2 & \text{if } \sigma_{u_1}\tau_{u_2}^{-1} = e\\
    q^2 \cdot w & \text{if } \sigma_{u_1}\tau_{u_2}^{-1} = (12)
\end{cases}
\end{align}
where
\begin{align}
    w &:= \int_{U(q)}dU \sum_m \frac{1}{q} \tr(\text{diag}(\lvert U_{m,\cdot} \rvert^4)) = q\int_{U(q)}dU \lvert U_{0,0} \rvert^4 \\
    &= q\sum_{\sigma,\tau \in S_2} \Wg(\sigma\tau^{-1},q)=2q\sum_{\sigma \in S_2}  \Wg(\sigma,q)= 2q\left( \frac{1}{q^2-1}-\frac{1}{q(q^2-1)}\right)\\
    &=\frac{2}{q+1},
\end{align}
where in the second line we have invoked the Haar integration formula that appears in Eq.~\eqref{eq:haarintegration} later in Appendix \ref{app:statmechdetails}, and then substituted the explicit values for the Weingarten function when $k=2$. The formula for $\text{weight}(\langle s_u x_a \rangle)$ is given similarly.

We can see that for all $q>1$, the weight is larger when the values of the nodes agree than when they disagree, indicating a ferromagnetic Ising interaction. Indeed, the interaction for $k=2$ will be ferromagnetic regardless of what weak measurement $M$ is made since $\tr(M^\dagger M)^2 \geq \tr((M^\dagger M)^2)$ holds for all $M$. Furthermore, for our choice of weak measurement, the ferromagnetic Ising interaction becomes stronger as $q$ increases.

}

\paragraph{Eliminating negative weights via decimation when $k=2$.}
{
The possibility of a negative weight on the horizontal edges of the honeycomb lattice in Figure \ref{fig:weakmeasurementmapping}(b) appears to impede further progress in the analysis since the classical model cannot be viewed as a physical system with real interaction energies at a real temperature. For $k=2$, this problem may be circumvented by decimating half of the spins; that is, we explicitly perform the sum over $\{\tau_u\}_u$ in the partition function in Eq.~\eqref{eq:partitionfunction}, yielding a new stat mech model involving only the outgoing nodes $s_u$. Since the decimated incoming nodes (except for those in the first layer) each have three neighbors, all three of which are undecimated outgoing nodes, the new model will have a three-body interaction between each such trio of nodes. We may furthermore observe that, for our choice of weak measurement when $k=2$, the three-body weight may be re-expressed as the product of three two-body weights acting on the three edges of the triangle. Below we give formulas for the two-body weights; our formulas are a unique decomposition of the three-body interaction up to a shifting of overall constant factors from one link to another. Thus, via decimation we have moved from the honeycomb lattice with two-body interactions to the triangular lattice with two-body interactions, as illustrated in Figure \ref{fig:weakmeasurementmapping}(c). There are two kinds of two-body interactions on this triangular lattice.  Vertically oriented links between nodes $s_{u_1}$ and $s_{u_2}$ host anti-ferromagnetic interactions
\begin{align}\label{eq:AFMweightvertical}
    \text{weight}(\langle s_{u_1}s_{u_2} \rangle) =
    \begin{cases}
    \frac{1}{q^4-1} & \text{if } \sigma_{u_1}\sigma_{u_2} =e\\
    \frac{w}{1+q^2}\left((q^2-w^2)(q^2w^2-1)\right)^{-1/2}  & \text{if } \sigma_{u_1} \sigma_{u_2} = (12)
    \end{cases}
\end{align}
and diagonally oriented links host ferromagnetic interactions, where
\begin{align}\label{eq:FMweightdiagonal}
    \text{weight}(\langle s_{u_1}s_{u_2} \rangle) =
    \begin{cases}
    q\sqrt{q^2-w^2} & \text{if } \sigma_{u_1}\sigma_{u_2} =e\\
    q\sqrt{w^2q^2-1} & \text{if } \sigma_{u_1} \sigma_{u_2} = (12).
    \end{cases}
\end{align}
For all values of the measurement strength $p$, the ferromagnetic interactions are stronger than the anti-ferromagnetic interaction.
}

\paragraph{Phase diagram.}
{
The model described above for $k=2$ is exactly the anisotropic Ising model on the triangular lattice. In general this model may be described by its energy functional
\begin{equation}
E/kT = -J_1 \sum_{\langle ij \rangle_1 }g_ig_j-J_2 \sum_{\langle ij \rangle_2 }g_ig_j-J_3 \sum_{\langle ij \rangle_3 }g_ig_j
\end{equation}
where $g_i \in \{+1,-1\}$ are Ising spin variables and the three sums are over links along each of the three triangular axes. This model has been studied and its phase diagram is well understood \cite{houtappel1950order,stephenson1970ising}.  In the setting where along two of the axes the interaction strength is equal in magnitude and ferromagnetic, while along the third axis it is weaker in magnitude and antiferromagnetic, the model is known to experience a phase transition as the temperature is varied. At high temperatures, it is in the disordered phase; in other words, samples drawn from the thermal distribution exhibit exponentially decaying correlations between spin values $\sigma_u$ with a constant correlation length of $\xi$. At low temperatures, it is in an ordered phase where samples exhibit long-range correlation. At the critical point, the interaction strengths satisfy the equation \cite{houtappel1950order,stephenson1970ising}
\begin{equation}
    \sinh(2J_1)\sinh(2J_2) + \sinh(2J_2)\sinh(2J_3) +  \sinh(2J_1)\sinh(2J_3) = 1.
\end{equation}
For us, parameter $q$ plays the role of the temperature, and the interaction strengths, derived from Eqs.~\eqref{eq:AFMweightvertical} and \eqref{eq:FMweightdiagonal}, are given by
\begin{align}
    J_1=J_2 &= \frac{1}{4}\log\left(\frac{q^2-w^2}{w^2q^2-1}\right), \label{eq:J1}\\
    J_3 &= -\frac{1}{2}\log\left(\frac{w(q^2-1)}{\sqrt{(q^2-w^2)(q^2w^2-1)}}\right). \label{eq:J3}
\end{align}
Using these equations, we can solve for the critical point, and we find it to be $q_c = 3.249$. Only integer values of $q$ correspond to valid quantum circuits, so we conclude that the model is disordered when $q=2$ or $q=3$ and ordered when $q \geq 4$. We plot this one dimensional phase diagram in Figure \ref{fig:weakmeasurementphasediagram}.

\begin{figure}[ht]
    \centering
    \includegraphics[width=0.6\columnwidth]{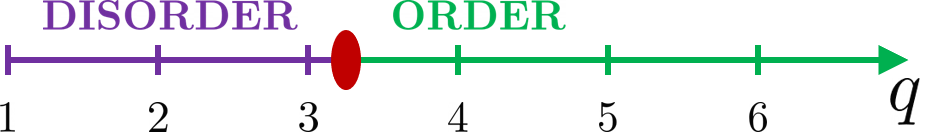}
    \caption{\label{fig:weakmeasurementphasediagram} Phase diagram showing for which values of $q$ the anisotropic Ising model on the triangular lattice is ordered and disordered. The critical point, indicated by the red dot, occurs at $q_c=3.249$.}
\end{figure}

\paragraph{Connection between (dis)order and scaling of entanglement entropy.} We expect the scaling of the quantity $\tilde{S}_2 = F_{2,A}-F_{2,\emptyset} = -\log(\E_U(Z_{2,A})/ \E_U(Z_{2,\emptyset}))$ to be related to the order or disorder of the model by the following argument. For $\E_U(Z_{2,\emptyset})$, the auxiliary spins are all set to $\chi_a = e$, biasing the bulk spins nearby to prefer $e$ over $(12)$. For $\E_U(Z_{2,A})$, the spins within the region $A$ are twisted so that $\chi_a = (12)$, introducing a domain wall at the boundary. In the ordered phase, the bias introduced at the boundary extends throughout the whole bulk since there is no decay of correlation with distance. The domain wall at the boundary in the calculation of $\E_U(Z_{2,A})$ forces the bulk to separate into two regions with distinct phases separated by a domain wall that cuts through the bulk. The domain wall has length of order $\min(|A|,d)$ where $|A|$ is the number of sites in region $A$ and $d$ is the depth. In the calculation  of $\E_U(Z_{2,\emptyset})$, there is no domain wall. The addition of one additional unit of domain wall within a configuration leads the weight of the configuration to decrease by a constant factor, so in the ordered phase we expect $-\log(\E_U(Z_{2,A})/\E_U(Z_{2,\emptyset})) = O(\min(|A|,d))$. Meanwhile, in the disordered phase, there is a natural length scale $\xi$ that boundary effects will penetrate into the bulk. The domain wall at the boundary due to twisted boundary conditions will be washed out by the bulk disorder after a distance on the order of $\xi = O(1)$. Thus we expect $-\log(\E_U(Z_{2,A})/\E_U(Z_{2,\emptyset})) = O(1)$. A cartoon illustrating this logic appears later in Figure \ref{fig:2Dcircuitsdomainwalls} when we consider the stat mech mapping applied to 2D circuits.  For further discussion of the connection between order-disorder properties of the stat mech model and entropic properties of the underlying quantum objects, see \cite{vasseur2018entanglement,bao2019theory,jian2019measurement}.

This logic suggests that, if we take the scaling of $\tilde{S}_2$ to be a good proxy for the scaling of $\langle S_2 \rangle$, the disorder-to-order phase transition in the classical model would be accompanied by an area-law-to-volume-law phase transition in the R\'{e}nyi-2 entropy of the output of random circuits.
}

\paragraph{Relationship to numerical simulation of \texttt{SEBD} on \texttt{CHR}.}
{
In Section \ref{se:CHR}, with fixed $q=2$, it was established that the effective dynamics of \texttt{SEBD} running on \texttt{CHR} are alternating layers of entangling two-qubit $CZ$ gates and weak measurements on every qubit of a 1D line, where the form of the weak measurement is given explicitly. The dynamics we have studied in this section use the same weak measurement, but choose the two-qubit entangling gates to be Haar-random. We have established that the quasi-2-entropy $\tilde{S}_2$ satisfies an area law for this process when $q=2$, and the statement remains true for $q=3$ when the weak measurement corresponds to a natural generalization of the \texttt{CHR} problem to larger local dimension. For $q=4$, it is no longer true; the dynamics of $\tilde{S}_2$ satisfy a volume law.

Due to the similarity between the dynamics studied in this section and that of \texttt{SEBD} running on \texttt{CHR}, our conclusion provides a partial explanation for the numerical observation presented in Section \ref{se:numerics} that the average entanglement entropy $\langle S_k \rangle$ satisfies an area law when \texttt{SEBD} runs on \texttt{CHR} for $q=2$ and various values of $k$.
}

\paragraph{Additional observations appearing in previous work.} The above analysis is essentially a restatement of what appears in recent works by Bao, Choi, and Altman \cite{bao2019theory} and separately Jian, You, Vasseur, and Ludwig \cite{jian2019measurement}, except that here we analyzed a different weak measurement. In particular, \cite{bao2019theory} considered the case where a projective measurement occurs with some probability $p$ on each qudit after each layer of unitaries, and otherwise there is no measurement. They made the observation that we describe above that the $k=2$ mapping can be written as a 2-body anisotropic Ising model on the triangular lattice with an exact solution. Both of these papers went beyond what we have presented here to analyze the $q \rightarrow \infty$ limit directly, where they observed that the stat mech model becomes a standard ferromagnetic Potts model on the square lattice for all integers $k$. For $k=2$ this is exactly the square lattice Ising model and indeed, we can see from Eq.~\eqref{eq:J3} that when $q \rightarrow \infty$, $J_3 \rightarrow 0$; the anti-ferromagnetic links along one axis vanish leaving a square lattice with exclusively ferromagnetic interactions. The fact that the model becomes tractable for all integers $k \geq 2$ allows these papers to invoke analytic continuation and make sense of the $k \rightarrow 1$ limit, where the quasi-entropy $\tilde{S}_k$ exactly becomes the expected von Neumann entropy $\langle S \rangle$.

\section{Evidence for efficiency of algorithms from statistical mechanics}\label{se:statmech2D}

The efficiency of both algorithms hinges on the behavior of certain entropic quantities of the quantum state produced by the quantum circuit. In the \texttt{SEBD} algorithm, the bond dimension of the 1D MPS is truncated at each step of the MPS evolution, introducing small error but keeping the runtime of the algorithm efficient (i.e. polynomial in $n$) so long as the entanglement entropy across any cut of the MPS obeys an area law\footnote{Technically, an area law for the $k$-R\'{e}nyi entropy with $0 < k < 1$ is needed to fully justify algorithmic efficiency \cite{schuch2008entropy}.}. In the \texttt{Patching} algorithm, the classical conditional mutual information $I(A:C|B)$ of the joint distribution of measurement outcomes must decay exponentially in the size of region $B$, where $A$, $B$, and $C$ are regions of the lattice with $B$ separating $A$ from $C$. In this section, we examine the behavior of these entropic quantities for typical circuit instances by employing the mapping introduced in Section \ref{se:statmech}.

While the stat mech mapping falls short of providing a fully rigorous justification of the efficiency of the algorithms, it is the most promising tool we are aware of to analytically understand the behavior of the \texttt{SEBD} and \texttt{Patching} algorithm when they are running on circuits made from Haar-random gates. At the very least, the mapping provides important intuition for why these algorithms would be efficient in certain cases and not efficient in others. Furthermore, it lays the groundwork for what could constitute a more rigorous justification. The main roadblock is the fact that the stat mech mapping gives an expression for the ``quasi-entropy,'' which is only exactly equal to the average entanglement entropy in a limit that cannot be directly accessed by the mapping. From a practical perspective, the stat mech model could provide a way to numerically estimate some of these quasi-entropies, which could possibly aid in predicting the runtime and error of the algorithms in regimes where it may not be easy to tell by running the algorithm itself.

Modulo the concern of quasi-entropy vs.~entropy, the stat mech mapping predicts the following about the efficiency of the algorithms:

\begin{enumerate}
    \item For 2D circuits with nearest-neighbor Haar-random gates of sufficiently small depth $d$ and sufficiently small local dimension $q$, we expect both algorithms to be efficient in the system size. That is, they produce samples from a distribution with $\varepsilon$ total variation distance from the ideal distribution on a fraction $1-\delta$ of circuit instances, and they have runtime $\poly(L_1, L_2, 1/\varepsilon,1/\delta)q^{\poly(d)}$. In particular, this is true for depth $d=3$ and local dimension $q=2$ for the uniform ``brickwork'' architecture where exact simulation is known to be hard in the worst case, under plausible complexity theoretic assumptions.
    \item For a fixed depth $d$, we expect the algorithms to become inefficient once the local dimension exceeds some critical value $q_c$.
    \item For a fixed local dimension $q$, we expect the algorithms to become inefficient once the depth exceeds some critical value $d_c$.

\end{enumerate}

The efficient-to-inefficient transition in computational complexity is related to a disorder-to-order phase transition in the classical model. Item 1 above should be regarded as evidence for \Cref{con:efficient}, while items 2 and 3 are evidence for \Cref{con:transition}.

\subsection{Mapping applied to general 2D circuits} \label{se:2dcircuitmapping}

We now apply the mapping described in Section \ref{se:mapping} directly to a depth-$d$ circuit acting on a $\sqrt{n} \times \sqrt{n}$ lattice of qudits consisting of nearest-neighbor two-qudit Haar-random gates, but without measurements in between the gates. This is the relevant case for the algorithms presented in this paper. In this section, we will assume for concreteness that the first unitary layer includes gates that act on qudits at gridpoints $(i,j)$ and $(i,j+1)$ for all odd $i$ and all $j$, the second layer on $(i,j)$ and $(i,j+1)$ for all even $i$ and all $j$, the third layer on $(i,j)$ and $(i+1,j)$ for all $i$ and all odd $j$, and the fourth layer on $(i,j)$ and $(i+1,j)$ for all $i$ and all even $j$. Subsequent layers then cycle through these four orientations.

\paragraph{The model resulting from the stat mech mapping.}
{
Replacing the unitaries in the circuit diagram with pairs of nodes and connecting them as described in Section \ref{se:mapping} yields a graph embedded in three dimensions. The nodes in this graph still have degree three, so locally the graph looks similar to the honeycomb lattice, but globally the nodes form a 3D lattice that can be viewed roughly as a $\sqrt{n} \times \sqrt{n} \times d$ slab, although the details of how these nodes connect is not straightforward to visualize. We have included pictures of the graph in Figure \ref{fig:2Dcircuitgraph}.

\begin{figure}[h]
    \centering
    \subfloat[Depth-4 circuit on $4 \times 4$ lattice]{{\includegraphics[width=0.45\textwidth]{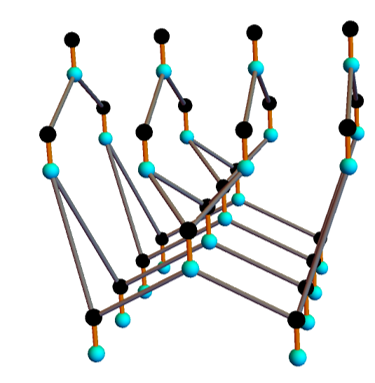} }}%
	\qquad
	\subfloat[Depth-5 circuit on $28 \times 28$ lattice]{{\includegraphics[width=0.45\textwidth]{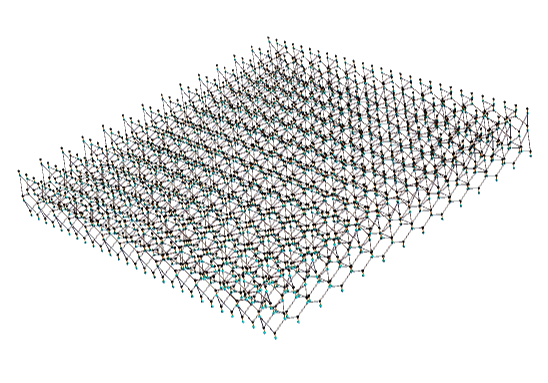} }}%
    \caption{\label{fig:2Dcircuitgraph}
    The graph produced by the stat mech mapping on shallow 2D circuits. (a) A close up view of the graph reveals that the degree of most nodes is three, similar to the honeycomb lattice. (b) A far-away view reveals that globally the graph looks like a two dimensional slab of thickness roughly $d$.
    }
\end{figure}
As before, edges between nodes originating from the same unitary are assigned a weight equal to the Weingarten function. Since there is no measurement between unitaries, we may take the measurement operator applied on the edges that connect successive unitaries to be the identity, and we find that the weight contributed by these edges is $\text{weight}(\langle s_{u_1} t_{u_2} \rangle) = q^{C(\sigma_{u_1}\tau_{u_2}^{-1})}$. For $k=2$ this amounts to a ferromagnetic Ising interaction where
\begin{equation}
    \text{weight}(\langle s_{u_1}t_{u_2} \rangle) =
    \begin{cases}
    q^2 & \text{if } \sigma_{u_1}\tau_{u_2} =e\\
    q & \text{if } \sigma_{u_1} \tau_{u_2} = (12).
    \end{cases}
\end{equation}

To analyze the output state, we will divide the $n$ qudits into three groups $A$, $B$, and $C$. We suppose that after the $d$ unitary layers have been performed, a projective measurement is performed on the qudits in region $B$. Qudits in regions $A$ and $C$ are left unmeasured and we wish to calculate quantities like $\E_U(Z_{k,\emptyset/A})$. The mapping calls for us to introduce an auxiliary node for each qudit in the circuit. However, the projective measurement of qudits in the region $B$ effectively isolates auxiliary nodes associated with qudits in region $B$; the edge connecting it to one of the bulk nodes carries a weight that is constant across all bulk configurations. Thus we may equivalently omit the introduction of auxiliary nodes in region $B$ as well as any edges that would be connected to them.
}

\paragraph{Eliminating negative weights via decimation when $k=2$.}
{
The quantities $\E_U(Z_{k,\emptyset/A})$ are now given by classical partition functions on this graph with appropriate boundary conditions for the auxiliary nodes in regions $A$ and $C$. We wish to understand whether this stat mech model is ordered or disordered. Again, we are faced with the issue that the Weingarten function can take negative values and thus some configurations over this graph could have negative weight. As in the case of 1D circuits with weak measurements, we can partially rectify this by decimating all the incoming nodes. The resulting graph has half as many nodes and interactions between groups of three adjacent nodes. In the case of 1D circuits with weak measurements when $k=2$, we could decompose this three-body interaction into three two-body interactions. If we try to do the same here, we find the two-body interactions have infinite strength. Instead, we work directly with the three-body interaction between nodes $s_{u_1}$, $s_{u_2}$, and $s_{u_3}$, where unitary $u_3$ acts after $u_1$ and $u_2$, for which there is a simple formula for the weights:
\begin{equation}
    \text{weight}(\langle s_{u_1}s_{u_2}s_{u_3} \rangle) =
    \begin{cases}
    1 & \text{if } \sigma_{u_1} = \sigma_{u_2} = \sigma_{u_3} \\
    \frac{1}{q+q^{-1}} & \text{if } \sigma_{u_2} \neq \sigma_{u_3} \\
    0 & \text{if } \sigma_{u_1} \neq \sigma_{u_2} = \sigma_{u_3}.
    \end{cases}
\end{equation}

Now, all the weights are non-negative. Moreover, the largest weight occurs when all the nodes agree, indicating a generally ferromagnetic interaction between the trio of nodes. If one of the bottom two nodes disagrees with the other two, the weight is reduced by a factor of $q+1/q$. When the top node disagrees, the weight is 0; these configurations are forbidden and contribute nothing to the partition function.

}

\paragraph{Allowed domain wall configurations and disorder-order phase transitions.}
{
Using this observation, we can understand the kinds of domain wall structures that will appear in configurations that contribute non-zero weight. In this setting, domain wall structures are membrane-like since the graph is embedded in 3D. Membranes that have upward curvature, shaped like a bowl, are not allowed, because somewhere there would need to be an interaction where the upper node disagrees with the two below it. On the other hand, cylindrically shaped domain wall membranes do not have this issue, nor do dome-shaped membranes. The weight of a configuration is reduced by a factor of $q+1/q$ for each unit of domain wall, an effect that acts to minimize the domain wall size when drawing samples from the thermal distribution (energy minimization). On the other hand, larger domain walls have more configurational entropy --- there are many ways to cut through the graph with a cylindrically shaped membrane --- an effect that acts to bring out more domain walls in samples from the thermal distribution (entropy maximization). The question is, which of these effects dominates? For a certain setting of the depth $d$ and local dimension $q$, is there long-range order, or is there exponential decay of correlations indicating disorder? Generally speaking, increasing depth magnifies both effects: cylindrical domain wall membranes must be longer --- meaning larger energy --- when the depth is larger; however, longer cylinders also have more ways of propagating through the graph. Meanwhile, increasing $q$ only magnifies the energetic effect since it increases the interaction strength and thus the energy cost of a domain wall unit but leaves the configurational entropy unchanged.

Thus, in the limit of large $q$ we expect the energetic effect to win out and the system to be ordered for any circuit depth $d$ and any circuit architecture. What about small $q$? Physically speaking, $q$ must be an integer at least 2 since it represents the local Hilbert space dimension of the qudit. However, the statistical mechanical model itself requires no such restriction, and we can allow $q$ to vary continuously in the region $[1,\infty)$. Then for $q \rightarrow 1$, the energy cost of one unit of domain wall becomes minimal (but it does not vanish). Depending on the exact circuit architecture and the depth of the circuit, the system may experience a phase transition into the disordered phase once $q$ falls below some critical threshold $q_c$. The depth-3 circuit with brickwork architecture that we present later in Section \ref{sec:brickwork} provides an example of such a transition. It is disordered when $q=2$ and experiences a phase transition as $q$ increases to the ordered phase at a transition point we estimate to be roughly $q_c \approx 6$.

When $q$ is fixed and $d$ is varied, it is less clear what to expect. Suppose for small $d$, the system is disordered. Then increasing $d$ will amplify both the energetic and entropic effects, but likely not in equal proportions. If the amplification of the energetic effect is stronger with increasing depth, then we expect to transition from the disordered phase to the ordered phase at some critical value of the depth $d_c$. Without a better handle on the behavior of the stat mech model, we cannot definitively determine if and when this depth-driven phase transition will happen.

However, we have other reasons to believe that there should be a depth-driven phase transition. In particular, we now provide an intuitive argument for why a disorder-order transition in the parameter $q$ should imply a disorder-order transition in the parameter $d$. Consider fixed $d$, and another fixed integer $r\geq 1$ such that $d/r \gg 1$. We may group together $r \times r$ patches of qudits to form a ``supersite'' with local dimension $q^{r^2}$. Similarly, we may consider a ``superlayer'' of $O(r)$ consecutive unitary layers. Since $O(r)$ layers is sufficient to implement an approximate unitary $k$-design on a $r\times r$ patch of qudits (taking $k=O(1)$) \cite{HM18}, we intuitively take each superlayer to implement a Haar-random unitary between pairs of neighboring supersites. Thus, a depth-$d$ circuit acting on qudits of local dimension $q$ is roughly equivalent to a depth-$O(d/r)$ circuit acting on qudits of local dimension $q^{r^2}$ in the supersite picture. If for a fixed $d$, we observe a disorder-order phase transition for increasing $q$, then for fixed $q$ and fixed $d/r$, we should also observe a disorder-order phase transition with increasing $r$. Equivalently, we should see a transition for fixed $q$ and increasing $d$. This logic is not perfect because superlayers do not exactly map to layers of Haar-random two-qudit gates between neighboring supersites, but nonetheless we take it as reason to expect a depth-driven phase transition.

}

\subsection{Efficiency of \texttt{SEBD} algorithm from stat mech}\label{se:statmechsebd}

The efficiency of the \texttt{SEBD} algorithm relies on the error incurred during the MPS compression being small. If the inverse error has a polynomial relationship (or better) with the bond dimension of truncation, then the algorithm's time complexity is polynomial (or better) in the inverse error and the number of qudits. This will be the case if the MPS prior to truncation satisfies an area law for the $k$-R\'{e}nyi entropy for some $0 < k < 1$. The stat mech mapping is unable to probe these values of $k$. However, we hypothesize that the behavior of larger values of $k$ is indicative of the behavior for $k < 1$ since the examples where the $k$-R\'{e}nyi entropy with $k \geq 1$ satisfies an area law but efficient MPS truncation is not possible require contrived spectrums of Schmidt coefficients. Although some physical processes give rise to situations where the von Neumann and $k$-R\'{e}nyi entropies with $k>1$ exhibit different behavior (see e.g.~\cite{yichen-renyi}, which showed that for random 1D circuits without measurements but with the unitaries chosen to commute with some conserved quantity, after time $t$ the entropy is $O(t)$ for $k=1$ but $O(\sqrt{t\log t})$ for $k>1$), the numerical evidence we gave in \Cref{se:numerics}, where the scaling of all the $k$-R\'{e}nyi entropies appears to be the same, suggests our case is not one of these situations.

In Section \ref{se:weakmeasurementmapping} we discussed how for 1D circuits with alternating unitary and weak measurement dynamics, there has been substantial numerical evidence for a phase transition from an area-law phase to a volume-law phase as the parameters of the circuit are changed. We also reviewed previous work applying the stat mech model to this setting, and explicitly showed how there is a $q$-driven phase transition from an area law to volume phase for the quasi-entropy for a specific choice of weak measurement. The \texttt{SEBD} algorithm simulating a 2D circuit of constant depth made from Haar-random gates may be viewed as a system with very similar dynamics --- an alternation between entanglement-creating unitary gates and entanglement-destroying weak measurements. However, none of the unitary-and-measurement models that have been previously studied capture the exact dynamics of \texttt{SEBD}, one reason being that \texttt{SEBD} tracks the evolution of several columns of qudits at once (recall it must include all qudits within the lightcone of the first column). The Haar-random unitaries create entanglement within these columns of qudits, but not in the exact way that entanglement is created by Haar-random nearest-neighbor gates acting on a single column. Nonetheless, we expect the story to be the same for the dynamics of \texttt{SEBD} since the main findings of studies of these unitary-and-measurement models have been quite robust to variations in which unitary ensembles and which measurements are being implemented; we expect that varying parameters of the circuit architecture like $q$ and $d$ can lead to entanglement phase transitions, and thus transitions in computational complexity.

Indeed, the discussion from the previous section suggests precisely this fact. When we apply the stat mech mapping directly to 2D circuits instead of to 1D unitary-and-measurement models, we expect disorder-order phase transitions as both $q$ and $d$ are varied. To make the connection to entanglement entropy explicit here, we note that after $t$ steps of the \texttt{SEBD} algorithm, all $\sqrt{n}$ qudits in the first $t$ columns of the $\sqrt{n} \times \sqrt{n}$ lattice have been measured, and we have an MPS representation of the state on columns $t+1$ through $t+r$, where $r = O(d)$ is the radius of the lightcone (which depends on circuit architecture, but cannot be larger than $d$).  To calculate the entropy of the MPS, we take the region $A$ to be the top half of these $r$ columns, and region $C$ to be the bottom half. Region $B$ consists of the first $t$ columns, which experience projective measurements. The prescription for computing $\E_U(Z_{2,A})/\E_U(Z_{2,\emptyset})$ calls for determining the free energy cost of twisting the boundary conditions in region $A$, which creates a domain wall along the $A:C$ border. If the bulk is in the ordered phase, then this domain wall membrane will penetrate through the graph a distance of $t$, leading to a domain wall area of $O(td)$. If the bulk is in the disordered phase, it will only penetrate a constant distance, on the order of the correlation length $\xi$ of the disordered stat mech model, before being washed out by the disorder, leading to a domain wall area of only $O(\xi d)$. The typical domain wall configurations before and after twisting boundary conditions in the ordered and disordered phases is is reflected in the cartoon in Figure \ref{fig:2Dcircuitsdomainwalls}. As we discussed in Section \ref{se:weakmeasurementmapping}, we expect there to be a correspondence between the scaling of the domain wall size and the free energy cost after twisting the boundary conditions of the stat mech model.

\begin{figure}[h]
    \centering
    \subfloat[Ordered phase]{{\includegraphics[width=0.45\textwidth]{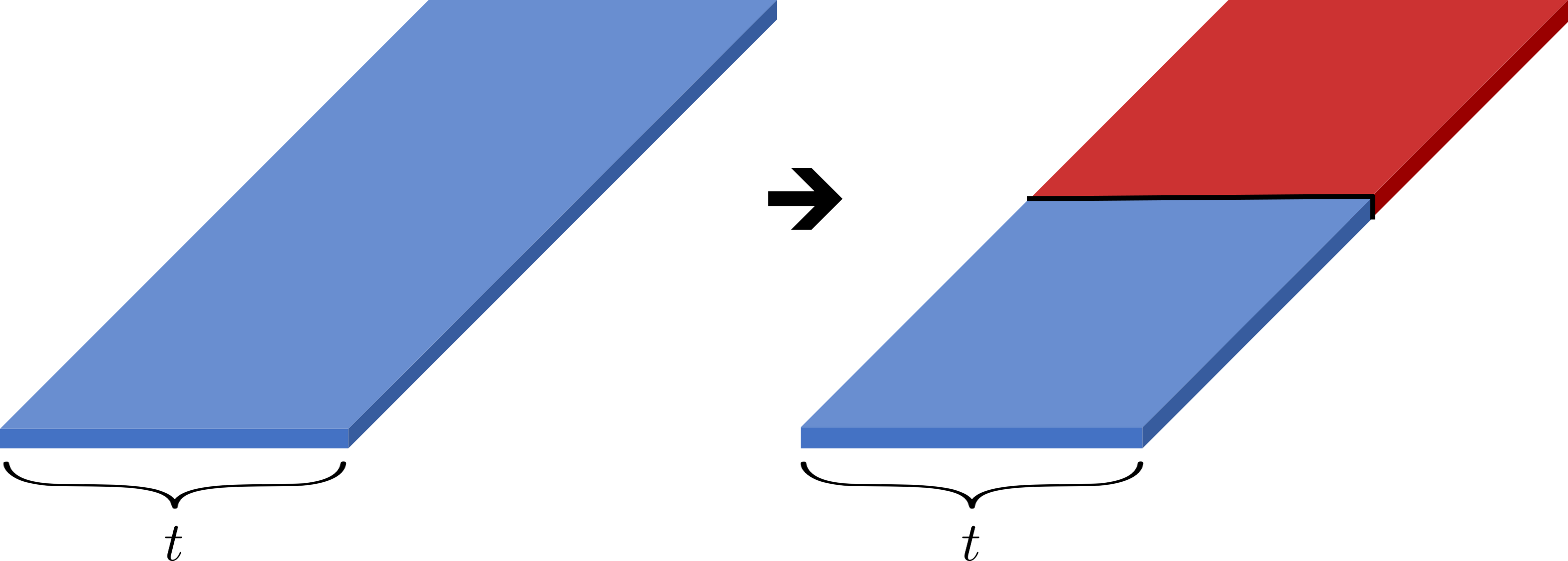} }}%
	\qquad
	\subfloat[Disordered phase]{{\includegraphics[width=0.45\textwidth]{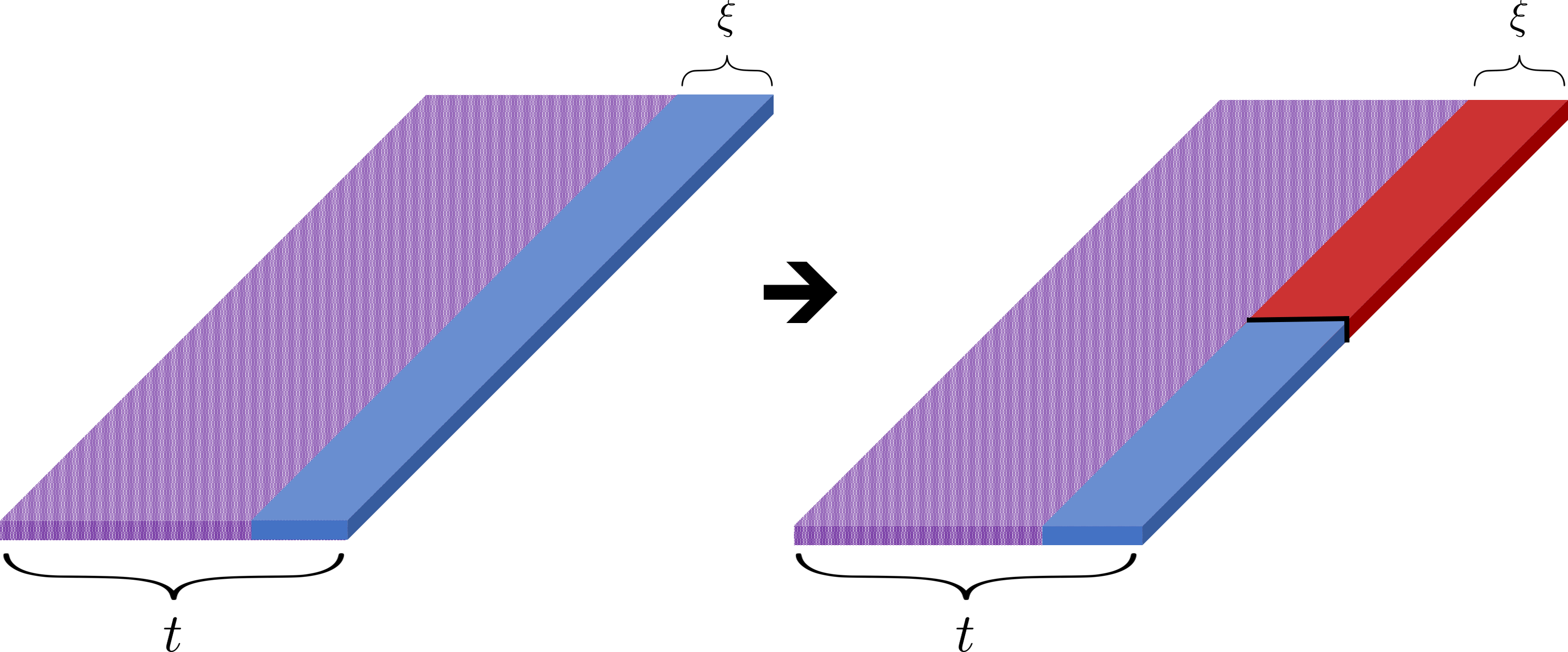} }}
    \caption{\label{fig:2Dcircuitsdomainwalls}  The stat mech mapping yields nodes arranged within a roughly $\sqrt{n} \times t \times d$ prism. (a) In the ordered phase, twisting the boundary conditions at the right boundary introduces a domain wall between the two phases (indicated by red and blue) that propagates through the bulk for a total area of $O(td)$. (b) In the disordered phase, boundary conditions introduce bias that is noticeable only within a constant $O(\xi)$ distance of the boundary, and the domain wall membrane introduced by twisting the boundary conditions is quickly washed out by the bulk disorder (dotted purple). The total area is $O(\xi d)$.}
\end{figure}

This implies that the quasi-entropy $\tilde{S}_2$ is in the area (resp.~volume) law phase when the classical stat mech model is in the disordered (resp.~ordered) phase. Heuristically we might expect the runtime of the \texttt{SEBD} algorithm to scale like $\text{poly}(n)\exp(O(\tilde{S}_2))$, suggesting that the disorder-to-order transition is accompanied by an efficient-to-inefficient transition in the complexity of the \texttt{SEBD} algorithm. Furthermore, near the transition point within the volume-law phase, the quasi-entropy scales linearly with system size but with a small constant prefactor, suggesting that the \texttt{SEBD} runtime, though exponential, could be considerably better than previously known exponential-time techniques.

\subsection{Efficiency of \texttt{Patching} algorithm from stat mech}\label{se:statmechpatching}
We now study the predictions of the stat mech model for the fate of the \texttt{Patching} algorithm we introduced in \Cref{sec:patching}. To do so, we in turn study the predictions of the stat mech model for entropic properties of the \emph{classical} output distribution, as \texttt{Patching} is efficient if the CMI of the classical output distribution is exponentially decaying with respect to shielded regions.

We have previously applied the stat mech model to study expected entropies of quantum states. However, we now wish to study expected entropies of the classical output distribution. To this end, we now consider the non-unitary quantum circuit consisting of the original, unitary circuit followed by a layer of dephasing channels applied to every qudit. The resulting mixed state is classical (i.e., diagonal in the computational basis) and is exactly equal the output distribution we want to study. That is, the state after application of the dephasing channels is $\sum_{\vb{x}} \mathcal{D}(\vb{x}) \dyad{\vb{x}}$ where $\mathcal{D}$ is the output distribution of the circuit. Note that the application of the dephasing channel is not described in the formalism we have discussed previously, but is easily incorporated. In particular, we need to compute the weights between the auxiliary node $x_a$ and the corresponding outgoing node $s_u$ associated with the unitary $u$ that is the last in the circuit to act on qudit $a$. We may update Eq.~\eqref{eq:weightsuxa} (whose original form was derived in Appendix \ref{app:statmechdetails}) and compute the following, letting $\ket{\Phi_k} \equiv \qty(\sum_{i=1}^{q} \ket{i}\otimes \ket{i})^{\otimes k}$.
\begin{align}
	\text{weight}(\expval{s_u x_a}) &= \expval{(I\otimes W_{\sigma^{-1}_u})\qty(\sum_{i=1}^q \dyad{i}\otimes \dyad{i})^{\otimes k} (I\otimes W_{\chi_a})}{\Phi_k} \\
	&= \sum_{i_1, \dots, i_k} \expval{W_{\sigma_u^{-1}}}{i_1, \dots, i_k} \expval{W_{\chi_a}}{i_1, \dots, i_k}.
\end{align}

We therefore see that $\text{weight}(\expval{s_u x_a} )$ in this setting is exactly equal to the number of $k$-tuples of indices $(i_1, \dots, i_k)$ with $i_j \in [q]$ that are invariant under both permutation operators $\sigma_u, \chi_a \in S_k$ acting as $\sigma_u \cdot (i_1, \dots, i_k) = (i_{\sigma(1)}, \dots, i_{\sigma(k)})$.  In fact, for our purposes, the auxiliary spin $\chi_a$ will either be set to the identity $e$ or to the $k$-cycle permutation $(1\dots k )$. In the former case, the weight reduces to $\tr(W_{\sigma_u}) = q^{C(\sigma_u)}$. In the latter case, since the only tuples that are invariant under application of the cycle permutation $(1\dots k )$ are the $q$ tuples of the form $(x,x,\dots,x)$ for $x\in [q]$, the weight is simply $q$ for all $\sigma_u$. Summarizing,
\begin{equation}\label{eq:dephasingWeights}
\text{weight}(\expval{s_u x_a})  =  \begin{cases}
	q^{C(\sigma_u)}, & \chi_a = e \\
	q, & \chi_a = (1\dots k).
\end{cases}
\end{equation}
From these expressions, we may immediately note the following facts. First, flipping some auxiliary spin from $e$ to $(1\dots k)$ cannot increase the weight of a configuration, and hence such a flip corresponds to an increase in free energy. Second, if an auxiliary spin is in the $(1\dots k)$ configuration, then the auxiliary spin may be effectively removed from the system since in this case the contribution of the auxiliary spin to the weight of a configuration is constant across all configurations.

With these modified weights, we may now compute ``quasi-entropies'' $\tilde{S}_k(X)$ as before, where now in the $k\rightarrow 1$ limit $\tilde{S}_k(X)$ approaches the expected Shannon entropy of the marginal of the output distribution on subregion $X$, $\langle S(X)_{\mathcal{D}} \rangle$, where the average is over random circuit instances.

\paragraph{Disordered stat mech model suggests \texttt{Patching} is successful.}
We consider the quasi-CMI defined by
\begin{equation}
\tilde{I}_2(A:C|B) := \tilde{S}_2(AB) + \tilde{S}_2(BC) - \tilde{S}_2(B) - \tilde{S}_2(ABC),
\end{equation}
where all quasi-entropies are taken with respect to the collection of classical output distributions that arise from the quantum circuit architecture. This definition is in analogy to the definition of CMI as $I(A:C|B) = S(AB) + S(BC) - S(B) - S(ABC)$ \cite{cover1991elements}. Note that we may define the quasi-$k$-CMI $\tilde{I}_k(A:C|B)$ analogously for any nonnegative $k$, and it holds that $\langle I(A:C|B)_\mathcal{D} \rangle = {\lim_{k\rightarrow 1} \tilde{I}_k (A:C|B)}$ where the angle brackets denote an expectation over random circuit instances.

Recalling that $\tilde{S}_2(X) = F_{2,X} - F_{2,\emptyset}$, we may rewrite the quasi-2-CMI  as
\begin{equation}
\tilde{I}_2(A:C|B) = (F_{2,AB} - F_{2,B}) - (F_{2,ABC} - F_{2,BC}).
\end{equation}
In stat mech language, the quasi-CMI is essentially the difference in free energy costs of twisting the boundary condition of subregion $A$ in the case where (1) no other spins have boundary conditions, and the case where (2) subregion $C$ also has an imposed boundary condition.

Now, consider some random circuit family $\mathcal{C}$ with associated stat mech model that is in the disordered phase for $k=2$. For any subregion $X$ of qudits, and partition of $X$ into subregions $X = A \cup B \cup C$, we expect this difference between free energy costs will decay exponentially with the separation between $A$ and $C$ as
\begin{equation}
\tilde{I}_2(A:C|B) \leq \poly(n,q) e^{-\dist(A,C)/\xi}
\end{equation}
where $\xi$ is a correlation length. This is because in the disordered phase of the stat mech model, information about the boundary of region $C$ will be exponentially attenuated as the distance from region $C$ grows. If we take $\tilde{I}_2(A:C|B)$ as a proxy for the average CMI of the output distribution, $\langle I(A:C|B)_{\mathcal{D}} \rangle$, we conclude that the random circuit family $\mathcal{C}$ is $\poly(n,q) e^{-\Theta(l)}$-Markov as defined in \Cref{sec:patching}. The results of that section then show that \texttt{Patching} can be used to efficiency sample from the output distribution and estimate output probabilities with high precision and high probability. We take this exponential decay of quasi-2-CMI as evidence that the average CMI also decays exponentially, and therefore that  \texttt{Patching} is successful.

\paragraph{Ordered stat mech model suggests \texttt{Patching} is unsuccessful.}
We first obtain exact, closed form results in the zero-temperature limit of the stat mech model, which corresponds to the $q\rightarrow \infty$ limit. However, we expect that qualitatively similar results hold outside of this limit.

As before, consider the stat mech model obtained by applying dephasing channels to all qudits after the application of all gates. Consider some connected, strict subset $A$ of qudits on the original grid. Suppose we are interested in the quasi-entropy $\tilde{S}_k(A) = (F_{k,A} - F_{k,\emptyset})/(k-1)$ of the output distribution on this region. This quantity is given by the free energy cost of twisting the boundary conditions (auxiliary spins) associated with region $A$ from $e$ to $(1\dots k)$. The auxiliary spins associated with qudits in the complement of $A$ are fixed to be in the identity permutation configuration, $e$. For both sets of boundary conditions, all non-auxiliary spins will order in the configuration $e$. This is because the configuration $e$ maximizes the weights in \Cref{eq:dephasingWeights} for spins connected to auxiliary spins in the configuration $e$, and the weight of a spin connected to an auxiliary spin in the configuration $(1\dots k)$ is independent of that spin's configuration. Hence, regardless of the configuration of the auxiliary spins, all bulk spins are in the identity permutation configuration in the $q\rightarrow \infty$ limit of infinitely strong couplings.

Therefore, twisting a single auxiliary spin from $e$ to $(1\dots k)$ results in a reduction of the total weight by a factor of $q/q^{C(e)} = q/q^k = q^{1-k}$, corresponding to  a free energy increase of $(k-1)\log(q)$. We therefore compute
\begin{equation}
\tilde{S}_k(A) = \frac{F_{k,A} - F_{k,\emptyset}}{k-1} = |A| \log(q).
\end{equation}
Note that this result is exact in the $q\rightarrow \infty$ limit. Notably, we find that all integer quasi-entropies are equal in this limit, and so we may trivially perform the analytic continuation to the von Neumann (i.e. Shannon) entropy:
\begin{equation}
\expval{S(A)} = \lim_{k\rightarrow 1} |A| \log(q) = |A| \log(q).
\end{equation}
Hence, in the $q\rightarrow \infty$ limit, the entropy of a strict subregion of the output distribution is maximal.

Now, let $X$ denote the set of \emph{all} qudits. We want to compute $\expval{S(X)}$. We again proceed by computing the quasi-entropies: $$\tilde{S}_k(X) = \frac{F_{k,X}-F_{k,\emptyset}}{k-1}. $$

As before, for each auxiliary spin associated with region $X$ that we ``twist'', the weight of the configuration is decreased by a factor of $q^{1-k}$ relative to the configuration in which all auxiliary spins are set to $e$. However, in this case, as opposed to our previous calculation, \emph{all} of the auxiliary spins are twisted. Recall from \Cref{eq:dephasingWeights} that the weight between a twisted auxiliary spin and  a bulk spin is independent of the value of the bulk spin. Hence, if all auxiliary spins are twisted, the lowest energy state in the bulk is no longer just the configuration in which all spins take the value $e$ -- in the absence of a symmetry-breaking boundary condition, there is now a global spin-flip symmetry and the ground space is $k!$-fold degenerate, consisting of all configurations in which all bulk spins are aligned. This symmetry contributes a factor of $k!$ to the partition function and $-\log(k!)$ to the free energy. We hence calculate
\begin{equation}
	\tilde{S}_k (X) =   |A| \log(q) - \frac{\log(k!)}{k-1}.
\end{equation}
We now perform the analytic continuation to the Shannon entropy:
\begin{align}
	\expval{S(X)} &= \lim_{k\rightarrow 1} \tilde{S}_k (X) \\
	&= |A| \log(q) - \lim_{k\rightarrow 1} \frac{\log(k!)}{k-1} \\
	&= |A| \log(q) - \frac{1-\gamma}{\ln(2)} \\
	&\approx |A| \log(q) - 0.61,
\end{align}
where $\gamma \approx 0.557$ denotes the Euler constant. The expected Shannon entropy of the output distribution is therefore $\frac{1-\gamma}{\ln(2)}$ less bits than maximal in the low-temperature limit, corresponding to $q\rightarrow \infty$.

From the above facts, we can immediately compute the expected CMI of the output distribution in this limit. Let $(A,B,C)$ be any partition of the qudits. We have
\begin{align}
\expval{I(A:C|B)_{\mathcal{D}}} &\equiv \expval{S(AB)_{\mathcal{D}} + S(BC)_{\mathcal{D}} - S(B)_{\mathcal{D}} - S(ABC)_{\mathcal{D}}} \\
&= [(|A|+|B|)\log(q) ] + [(|B|+|C|)\log(q) ] \\
&\qquad - [(|B|)\log(q) ] - [(|A|+|B| + |C|)\log(q) - \frac{1-\gamma}{\ln(2)}]  \\
&= \frac{1-\gamma}{\ln(2)} \approx 0.61.
\end{align}

We therefore find that in this limit, the expected CMI of the classical output distribution approaches a constant equal to $\frac{1-\gamma}{\ln(2)}$. While this result was derived with respect to the \emph{completely} ordered stat mech model, corresponding to   $q\rightarrow \infty$, we expect similar behavior for ordered stat mech models in general. In particular, if $X$ denotes the set of all qudits, in the case of an ordered $k\textsuperscript{th}$-order stat mech model,  $\tilde{S}_k(X)$ will similarly receive an extra contribution corresponding to the global spin-flip symmetry, which will also be contributed to the corresponding quasi-CMI $\tilde{I}_k(A:C|B)_{\mathcal{D}}$.  Hence, we do not expect the quasi-CMIs to decay when the corresponding stat mech model is in an ordered phase. We take this as evidence that the average CMI does not decay, and therefore that \texttt{Patching} is not successful in efficiently sampling from the output distribution with small error.

\subsection{Depth-3 2D circuits with brickwork architecture}\label{sec:brickwork}

In the previous sections we discussed the implications of the stat mech mapping for random 2D circuits of variable depth $d$. In this section we fix $d=3$ and examine the order or disorder properties of the model. In particular, we choose uniform depth-3 circuits with the so-called ``brickwork'' structure, where three layers of two-qudit gates are performed on a 2D lattice of qudits as shown in Fig. \ref{fig:brickwork}(a). Note that this architecture was also introduced in \Cref{se:proof}; the architecture we consider here is exactly the ``extended brickwork architecture'' of that section with the extension parameter $r$ fixed to be one.

\begin{figure}[h]
    \centering
            \def\svgwidth{0.9\columnwidth}
    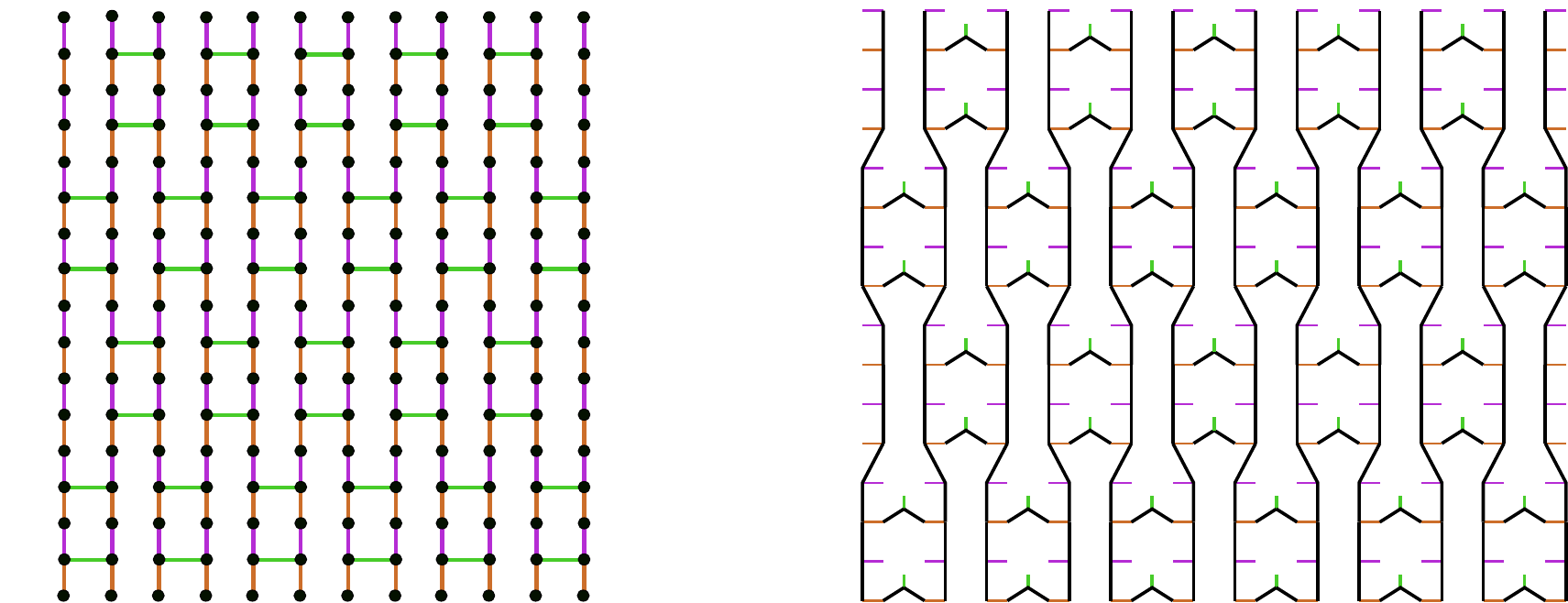
    \caption{ \label{fig:brickwork} (a) Brickwork architecture.  Qudits lie at location of black dots. Three layers of two-qudit gates act between nearest-neighbor qudits --- first qudits linked by a vertical purple edge, then vertical orange, then horizontal green. In our \texttt{SEBD} simulation of this circuit architecture, we sweep from left to right. (b) Result of stat mech mapping applied to brickwork circuit. Nodes are implied to lie at the ends of each edge in the graph. Purple, orange, green edges carry Weingarten weight. Black edges carry weight given by $\text{weight}(\langle s_{u_1}t_{u_2} \rangle) = q^{C(\sigma_{u_1} \tau_{u_2}^{-1})}$}
\end{figure}

This structure is known to be universal in the sense that one may simulate any quantum circuit using a brickwork circuit (with polynomial overhead in the number of qudits) by judiciously choosing which two-qudit gates to perform and performing adaptive measurements \cite{broadbent2009universal}. Thus, it is hard to exactly sample or compute the output probabilities of brickwork circuits in the worst case assuming the polynomial hierarchy does not collapse, and we expect neither the \texttt{SEBD} algorithm nor the \texttt{Patching} algorithm to be efficient. However, we give evidence that these algorithms are efficient in the ``average-case,'' where each two-qudit gate is Haar random, by considering the order/disorder properties of the stat mech model that the brickwork architecture maps to.

\paragraph{Stat mech mapping for general $k$.}
{
The stat mech mapping proceeds as previously discussed for 2D circuits, but we will see that the brickwork architecture allows us to make some important simplifications. Each gate in the circuit is replaced by a pair of nodes, which are connected with an edge. Then, the outgoing nodes of the first (purple) layer are connected to the incoming nodes of the second (orange) layer, and the outgoing nodes of the second (orange) layer are connected to the incoming nodes of the third (green) layer. The resulting graph is shown in Fig.~\ref{fig:brickwork}(b). Edges connecting incoming and outgoing nodes of the same layer are shown in color (purple, orange, green) and carry weight equal to the Weingarten function. Edges connecting subsequent layers are black. No weak measurement is performed between layers, so we may take $\mathcal{M} = \{\mathbb{I}_q\}$ along each of the black edges. Thus, these edges carry weight given by $\text{weight}(\langle s_{u_1} t_{u_2} \rangle) = q^{C(\sigma_u \tau_{u_2}^{-1})}$.

To perform the full mapping, we would also add a layer of auxiliary nodes to the graph and connect them to the third layer. However, we will suppose that most of the qudits undergo projective measurement after the third layer, and thus the auxiliary nodes may be omitted for those qudits. The auxiliary nodes will be important for any unmeasured qudits, but we assume these exist only at the edges of the graph. We do not need to consider auxiliary nodes to understand the bulk order/disorder properties of the model.

Looking at Fig.~\ref{fig:brickwork}(b), we see that some of the nodes have degree 1 and connect to the rest of the graph via a (purple or green) Weingarten link. We can immediately decimate these nodes from the graph. For any $\tau$, we have \cite{gu2013moments}
\begin{equation}
    \sum_{\sigma \in S_k} \Wg(\tau \sigma^{-1},q^2) = \sum_{\sigma \in S_k} \Wg(\sigma,q^2) = \frac{(q^2-1)!}{(k+q^2-1)!}
\end{equation}
which is independent of $\tau$, so decimating these spins merely contributes the above constant to the total weight. This constant will appear in both the numerator and denominator of quantities like $\E_U(Z_{k,A})/\E_U(Z_{k,\emptyset})$, and we ignore them henceforth. The remaining graph can be straightened out, yielding Fig.~\ref{fig:brickworkdecimated}(a). The fact that Fig.~\ref{fig:brickworkdecimated}(a) is a graph embedded in a plane that includes only two-body interactions is one upshot of studying the brickwork architecture, as it makes the analysis more straightforward and the stat mech model easier to visualize. This property and the fact that the brickwork architecture is universal for MBQC constitute the primary reasons we studied this architecture in the first place. Architectures with larger depth would lead to stat mech models that cannot be straightforwardly collapsed onto a single plane while maintaining the two-body nature of the interactions.
\begin{figure}[h]
    \centering
            \def\svgwidth{0.95\columnwidth}
    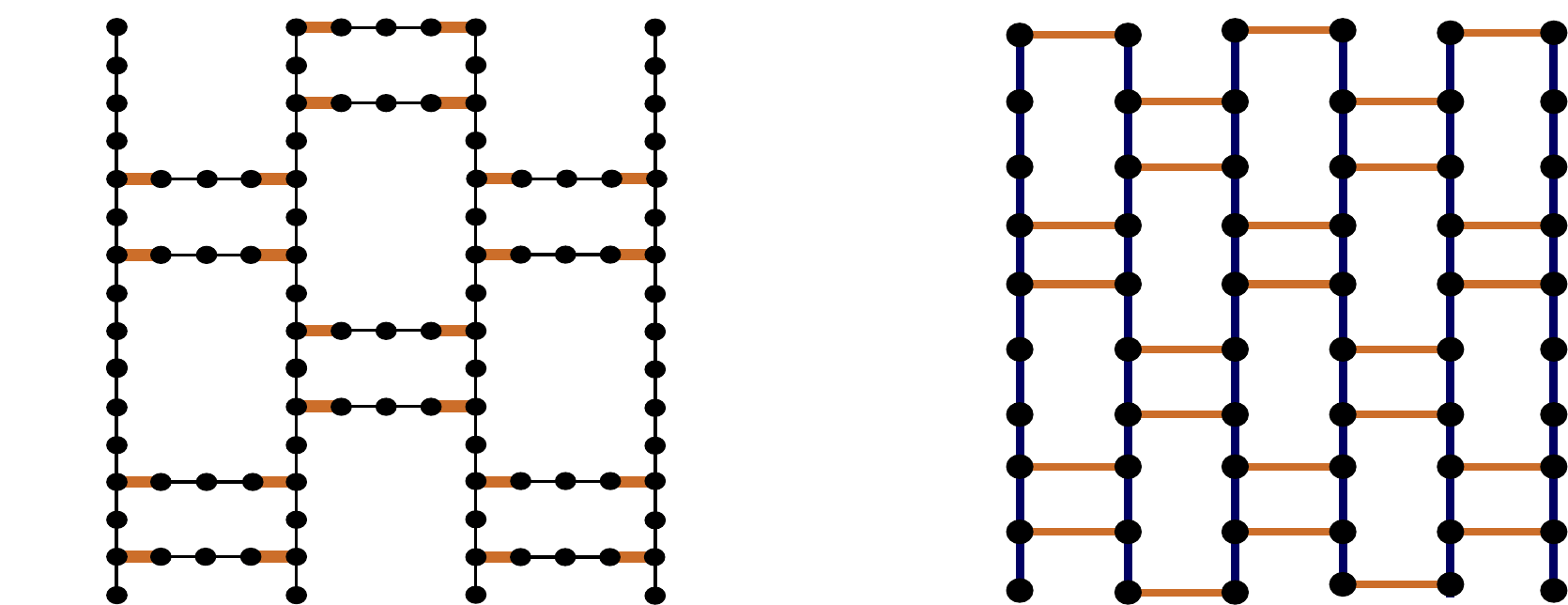
    \caption{ \label{fig:brickworkdecimated} (a) The graph that results from decimating degree-1 nodes in Fig.~\ref{fig:brickwork}(b). Each thin black link carries weight equal to the function $q^{C(\sigma\tau^{-1})}$ while each thick orange link carries weight equal to $\Wg(\sigma\tau^{-1},q^2)$. (b) The graph that results from decimating all degree two nodes for the graph in (a). For $k=2$, both the horizontal orange and the vertical dark blue links are ferromagnetic, but have different strengths.}
\end{figure}
}

\paragraph{Simplifications when $k=2$.}
{
As in previous examples, we examine the $k=2$ case. In this case we might as well decimate all the degree-2 nodes in the graph in Fig.~\ref{fig:brickworkdecimated}(a). This yields a graph with entirely degree-3 nodes, as shown in Fig.~\ref{fig:brickworkdecimated}(b).
The graph has two kinds of links, both carrying standard Ising interactions. The vertical dark blue links have weights given by
\begin{equation}
    \text{weight}(\langle s_us_{u'} \rangle) =
    \begin{cases}
    q^2(q^2+1) & \text{if } \sigma_u\sigma_{u'} =e\\
    q^2(2q) & \text{if } \sigma_u \sigma_{u'} = (12)
    \end{cases}
\end{equation}
while the horizontal orange links have weights given by
\begin{equation}
    \text{weight}(\langle s_us_{u'} \rangle) =
    \begin{cases}
    \frac{1}{q^2(q^4+1)^2}\left(q^6+q^4-4q^3+q^2+1\right) & \text{if } \sigma_u\sigma_{u'} =e\\
    \frac{1}{q^2(q^4+1)^2}\left(2q^5-2q^4-2q^2+2q\right) & \text{if } \sigma_u \sigma_{u'} = (12)
    \end{cases}
\end{equation}
Both of these interactions are ferromagnetic and become stronger as $q$ increases. We may think of the model as the square lattice Ising model for which 1/2 of the links carry a ferromagnetic interaction of one strength, 1/4 of the links carry ferromagnetic interactions of another strength, and the final 1/4 of the links have no interaction at all. The energy functional can be written
\begin{equation}
    E/(kT) = -J_{\text{vert}}\sum_{\langle i j\rangle} s_is_j -J_{\text{horiz}}\sum_{\langle i j\rangle} s_is_j
\end{equation}
where $s_i$ take on values in $\{+1,-1\}$. For $q=2$ we have $J_{\text{vert}} = \log(5/4)/2 = 0.112$ and $J_{\text{horiz}} = \log(53/28)/2 = 0.319$. Both of these values are weaker than the critical interaction strength for the square lattice Ising model of $J_{\text{square}}=\log(1+\sqrt{2})/2 = 0.441$. This indicates that the graph generated by the stat mech mapping on 2D circuits of depth 3 with brickwork architecture is in the disordered phase when $q=2$. This remains true for $q=3$. For $q=4$, $J_{\text{horiz}} =0.500 > J_{\text{square}}$, but $J_{\text{vert}} = 0.377 < J_{\text{square}}$. Recall that 1/4 of the links can be thought to have $J=0$ since they are missing. Taking this into account, the value of $J$ averaged over all the links remains below $J_{\text{square}}$ for $q=5$, and slightly exceeds it for $q=6$.

This indicates that when we run $\texttt{SEBD}$ on these uniform depth-3 circuits with Haar-random gates, the quantity $\tilde{S}_2 = O(1)$ (independent of the number of qudits $n$) when $q=2$ or $q=3$ (and probably also for $q=4$ and $q=5$). Moreover, it indicates that for a partition of the output state into regions $A$, $B$, and $C$, the quasi-2-CMI of the distribution over classical output distributions $\tilde{I}_2(A:C|B)$ decays exponentially in $\text{dist}(A,C)$. We take this as evidence that the \texttt{SEBD} and \texttt{Patching} algorithms would be efficient for these circuits.
}

\section*{Acknowledgments}
We thank Nicole Yunger Halpern, Richard Kueng, Saeed Mehraban, Ramis Movassagh, Anand Natarajan, and Mehdi Soleimanifar for helpful discussions. Numerical simulations were performed using the ITensor Library\footnote{\href{http://itensor.org}{http://itensor.org}}.  This work was funded by NSF grants CCF-1452616, CCF-1729369, PHY-1818914, and DGE‐1745301, as well as ARO
contract W911NF-17-1-0433, the MIT-IBM Watson AI Lab under the project {\it Machine Learning in Hilbert space} and the Dominic Orr Fellowship at Caltech. The Institute for Quantum Information and Matter (IQIM) is an NSF Physics Frontiers Center (PHY-1733907).

\appendix

\section{Relation to worst-to-average-case reductions based on truncated Taylor series}\label{app:supremacy}
In this section, we discuss the relation between our algorithms (\texttt{SEBD} and \texttt{Patching}) applied to the computation of output probabilities and a recent result \cite{bouland2019complexity} on the hardness of average-case simulation of random circuits based on polynomial interpolation. In particular, we discuss how this polynomial interpolation argument is insufficient to show that  the task of even \emph{exactly} computing output probabilities and sampling from the output distribution of a constant-depth Haar-random circuit instance with high probability using our algorithms is classically hard, even though these circuits possess worst-case hardness. We first briefly review their technique before discussing a limitation in the robustness of the polynomial interpolation scheme.  We then discuss how this robustness limitation makes the interpolation scheme inapplicable to our algorithms.

The main point is that our algorithms exploit unitarity (via the fact that gates outside of the lightcone of the qudits currently under consideration are ignored), but the hardness result of \cite{bouland2019complexity} holds with respect to circuit families that are non-unitary, albeit very close to unitary in some sense. Our algorithms are unable to simulate these slightly non-unitary circuits to the precision required for the worst-to-average case reduction, regardless of how well they can simulate Haar-random circuit families. While it is true that in this scheme there is an adjustable parameter $K$ which, when increased, brings the non-unitary circuit family closer to approximating the true Haar-random family, increasing $K$ also increases the degree of the interpolating polynomial. This makes the interpolation more sensitive to errors in such a way that, for any choice of $K$, the robustness that the interpolation can tolerate is not large enough to overcome the inherent errors that our algorithms make when trying to simulate these non-unitary families.  The existence of simulation algorithms like \texttt{SEBD} and \texttt{Patching}, which exploit the unitarity of the circuit, may present an obstruction to applying worst-to-average-case reduction techniques that obtain a polynomial structure at the expense of unitarity. Note that, as discussed in the main text, a very recent alternative worst-to-average case reduction \cite{movassagh2019} based on ``Cayley paths'' rather than truncated Taylor series does not suffer from this same limitation.

\subsection*{Background: truncated Haar-random circuit ensembles and polynomial interpolation}
In this section, we give an overview (omitting some details) of the interpolation technique of \cite{bouland2019complexity} used to show their worst-to-average-case reduction, partially departing from their notation. Suppose $U$ is a unitary operator. Then we define the $\theta$-contracted and $K$-truncated version of $U$ to be $U^\prime(\theta, K) = U \sum_{k=0}^K \frac{(- \theta \ln U )^k}{k!}$. Note that $U^\prime(\theta, \infty) = U e^{-i \theta(-i \ln U)}$ is simply $U$ pulled-back by angle $\theta$ towards the identity operator $I$. Note that $U^\prime(0,\infty) = U$ and $U^\prime(1,\infty) = I$. For $U^\prime(\theta,K)$ for $K<\infty$, the operator that performs this pullback is then approximated by a Taylor series which is truncated at order $K$. If $K < \infty$, $U^\prime(\theta,K)$ is (slightly) non-unitary.

Suppose $C$ is some circuit family for which computing output probabilities up to error $2^{-\poly(n)}$ is classically hard. Now, for each gate $G$ in $C$, multiply that gate by $H^\prime(\theta, K)$ with $H$  Haar-distributed and supported on the same qubits as $G$. This yields some distribution over non-unitary circuits that we call $\mathcal{D}(C,\theta,K)$. Note that if $\theta = 0$, $\mathcal{D}$ exactly becomes the Haar-random circuit distribution with the same architecture as $C$. When $\theta = 1$, the hard circuit $C$ is recovered up to some small correction due to the truncation. If $K$ is sufficiently large, we can assume that computing output probabilities for this slightly perturbed version of $C$ is also classically hard.

Fix some circuit $A$ drawn from $\mathcal{H}(C)$, the distribution over circuits with the same architecture as $C$ with gates chosen according to the Haar measure. Let $A(C,\theta,K)$ denote the circuit obtained when the $\theta$-pulled-back and $K$-truncated gates of $A$ are multiplied with their corresponding gates in $C$. Note that $A(C,\theta,K)$ is distributed as $\mathcal{D}(C,\theta,K)$.   Define the quantity
\begin{equation}
    p_0(A,\theta,K) := \abs{\expval{A(C,\theta,K)}{0}}^2.
\end{equation}

Assuming the circuit $C$ has $m$ gates, it is easy to verify that $p_0(A,\theta,K)$ may be represented as a polynomial in $\theta$ of degree $2mK$. Note also that $p_0(A,1,\infty) = p_0(C)$, which is assumed to be classically hard to compute.

Now, assume that there exists some classical algorithm $\mathcal{A}$ and some $\epsilon = 1/\poly(n)$ such that, for some fixed $K \leq \poly(n)$ and for all $0\leq \theta \leq \epsilon$, $\mathcal{A}$ can compute $p_0(A,\theta,K)$ up to additive error $\delta \leq 2^{-n^c}$ for some constant $c$, with probability $1-1/\poly(n)$ over $A(C,\theta,K) \sim \mathcal{D}(C,\theta,K)$. Then, $\mathcal{A}$ may evaluate $p_0(A,\theta,K)$ for $2mK+1$ evenly spaced values of $\theta$ in the range $[0,\epsilon]$ (up to very small error), and construct an interpolating polynomial $q_0(A,\theta,K)$. By a result of Rakhmanov \cite{rakhmanov}, there is some interval $[a,b] \subset [0,\epsilon]$ such that $b-a \geq 1/\poly(n)$ and $|p_0(A,\theta,K) - q_0(A,\theta,K)| \leq 2^{-n^{c^\prime}}$  for $\theta \in [a,b]$ where $c^\prime$ depends on $c$. One then invokes the following result of Paturi.

\begin{lemma}[\cite{paturi}]\label{paturi}
Let $p\, : \, \mathbb{R} \rightarrow \mathbb{R}$ be a real polynomial of degree $d$, and suppose $|p(x)| \leq \delta$ for all $|x|\leq \epsilon$. Then $|p(1)| \leq \delta e^{2d(1+1/\epsilon)}$.
\end{lemma}

Applying this result, we find $|p_0(A,1,K) - q_0(A,1,K)| \leq 2^{-n^{c^\prime}} e^{\poly(n,m,K)}$. If $c$ is sufficiently large, then $|p_0(A,1,K) - q_0(A,1,K)| \leq 2^{-\poly(n)}$ and the quantity $q_0(A,1,K)$ is hard to compute classically. But this would be a contradiction, because $q_0(A,1,K)$ can be efficiency evaluated  classically by performing the interpolation.

Hence, this argument shows that for some choice of $K$ and a sufficiently large $c$ depending on $K$, computing output probabilities of circuits in the truncated families $\mathcal{D}(C,\theta,K)$ with $\theta\leq 1/\poly(n)$ up to error $2^{-n^c}$ is hard (assuming standard hardness conjectures).

\subsection*{Limitation of the interpolation argument}
The above argument shows that the average-case simulation of some family $\mathcal{D}(C,\theta,K)$ of non-unitary circuits which in some sense is close to the corresponding Haar-random circuit family to precision $2^{-\poly(n)}$ is classically hard, if simulating $C$ is classically hard and the polynomial in the exponent is sufficiently large.

We now explain how, based on this argument, we are unable to conclude that exactly computing output probabilities of Haar-random circuits is classically hard.\footnote{A similar argument recently appeared independently in \cite{movassagh2019}. That argument shows that $K$ must be at least exponentially large for the interpolation result to work for non-truncated Haar-random circuits, while we argue that in fact no value of $K$ is sufficient.} In other words, suppose that with probability $1-1/\poly(n)$, some algorithm $\mathcal{A}$ can \emph{exactly} compute output probabilities from the distribution $\mathcal{H}(C)$. We argue that a straightforward application of the above result based on Taylor series truncations and polynomial interpolation is insufficient to compute $p_0(C)$ with small error.

Consider some circuit realization $A$ drawn from $\mathcal{H}(C)$, and assume that we can exactly compute its output probability $p_0(A)$. To use the argument of \cite{bouland2019complexity}, we actually need to compute $p_0(A,\theta,K)$ for some fixed value of $K$ and $\theta$ in some range $[0,\epsilon]$. We first find an upper bound for $\epsilon$ which must be satisfied for the interpolation to be guaranteed to succeed with high probability. To this end, we note that  \cite{bouland2019complexity}  the total variation distance between the distributions $\mathcal{D}(C,\theta,\infty)$ and $\mathcal{D}(C,0,\infty)$ is bounded by $O(m\theta)$. Hence, if we try to use the algorithm $\mathcal{A}$ to estimate $p_0(A,\theta,\infty)$, the failure probability over random circuit instances could be as high as $O(m\theta)$. Therefore, since the $\theta$ values to be evaluated are uniformly spaced on the interval $[0,\epsilon]$, the union bound tells us that the probability that one of the  $2mK+1$ values $p_0(A,\theta,K)$ is erroneously evaluated is bounded by $O(m^2 K \epsilon)$. Hence, in order to ensure that all $2mK+1$ points are correctly evaluated, we should take $\epsilon \leq O(1/m^2 K)$.

Now, assume that we have chosen $\epsilon \leq O(1/m^2 K)$ and all $2mK+1$ points $p_0(A,\cdot,\infty)$ are correctly evaluated. Let $\theta$ be one of the evaluation points. We now must consider the error made by approximating the ``probability'' associated with the truncated version of the circuit with the probability associated with the untruncated version of the circuit, namely $|p_0(A,\theta,\infty) - p_0(A,\theta,K)|$. This error associated with the truncated Taylor series is upper bounded by $\delta  \leq \frac{2^{O(nm)}}{K!}$ \cite{bouland2019complexity}.

Plugging these values into \Cref{paturi}, we find that if we use these values to try to interpolate to the classically hard-to-compute quantity $p_0(C,1,K)$, the error bound guaranteed by Paturi's lemma is no better than $\frac{2^{O(nm)}}{K!} \exp\qty(O(2mK(1+m^2 K)))$, which diverges in the limit $n\rightarrow \infty$ for any scaling of $m$ and $K$. Hence, the technique of \cite{bouland2019complexity} is insufficient to show that exactly computing output probabilities of circuits drawn from the Haar-random circuit distribution $\mathcal{H}_C$ with high probability is hard.

Intuitively, the limitation arises because there is a tradeoff between the amount of truncation error incurred  and the degree of the interpolating polynomial. As the parameter $K$ is increased, the truncation error is suppressed, but the degree of the interpolating polynomial is increased, making the interpolation more sensitive to errors.

\subsection*{Inapplicability to \texttt{SEBD} and \texttt{Patching}}
To summarize the findings above, the argument of \cite{bouland2019complexity} for the hardness of computing output probabilities of random circuits applies not directly to Haar-random circuit distributions, but rather to distributions over slightly non-unitary circuits that are exponentially close to the corresponding Haar distributions in some sense. We argued that the interpolation scheme cannot be straightforwardly applied to circuits that are truly Haar-random, and therefore it cannot be used to conclude that simulating truly Haar-random circuits, even exactly, is classically hard.

\emph{A priori}, it is not obvious  whether this limitation is a technical artifact or a more fundamental limitation of the interpolation scheme. In particular, one might imagine that if some algorithm $\mathcal{A}$ is capable of exactly simulating Haar-random circuit families, some modified version of the algorithm $\mathcal{A}^\prime$ might be capable of simulating the associated truncated Haar-random circuit families, at least up to the precision needed for the interpolation argument to work. If this were the case, then the hardness argument \emph{would} be applicable.

However, \texttt{SEBD} and \texttt{Patching} appear to be algorithms that \emph{cannot} be straightforwardly used to efficiently simulate truncated Haar-random circuit families to the precision needed for the interpolation to work, even under the assumption that they can efficiently, exactly simulate Haar-random circuit families. This is because the efficiency of these algorithms crucially relies on the existence of a constant-radius lightcone for constant-depth circuits.  The algorithm is able to ignore all qubits and gates outside of the lightcone of the sites currently being processed. However, the lightcone argument breaks down for non-unitary circuits. If the gates are non-unitary and we want to perform an exact simulation, we are left with using Markov-Shi  or some other general-purpose tensor network contraction algorithm, with a running time of $2^{O(d \sqrt{n})}$ for a depth-$d$ circuit on a square grid of $n$ qubits.

Consider what happens if one tries to use one of these algorithms to compute output ``probabilities'' for a slightly non-unitary circuit coming from a truncated Haar-random distribution $\mathcal{D}(C,\theta,K)$, and then use these computed values to interpolate to the hard-to-compute value $p_0(C,1,K)$ via the interpolating polynomial of degree $2mK$ proposed in \cite{bouland2019complexity}. Even without any other sources of error, when one of these algorithms ignores gates outside of the current lightcone, it is essentially approximating each gate outside the lightcone as unitary. This causes an incurred error bounded by $2^{O(nm)}/K!$ for the computed output probability. Then, by an argument essentially identical to the one appearing in the previous section, one finds that this error incurred just from neglecting gates outside the lightcone is already large enough to exceed the error permitted for the polynomial interpolation to be valid. We conclude that this worst-to-average-case reduction based on truncated Taylor series expansions cannot be used to conclude that it is hard for \texttt{SEBD} or \texttt{Patching} to exactly simulate worst-case hard shallow Haar-random circuits with high probability.

\section{Deferred proofs}\label{appendix:error}
\samplingError*
\begin{proof}
We rely upon a well-known fact about the error caused by truncating the bond dimension of a MPS, which we state in \Cref{lem:truncationError}.

\begin{lemma}[follows from \cite{verstraete2006matrix}]\label{lem:truncationError}
Suppose the MPS $\ket{\psi}$ is compressed via truncation of small singular values, and $\epsilon$ is the sum of the squares of the discarded singular values. Then if $\ket{\psi^{(t)}}$ is the truncated version of the MPS after normalization,
\begin{equation}
    \| |\psi \rangle \langle \psi | - |\psi^{(t)} \rangle \langle \psi^{(t)} | \|_1 \leq \sqrt{8 \epsilon}.
\end{equation}
\end{lemma}
The second inequality follows from the fact that $\sqrt{\sum_i x_i^2} \leq \sum_i x_i$ for $x_i \geq 0$. To prove the first inequality, we  start by considering the version of the algorithm with no truncation, which we have argued samples exactly from $\mathcal{D}$. Let $\mathcal{N}_t$ denote the TPCP map corresponding to  the application of gates that have come into the lightcone of \texttt{column t} and the measurement of \texttt{column t}. That is,
\begin{equation}
\mathcal{N}_t (\rho) =  \sum_{\vb{x}_t} \Pi_t^{\vb{x}_t} V_t \rho  V_t^\dagger \Pi_t^{\vb{x}_t} ,
\end{equation}
where $\vb{x}_t$ indexes (classical) outcome strings of \texttt{column t}. Note that $\mathcal{N}_t(\rho)$ is a classical-quantum state for which the sites corresponding to the first $t$ columns are classical, and the quantum register consists of sites which are in the lightcone of \texttt{column $t$} but not in the first $t$ columns.   Define $\rho_t := \mathcal{N}_{t-1}(\rho_{t-1})$ and $\rho_1 := \dyad{1}^{\otimes L_1}_{\texttt{column 1}}$, so that $\rho_{L_2+1}$ is a classical state exactly corresponding to output strings on the $L_1 \times L_2$ grid distributed according to $\mathcal{D}$.

Now consider the ``truncated'' version of the algorithm, which is defined similarly except we use $\sigma_t$ to denote the state of the algorithm immediately after the truncation at the beginning of iteration $t$. That is, we define
\begin{equation}
\sigma_{t} := (T_{t} \circ \mathcal{N}_{t-1}) (\sigma_{t-1}),
\end{equation}
where $T_t$ denotes the mapping corresponding to the MPS truncation and subsequent renormalization at the beginning of iteration $t$, and we define $\sigma_1 := T_1(\rho_1) = \rho_1$ (there is no truncation at the beginning of the first iteration since the initial state is a product state).

We now have
\begin{align}
    \| \mathcal{D}_C - \mathcal{D}^\prime_C \|_1 &= \| \rho_{L_2+1} - \sigma_{L_2+1} \|_1 \\
    &\leq \| \rho_{L_2+1} - \mathcal{N}_{L_2}(\sigma_{L_2}) \|_1 + \| \mathcal{N}_{L_2}(\sigma_{L_2}) - \sigma_{L_1+1} \|_1 \\
    &\leq \| \rho_{L_2} - \sigma_{L_2} \|_1 + \| \mathcal{N}_{L_2}(\sigma_{L_2}) - \sigma_{L_2+1} \|_1,
\end{align}
where the first inequality follows from the triangle inequality, and the second from contractivity of TPCP maps. Applying this inequality recursively yields
\begin{equation}
\norm{\mathcal{D}_C-\mathcal{D}^\prime_C}_1  \leq \sum_{i=1}^{L_2} \| \mathcal{N}_i (\sigma_i) - \sigma_{i+1} \|_1 = \sum_{i=1}^{L_2 - 1} \| \mathcal{N}_i (\sigma_i) - (T_{i+1} \circ \mathcal{N}_i) (\sigma_i) \|_1
\end{equation}
where we also used the fact that no truncation occurs after $\mathcal{N}_{L_2}$ is applied (i.e. $T_{L_2+1}$ acts as the identity). Now, note that $\norm{ \mathcal{N}_i (\sigma_i) - (T_{i+1} \circ \mathcal{N}_i) (\sigma_i)}_1$ is exactly the expected error in 1-norm caused by the truncation in iteration $i+1$. (This is true because of the following fact about classical-quantum states: $\norm{  \E_i  \dyad{i}_C \otimes (\dyad{\psi_i}_Q - \dyad{\phi_i}_Q)}_1 = \E_i \norm{\dyad{\psi_i} - \dyad{\phi_i}}_1$ where $\{\ket{i}_C\}_i$ is an orthonormal basis for the Hilbert space associated with register $C$.) By \Cref{lem:truncationError}, this quantity is bounded by $\E \sqrt{8 \epsilon_{i+1}}$. Substituting this bound into the summation yields the desired inequality.
\end{proof}
\toyModel*
\begin{proof}
   Suppose that an EPR pair is measured $2t$ times, corresponding to each of the two qubits being measured $t$ times. A calculation shows that the probability of obtaining $s$ $M_1$ outcomes is given by a mixture of two binomial distributions. Letting $S$ be the random variable denoting the number of $M_1$ outcomes, we find
\begin{equation}
\Pr[S=s] = \frac{1}{2} \Pr[B_{2t, \sin^2 (\theta/2)} = s] + \frac{1}{2} \Pr[B_{2t, \cos^2 (\theta/2)} = s],
\end{equation}
where $B_{n,p}$ denotes a binomial random variable associated with  $n$ trials and success probability $p$. If after the $2t$ measurements we obtain outcome $M_1$ $s$ times, the post-measurement state is given by (up to normalization)
\be \ket{00} + \tan^{2(t-s)}(\theta/2)\ket{11}.\label{eq:outcome-s}\ee
Note that $s$ can be assumed to be generated by sampling from either  $B_{2t, \sin^2 (\theta/2)}$ or $B_{2t, \cos^2 (\theta/2)}$ with probability $1/2$ each. In the former case, the post-measurement state may be written
\be \ket{00} + \tan^{2\left(t-B_{2t,\sin^2(\theta/2)}\right)}(\theta/2)\ket{11} = \ket{00} + \tan^{2t\cos(\theta)-2 X_{2t,\sin^2(\theta/2)}}(\theta/2) \ket{11} \ee
where we have defined the random variable $X_{2t,\sin^2(\theta/2)}$ via $B_{n,p} = np + X_{n,p}$.  That is, the random variable $X_{n,p}$ is distributed as a binomial distribution shifted by its mean.  Now, defining $\gamma := (\tan(\theta/2))^{2\cos(\theta)}$ and $X_{n,p}' = X_{n,p} / \cos(\theta)$, we may write the post-measurement state as

\be \ket{00} + \gamma^{t - X'_{2t,\sin^2(\theta/2)}} \ket{11}. \ee

We assume WLOG that $0<\theta<\pi/2$, so that $0<\gamma<1$. Similarly, if $s$ is drawn from $B_{2t,\cos^2(\theta/2)}$, then the post-measurement state is given by

\be \ket{00} + \gamma^{-t - X'_{2t,\cos^2(\theta/2)}} \ket{11}. \ee

Note that, under a relabeling of basis states $0 \leftrightarrow 1$, the post-measurement state in this case is

\be \ket{00} + \gamma^{t - X'_{2t,\sin^2(\theta/2)}} \ket{11}, \ee
where we used the fact that $-X'_{2t,\cos^2(\theta/2)}$ is distributed identically to $X'_{2t,\sin^2(\theta/2)}$. Since we will be interested in studying the entanglement spectrum of this process, which is invariant under such local basis changes, we may assume WLOG that the random post-measurement state after $2t$ measurements is given by $ \ket{00} + \gamma^{t - X'_{2t,\sin^2(\theta/2)}} \ket{11}$.

We can then model the final state as
\be \bigotimes_t  \ket{00} + \gamma^{t - X'_{2t,\sin^2(\theta/2)}} \ket{11} \label{unnormalized}\ee
up to normalization. This allows an estimate of the tradeoff between rank, truncation error, and associated probability of success.

Let $Q(\ell)$ denote the number of ``strict partitions'' of $\ell$, i.e. the number of ways of writing $\ell = t_1 + t_2 + \dots $ for positive integers $t_1 < t_2 < \dots$. Precise asymptotics are known for $Q(\ell)$ (see \url{https://oeis.org/A000009} and \cite{ayoub1963partitions}):
\be Q(\ell) = \exp(\Theta(\sqrt{\ell})).\ee

By expanding \Cref{unnormalized} as a superposition over computational basis states, we obtain the unnormalized Schmidt coefficients $\tilde{\lambda}_1 \geq \tilde{\lambda}_2 \geq \cdots$; each coefficient in the expansion gives an unnormalized Schmidt coefficient. There are $Q(\ell)$ unnormalized Schmidt coefficients that are distributed as $\gamma^{\ell - X'_{2\ell , \sin^2(\theta/2)}}$, where we used the fact that $X'_{t_1, \sin^2(\theta/2)} + X'_{t_2, \sin^2(\theta/2)}$ is distributed as $X'_{t_1+t_2, \sin^2(\theta/2)}$. We say that these $Q(\ell)$ coefficients live in sector $\ell$.  For a fixed probability $p$, let $K_{\ell,p}$ denote the smallest positive integer for which, with probability at least $1-p$, all sector-$\ell$ coefficients lie in the range $[\gamma^{\ell+K_{\ell, p}}, \gamma^{\ell-K_{\ell, p}}]$.  By the union bound, to upper bound $K_{\ell, p}$ it suffices to find an integer $a$ for which
\be \Pr[ \left|X'_{2\ell, \sin^2 (\theta/2)} \right| \geq a ] \leq \frac{p}{Q(\ell)} = p \exp(-\Theta(\sqrt{\ell})). \ee
By Hoeffding's inequality, we have $\Pr[\left|X'_{2\ell, \sin^2 (\theta/2)} \right| \geq a ] \leq \exp(-\Theta(a^2 / \ell))$; this yields the bound
\be K_{\ell, p} \leq \Theta\left(\sqrt{\ell \log (1/p) + \ell \sqrt{\ell}}\right). \ee
Furthermore, note that since there are $\Theta(n^2)$ sectors, by the union bound, with probability at least $1-\delta$, for each sector $j$, all coefficients lie in the range $[\gamma^{j+K_{j,p}}, \gamma^{j-K_{j,p}}]$ if we take $p$ to be $p = \delta / \Theta(n^2)$. We make this choice of $p$ and assume for the remainder of the argument that all coefficients of sector $j$ lie in the given range, which is true with probability at least $1-\delta$.  We also note the following fact which will be used below: if $\ell$ and $p$ are related as $\ell \geq \Theta(\log(1/p))$, then $K_{\ell, p} = O(\ell)$.

Still working with the unnormalized state of \Cref{unnormalized}, we now study the scaling between the Schmidt index $i$ and corresponding coefficient $\tilde{\lambda}_i$ for $i$ in the regime $i \geq \exp(\Theta(\sqrt{\log(1/p)}))$. Note that $\tilde{\lambda}_i = \gamma^\ell$ for some integer $\ell$. We first lower bound $\ell$. Note that the lower bound is achieved if, for each sector $j$, all coefficients in that sector are equal to $\gamma^{j-K_{j,p}}$. In this case, the exponent $\ell$ is equal to $\ell' - K_{\ell', p}$, where $\ell'$ is the smallest integer such that
\be i \leq \sum_{j=1}^{\ell'} Q(\ell') = \exp(\Theta(\sqrt{\ell'})). \ee
Rearranging, we see that $\ell' = \Theta(\log^2(i)) \geq \Theta(\log(1/p))$, and hence $\ell = \Theta(\log^2(i))$ since $\ell' - K_{\ell', p} = \Theta(\ell')$. Similarly , an upper bound on $\ell$ is achieved if, for each sector $j$, all coefficients in that sector are equal to $\gamma^{j+K_{j,p}}$. In this case, $\ell$ is equal to $\ell' + K_{\ell',p}$, where $\ell'$ is defined as above. This yields a matching upper bound for $\ell$ of $\Theta(\log^2(i))$. We therefore have the scaling $\ell = \Theta(\log^2(i))$, which, using the fact that $\tilde{\lambda}_i = \gamma^\ell$ yields
\be \tilde{\lambda}_i = \exp(\Theta(-\log^2(i))), \, \, i \geq \exp(\Theta(\sqrt{\log(1/p)})). \ee
Noting that $\lambda_i$ is proportional to $\tilde{\lambda}_i$ via $\lambda_i = \frac{1}{N}\tilde{\lambda}_i$ with $N = \sqrt{\sum_i \tilde{\lambda}_i^2}$, this shows the first statement of the lemma.

Now, suppose that for some $i \geq i^* = \exp(\Theta(\sqrt{\log(1/p)}))$, we truncate all Schmidt coefficients with index $\geq i$. The incurred truncation error is
\be \epsilon = \sum_{j\geq i} \lambda_j^2 < \sum_{j\geq i} \tilde{\lambda}_j^2 = \exp(-\Theta(\log^2(i))) \ee
where the inequality holds because the unnormalized state has norm strictly greater than one (i.e. $N > 1$).  Rearranging, this becomes
\be i \leq \exp(\Theta(\sqrt{\log(1/\epsilon)})). \ee
Hence, if we truncate the state at the end of the process up to a truncation error of $\epsilon$, the rank $r$ of the post-truncation state is bounded by
\be r \leq \max\qty(\exp(\Theta(\sqrt{\log(1/\epsilon)})), \exp(\Theta(\sqrt{\log(1/p)}))) = \exp\left(\Theta\left(\sqrt{\log\left(\frac{n}{\epsilon\cdot \delta}\right)}\right)\right) \ee
as desired, where we used the relation $p = \delta / \Theta(n^2)$.
\end{proof}

\entanglementDecay*
\begin{proof}
We will use a smaller technical lemma, which we state and prove below.
\begin{lemma}\label{lem:entropy_decrease}
Let $\ket{\psi}_{AB}$ be some state on subsystems $A$ and $B$ with subsystem $B$ a qubit, and let $\ket{H}_{CD}$ be some two-qubit Haar-random state on subsystems  $C$ and $D$. Suppose a Haar-random two-qubit gate $U$ is applied to subsystems $B$ and $C$. If subsystem $B$ is measured in the computational basis and outcome $b$ is obtained, then the von Neumann entropy of the post-measurement state $\ket{\psi_b}_{ABCD}$ in subsystem $A$ satisfies
\begin{equation}
\E_{b,H,U} S(A)_{\psi_b} \leq c \cdot S(A)_{\psi}
\end{equation}
for some constant $c < 1$, where the expectation is over the random measurement outcome, the random state $\ket{H}_{CD}$, and the Haar-random unitary $U$.
\end{lemma}
\begin{proof}
Consider the Schmidt decomposition $\ket{\psi}_{AB} = \sqrt{p} \ket{e_1}_A \ket{f_1}_B + \sqrt{1-p} \ket{e_2}_A \ket{f_2}_B$ where we assume WLOG that $p \geq 1/2$. We also assume that $p<1$, because the statement is trivially true for any value of $c$ when $p=1$. Note that the entanglement entropy of this state is simply $S(A)_\psi = H_2(p)$ where $H_2(p) := -p \log p - (1-p) \log(1-p)$ is the binary entropy function.  Let $M_0 := (\Pi_0 \otimes I)U$ and $M_1 := (\Pi_1 \otimes I)U$ denote the measurement operators acting on subsystems $B$ and $C$, where $\Pi_i$ denotes the projector onto the computational basis state $\ket{i}$ and $U$ is the Haar-random unitary applied to subsystems $B$ and $C$. Let $X$ denote a random variable equal to $1$ with probability $p$ and equal to $2$ with probability $1-p$. Let $Y$ denote the measurement outcome of $\{M_0, M_1\}$ when applied to the state $\ket{e_X}_A \ket{f_X}_B \ket{H}_{C, D}$. The probability of obtaining measurement outcome $b$ on the original state is simply $\Pr(Y=b)$, and the post-measurement state after obtaining outcome $b$ is
\begin{align}
&\frac{1}{\sqrt{\Pr(Y=b)}} \left(\sqrt{p\cdot \Pr(Y=b|X=1)} \ket{e_1}_A \ket{b}_B \ket{\phi_{b,1}}_{C,D} + \sqrt{(1-p)\cdot \Pr(Y=b|X=2)} \ket{e_2}_A \ket{b}_B \ket{\phi_{b,2}}_{C,D} \right) \\
= &\sqrt{\Pr(X=1 | Y=b)} \ket{e_1}_A \ket{b}_B \ket{\phi_{b,1}}_{C,D} + \sqrt{\Pr(X=2 | Y=b)} \ket{e_2}_A \ket{b}_B \ket{\phi_{b,2}}_{C,D}
\end{align}
where $\ket{\phi_{b,j}}_{C,D}$ are normalized states on subsystems $C$ and $D$. Define
\begin{equation}
\epsilon := \min_b |\langle \phi_{b,1} | \phi_{b,2} \rangle |^2.
\end{equation}
Letting $\rho_{A,b}$ denote the reduced density matrix on subsystem $A$ of the post-measurement state after obtaining measurement outcome $b$, the maximal eigenvalue of this matrix is lower bounded as $\lambda_{\max}(\rho_{A,b}) \geq \Pr(X=1|Y=b) + \epsilon \Pr(X=2 | Y=b)$. (To see this, observe that the reduced density matrix on $CD$ is $\sigma = \Pr(X=1|Y=b)\dyad{\phi_{b,1}} + \Pr(X=2|Y=b)\dyad{\phi_{b,2}}$, and the maximal eigenvalue is lower bounded as $\lambda_{\max}(\rho_{A,b}) = \lambda_{\max}(\sigma) \geq \expval{\sigma}{\phi_{b,1}}  \geq \Pr(X=1|Y=b) + \epsilon \Pr(X=2 | Y=b)$). Furthermore, note that
\begin{equation}
\E_{Y} \lambda_{\max}(\rho_{A,Y}) \geq \E_{Y} \qty[\Pr(X=1|Y) + \epsilon \Pr(X=2 | Y)] = p + \epsilon (1-p).
\end{equation}
Now, using concavity of the binary entropy function, we have
\begin{equation}
\E_{Y} S(A)_{\psi_Y} = \E_{Y} H_2(\lambda_{\max}(\rho_{A,Y})) \leq H_2(\E_{Y} \lambda_{\max}(\rho_{A,Y})) \leq H_2 (p + \epsilon(1-p)).
\end{equation}
Consider the ratio $r(p,\epsilon) := \frac{H_2(p + \epsilon (1-p))}{H_2 (p)}$. We want to argue that for any $\epsilon>0$, $r(p,\epsilon)$ is bounded away from one on the interval $p \in [1/2,1)$. This statement is clearly true for any $p$ bounded away from one since $H_2$ is monotonically decreasing on the interval $[1/2,1)$. Furthermore, it is straightforward to show $\lim_{p\rightarrow 1} r(p,\epsilon) = 1-\epsilon$.  Hence, we have
\begin{equation}
\frac{\E_Y S(A)_{\psi_Y}}{S(A)_{\psi}} \leq r(p,\epsilon) \leq c(\epsilon)
\end{equation}
where $c(\epsilon) < 1$ unless $\epsilon = 0$. We now average both sides over the choice of Haar-random state on $CD$ as well as the Haar-random unitary $U$  acting on $BC$. Since the event $\epsilon > 0$ occurs with nonzero probability (in fact, with probability one), we have the strict inequality $\mathbb{E}_{H,U} \left[c(\epsilon)\right] := c < 1$, from which the desired inequality follows.  \end{proof}

We may assume that $i \neq 0$ and $j \neq n$, as in these cases we trivially have $S(\rho_A(b)) = 0$. The post-measurement state may be constructed as follows. Apply all gates in the lightcone of qubit $i$, then measure qubit $i$. Now apply all gates in the lightcone of qubit $i+1$ not previously applied, then measure qubit $i+1$. Assume that qubits are introduced only when they come into the lightcone under consideration. Iterate until all qubits in region $B$ have been measured. Finally, apply any gates that have not yet been applied. It is straightforward to verify that this is equivalent to applying all gates of the circuit before performing the measurement of region $B$, in the sense that the measurement statistics are the same, and the post-measurement state given some outcome $b$ is the same.

By \Cref{lem:entropy_decrease}, after the first iteration we are left with the state $\ket{\psi}_{LR} \ket{b_{i_1}}_{i_1}$, such that $\mathbb{E} S(L)_{\psi} \leq c$ for some constant $c < 1$. In all iterations, we let $L$ denote the current subsystem to the left of the measured qubits, and $R$ denote the subsystem to the right of the measured qubits. Now consider the second iteration. Depending on whether $i$ was even or odd, $R$ may consist of one or two qubits immediately after the measurement of $i$. In the former case, we may apply \Cref{lem:entropy_decrease} again, obtaining  $\mathbb{E} S(L)_{\psi} \leq c^2$ after the measurement of qubit $i+1$, and obtaining a two-qubit subsystem to the right of the measured qubits. In the latter case, as a consequence of concavity of von Neumann entropy, we have $\mathbb{E} S(L)_{\psi} \leq c$ after measurement, and are left with a one-qubit subsystem to the right of the measured qubits. Iterating this process, after all qubits of subregion $B$ have been measured, we are left with some state $\ket{\psi}_{LR}$ such that $\E S(L)_{\psi} \leq c^{|B|/2} \leq c'^{|B|}$ where $c' = \sqrt{c} <  1$.  Finally, local unitary gates are applied to $\ket{\psi}_{LR}$ to obtain the final post-measurement state on the entire chain. Since each unitary is applied to only the left of region $B$ or only the right of region $B$, the entanglement entropy across the $(A, A^c)$ cut is unaffected by these gates, and remains bounded by $c^{|B|}$ in expectation.
\end{proof}

\section{Justification of stat mech mapping}\label{app:statmechdetails}

In this appendix, we provide the justification for the procedure in Section \ref{se:mapping}. Namely, we show that Eq.~\eqref{eq:partitionfunction}, which expresses $\E_U(Z_{k,\emptyset/A})$ as a partition function, is correct, and we derive the equations for the weights of the stat mech model.

To begin, for any integer $k\geq 2$, we rewrite $Z_{k,\emptyset/A}$ from Eqs.~\eqref{eq:Zk0def} and \eqref{eq:ZkAdef} as
\begin{align}
    Z_{k,\emptyset} &= \tr[(\rho \otimes \ldots \otimes \rho)] \label{eq:Zk0} \\
    Z_{k,A} &= \tr[(\rho \otimes \ldots \otimes \rho)W^{(A)}_{(1\ldots k)}] \label{eq:ZkA}
\end{align}
where each trace includes $k$ copies of $\rho$ and $W^{(A)}_{(1\ldots k)}$ is the linear operator that performs a $k$-cycle permutation, denoted $(1\ldots k)$ in cycle notation, of the copies for qudits within region $A$ while leaving the copies of the qudits outside of $A$ unpermuted. When $k=2$ there are two copies of $\rho$ and $W^{(A)}_{(12)}$ is the swap operator for qudits in $A$.

After substituting Eq.~\eqref{eq:rhocircuitoutput} for each copy of $\rho$ that appears in the equations above, we obtain an expression with $k$ copies of each unitary $U_u$ and $k$ copies of its adjoint $U_u^\dagger$, as well as $k$ copies each of $M_u$, $M_u^\dagger$, $M'_u$, and ${M'_u}^\dagger$. Taking the expectation $\E_U(Z_{k,\emptyset/A})$ introduces integrals over $U_u$ and expectations over $M_u$ and $M_u'$ drawn from distributions $\mu_u$ and $\mu'_u$, for each $u$. To perform the integrals, we rely on techniques for integration over the Haar measure, invoking the formula \cite{collins2003moments,collins2006integration}
\begin{equation}\label{eq:haarintegration}
    \int_{U(q^2)} dU \; U_{i_1j_1}\ldots U_{i_kj_k}U_{i_1'j_1'}^\dagger\ldots U_{i_k'j_k'}^\dagger = \sum_{\sigma,\tau \in S_k}\delta_\sigma(\vec{i},\vec{j}')\delta_\tau(\vec{i}',\vec{j}) \Wg(\tau \sigma^{-1},q^2)
\end{equation}
where on the left hand side $U_{ij}$ is the $(i,j)$ matrix element of the unitary $U$ and on the right hand side $S_k$ is the symmetric group, $\delta_\sigma(\vec{i},\vec{j}')$ is shorthand for $\prod_{a=1}^{k} \delta_{i_aj'_{\sigma(a)}}$ and $\Wg(\tau \sigma^{-1}, q^2)$ is the Weingarten function, which can be defined in several ways, for example by the following expansion \cite{collins2003moments,collins2006integration} over irreducible characters of the symmetric group $S_k$
\begin{equation}\label{eq:weingarten}
    \Wg(\pi,q^2) = \frac{1}{(q^2)!^2}\sum_{\lambda }\frac{\chi^\lambda(e)^2}{s_{\lambda,q^2}(1)} \chi^\lambda(\pi)
\end{equation}
where the sum is over all partitions $\lambda$ of the integer $k$, $\chi^\lambda$ is the irreducible character of $S_k$ associated with the partition $\lambda$, $e$ is the identity permutation, and $s_{\lambda,q^2}(1)$ is the Schur polynomial evaluated at 1 which is equal to the dimension of the representation of $U(q^2)$ associated with $\lambda$. Note that there exist permutations $\pi$ for which $\Wg(\pi,q^2)$ is negative.

In words, formula \eqref{eq:haarintegration} states that Haar integration can be performed by summing over all ways of pairing up the incoming index for each of the $k$ copies of $U$ with an outgoing index of a copy of $U^\dagger$, and the incoming index of each copy of $U^\dagger$ with an outgoing index of a copy of $U$. The permutations $\sigma, \tau \in S_k$ encode which copies are paired with each other, and each permutation pair $(\sigma, \tau)$ is weighted by $\Wg(\tau \sigma^{-1},q^2)$ in the sum. It is helpful to think of this formula graphically, as in Figure \ref{fig:haarintegration}, where we have depicted how the indices pair up after integration over a two-qudit Haar-random unitary.

\begin{figure}[h]
    \centering
    \subfloat[Haar integration formula applied to $k$ copies of two-qudit gate]{{\includegraphics[width=0.9\textwidth]{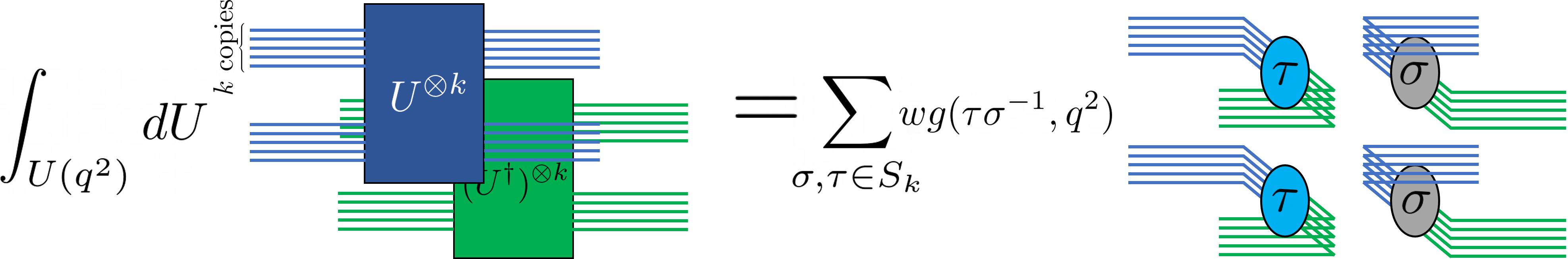} }}%
	\qquad
	\subfloat[Mapping of unitary in circuit diagram]{{\includegraphics[width=0.45\textwidth]{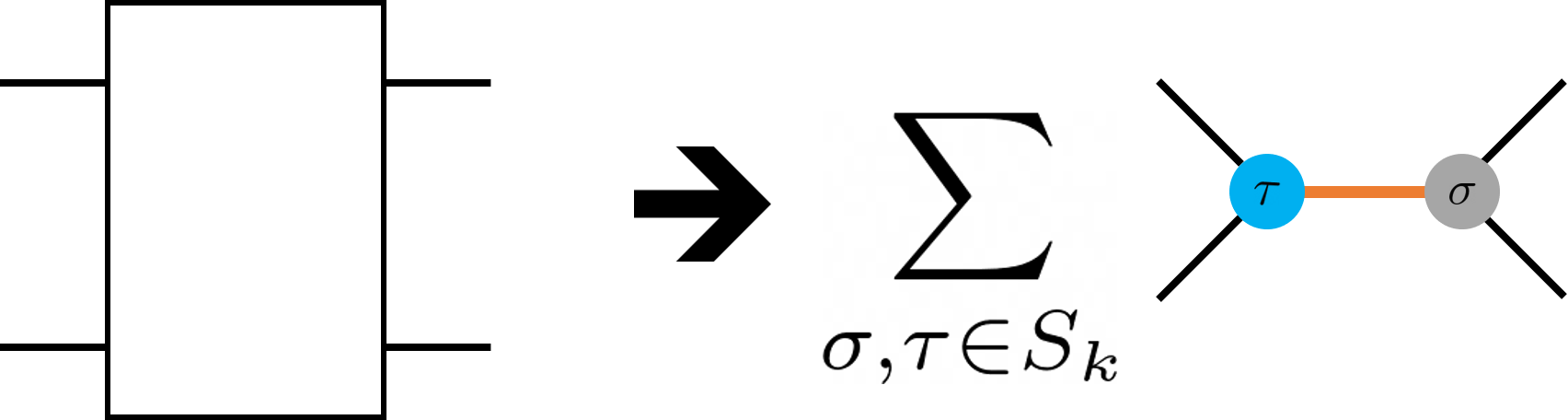} }}%
    \caption{ \label{fig:haarintegration} (a) Graphical depiction of Haar integration formula. (b) Haar integration formula allows us to replace Haar-random unitaries from circuit diagram with sums over configurations on a graph with nodes taking values in $S_k$, and edges between graphs contributing a factor to the weight of a configuration. }
\end{figure}

By applying this formula to all of the Haar-random gates, all of the integrals are eliminated and the tensor network representation of $Z_{k,\emptyset/A}$ can be expressed as a weighted sum over many networks, where each network in this sum corresponds to some choice of $(\sigma_u,\tau_u)$ for every unitary $u$ in the original circuit and some choice of $M_u$ and $M'_u$ from $\mathcal{M}_u$ and $\mathcal{M}'_u$. Furthermore, each network in this weighted sum is itself composed of many disjoint parts that can be individually evaluated. We can see this by observing that when unitary $u_2$ succeeds unitary $u_1$ and shares a qudit, Haar integration forces the $k$ tensor indices representing that qudit at that place in the circuit diagram to pair up with the $k$ dual indices for the qudit at the same place, according to some permutation. This happens both at the output of unitary $u_1$ (corresponding to permutation $\sigma_{u_1}$) and at the input of unitary $u_2$ (corresponding to permutation $\tau_{u_2}$) yielding a set of closed loops in the tensor network diagram. If the weak measurement acting on that qudit between unitaries $u_1$ and $u_2$ is $M$, then $k$ copies of $M$ and $k$ copies of $M^\dagger$ appear among these loops. An example of such a subdiagram is shown in Figure \ref{fig:Mweights}.

This observation justifies Eq.~\eqref{eq:partitionfunction}, as we have expressed $\E_U(Z_{k,\emptyset/A})$ as a weighted sum, with each term labelled by pairs of permutations at the locations of each unitary, where the weight is given by a product of factors that depend only on two of these permutations. These factors are the weights given by Eqs.~\eqref{eq:weightsutu}, \eqref{eq:weightsu1tu2}, and \eqref{eq:weightsuxa}. Eq.~\eqref{eq:weightsutu} accounts for the factor $\Wg(\tau\sigma^{-1},q^2)$ in Eq.~\eqref{eq:haarintegration}. Meanwhile, we can derive Eqs.~\eqref{eq:weightsu1tu2} and \eqref{eq:weightsuxa} by performing the expectation in Figure \ref{fig:Mweights}. We can graphically see that each term in the expansion of this expectation is given by $\mu(M)\tr(W_{\sigma}M^{\otimes k}W_{\tau^{-1}}{(M^\dagger)}^{\otimes k})$, and noting that $W_\pi$ commutes with $X^{\otimes k}$ for any $X$.

\begin{figure}[h]
    \centering
\includegraphics[width=0.4\columnwidth]{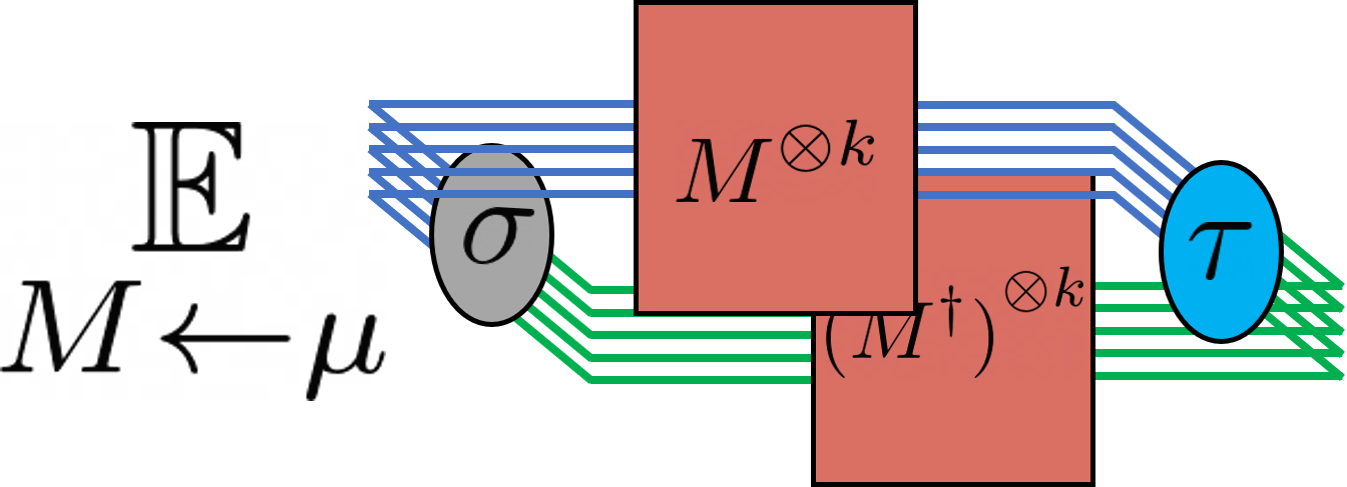}
    \caption{ \label{fig:Mweights} Disjoint part that forms tensor network representation of $\E_U(Z_{k,\emptyset/A})$ after performing integrals over Haar-random gates. The weight given by Eqs.~\eqref{eq:weightsu1tu2} and \eqref{eq:weightsuxa} is derived by evaluating this diagram and taking the expectation with $M$ drawn from $\mathcal{M}$ according to the distribution $\mu$. }
\end{figure}

The final part to justify is the introduction of the auxiliary nodes and corresponding weights in Eq.~\eqref{eq:weightsuxa}. The $\tr$ in the definition of $Z_{k,\emptyset}$ implies that the indices of the qudits at the circuit output are paired up with their dual indices without permutation. This creates a disjoint closed diagram for each qudit at the circuit output. To evaluate it, we may use the same formula as Eq.~\eqref{eq:weightsu1tu2} taking $\tau_{u_2}=e$, the identity permutation. This is equivalent to introducing auxiliary nodes, as we have done, that are fixed to $e$ for all qudits and across all terms in the partition function. The same follows for $Z_{k,A}$ with the exception that the operator $W^{(A)}_{(1\ldots k)}$ is applied to the circuit output, which permutes the output indices of any qudit $a \in A$ prior to connecting them with their dual indices. This is equivalent to introducing an auxiliary node and fixing it to the value $(1 \ldots k)$.

There is no need to introduce auxiliary nodes at the beginning of the circuit because we are assuming the circuit acts on the pure product state $\ket{1\ldots 1}\bra{1 \ldots 1}$. Thus, the $k$ copies of the index that feeds into the first unitary of the circuit are forced to be 0 and regardless of the permutation value of the incoming node for that unitary, this part of the circuit will contribute a factor of 1. If we had considered circuits that act initially on the maximally mixed state, we could have handled this by introducing a layer of auxiliary nodes at the beginning of the circuit and fixing their value to $e$.

\printbibliography[heading=bibintoc]

\end{document}

%% file: Figures/sebd.pdf_tex
\begingroup%
  \makeatletter%
  \providecommand\color[2][]{%
    \errmessage{(Inkscape) Color is used for the text in Inkscape, but the package 'color.sty' is not loaded}%
    \renewcommand\color[2][]{}%
  }%
  \providecommand\transparent[1]{%
    \errmessage{(Inkscape) Transparency is used (non-zero) for the text in Inkscape, but the package 'transparent.sty' is not loaded}%
    \renewcommand\transparent[1]{}%
  }%
  \providecommand\rotatebox[2]{#2}%
  \ifx\svgwidth\undefined%
    \setlength{\unitlength}{552.1228854bp}%
    \ifx\svgscale\undefined%
      \relax%
    \else%
      \setlength{\unitlength}{\unitlength * \real{\svgscale}}%
    \fi%
  \else%
    \setlength{\unitlength}{\svgwidth}%
  \fi%
  \global\let\svgwidth\undefined%
  \global\let\svgscale\undefined%
  \makeatother%
  \begin{picture}(1,0.58303327)%
    \put(0,0){\includegraphics[width=\unitlength,page=1]{sebd.pdf}}%
    \put(-0.00381892,0.57493367){\color[rgb]{0,0,0}\makebox(0,0)[lt]{\begin{minipage}{0.08150359\unitlength}\raggedright \end{minipage}}}%
    \put(0.0340513,0.53761195){\color[rgb]{0,0,0}\makebox(0,0)[lb]{\smash{$(b)$}}}%
    \put(0.36926138,0.53473036){\color[rgb]{0,0,0}\makebox(0,0)[lb]{\smash{$(c)$ }}}%
    \put(0.76209838,0.54145407){\color[rgb]{0,0,0}\makebox(0,0)[lb]{\smash{$(d)$ }}}%
    \put(0.76173778,0.18621423){\color[rgb]{0,0,0}\makebox(0,0)[lb]{\smash{$(e)$ }}}%
    \put(0.37118243,0.18621423){\color[rgb]{0,0,0}\makebox(0,0)[lb]{\smash{$(f)$ }}}%
    \put(0.03597235,0.18621423){\color[rgb]{0,0,0}\makebox(0,0)[lb]{\smash{$(a)$ }}}%
    \put(0.06410074,0.20816751){\color[rgb]{0,0,0}\makebox(0,0)[lb]{\smash{}}}%
    \put(0.34649444,0.00134951){\color[rgb]{0,0,0}\makebox(0,0)[lb]{\smash{$0$ }}}%
    \put(0.38555942,0.00134951){\color[rgb]{0,0,0}\makebox(0,0)[lb]{\smash{$1$ }}}%
    \put(0.42462445,0.00134951){\color[rgb]{0,0,0}\makebox(0,0)[lb]{\smash{$1$ }}}%
    \put(0.50275449,0.00134951){\color[rgb]{0,0,0}\makebox(0,0)[lb]{\smash{$1$ }}}%
    \put(0.46368951,0.00134951){\color[rgb]{0,0,0}\makebox(0,0)[lb]{\smash{$0$ }}}%
  \end{picture}%
\endgroup%

%% file: Figures/patching.pdf_tex
\begingroup%
  \makeatletter%
  \providecommand\color[2][]{%
    \errmessage{(Inkscape) Color is used for the text in Inkscape, but the package 'color.sty' is not loaded}%
    \renewcommand\color[2][]{}%
  }%
  \providecommand\transparent[1]{%
    \errmessage{(Inkscape) Transparency is used (non-zero) for the text in Inkscape, but the package 'transparent.sty' is not loaded}%
    \renewcommand\transparent[1]{}%
  }%
  \providecommand\rotatebox[2]{#2}%
  \ifx\svgwidth\undefined%
    \setlength{\unitlength}{390.04027619bp}%
    \ifx\svgscale\undefined%
      \relax%
    \else%
      \setlength{\unitlength}{\unitlength * \real{\svgscale}}%
    \fi%
  \else%
    \setlength{\unitlength}{\svgwidth}%
  \fi%
  \global\let\svgwidth\undefined%
  \global\let\svgscale\undefined%
  \makeatother%
  \begin{picture}(1,0.72488492)%
    \put(0,0){\includegraphics[width=\unitlength,page=1]{patching.pdf}}%
    \put(0.6991188,-5.46120936){\color[rgb]{0,0,0}\makebox(0,0)[lt]{\begin{minipage}{1.03366239\unitlength}\raggedright \end{minipage}}}%
    \put(0.88302037,-6.76572611){\color[rgb]{0,0,0}\makebox(0,0)[lb]{\smash{}}}%
    \put(0.51350503,-5.98210964){\color[rgb]{0,0,0}\makebox(0,0)[lb]{\smash{}}}%
    \put(0,0){\includegraphics[width=\unitlength,page=2]{patching.pdf}}%
    \put(0.25649825,0.28794054){\color[rgb]{0,0,0}\makebox(0,0)[lb]{\smash{\textit{$l$}}}}%
    \put(0.30649311,0.28794054){\color[rgb]{0,0,0}\makebox(0,0)[lb]{\smash{\textit{$l$}}}}%
    \put(0,0){\includegraphics[width=\unitlength,page=3]{patching.pdf}}%
    \put(0.30156453,0.22533557){\color[rgb]{0,0,0}\makebox(0,0)[lt]{\begin{minipage}{0.45287866\unitlength}\raggedright A\end{minipage}}}%
    \put(0.25237187,0.22533557){\color[rgb]{0,0,0}\makebox(0,0)[lt]{\begin{minipage}{0.45287866\unitlength}\raggedright B\end{minipage}}}%
    \put(0.34950572,0.22533557){\color[rgb]{0,0,0}\makebox(0,0)[lt]{\begin{minipage}{0.45287866\unitlength}\raggedright B\end{minipage}}}%
    \put(0,0){\includegraphics[width=\unitlength,page=4]{patching.pdf}}%
    \put(0.8204206,0.22705995){\color[rgb]{0,0,0}\makebox(0,0)[lt]{\begin{minipage}{0.42646367\unitlength}\raggedright A\end{minipage}}}%
    \put(0.723918,0.22705995){\color[rgb]{0,0,0}\makebox(0,0)[lt]{\begin{minipage}{0.42646367\unitlength}\raggedright B\end{minipage}}}%
    \put(0,0){\includegraphics[width=\unitlength,page=5]{patching.pdf}}%
    \put(0.19180121,0.19780712){\color[rgb]{0,0,0}\makebox(0,0)[lb]{\smash{\textit{$3l$}}}}%
    \put(0.16314617,0.18928505){\color[rgb]{0,0,0}\makebox(0,0)[lb]{\smash{}}}%
    \put(0,0){\includegraphics[width=\unitlength,page=6]{patching.pdf}}%
    \put(0.81636317,0.28202677){\color[rgb]{0,0,0}\makebox(0,0)[lb]{\smash{\textit{$3l$}}}}%
    \put(0,0){\includegraphics[width=\unitlength,page=7]{patching.pdf}}%
    \put(0.92295082,0.28266657){\color[rgb]{0,0,0}\makebox(0,0)[lb]{\smash{\textit{$l$}}}}%
    \put(-0.00250717,0.70150993){\color[rgb]{0,0,0}\makebox(0,0)[lb]{\smash{$(a)$}}}%
    \put(0.01447732,0.00452177){\color[rgb]{0,0,0}\makebox(0,0)[lb]{\smash{$(b)$}}}%
    \put(0.35225099,0.70150993){\color[rgb]{0,0,0}\makebox(0,0)[lb]{\smash{$(c)$}}}%
    \put(0.71137078,0.70150993){\color[rgb]{0,0,0}\makebox(0,0)[lb]{\smash{$(e)$}}}%
    \put(0.69205,0.00452177){\color[rgb]{0,0,0}\makebox(0,0)[lb]{\smash{$(d)$}}}%
  \end{picture}%
\endgroup%

%% file: Figures/statmech.pdf_tex
\begingroup%
  \makeatletter%
  \providecommand\color[2][]{%
    \errmessage{(Inkscape) Color is used for the text in Inkscape, but the package 'color.sty' is not loaded}%
    \renewcommand\color[2][]{}%
  }%
  \providecommand\transparent[1]{%
    \errmessage{(Inkscape) Transparency is used (non-zero) for the text in Inkscape, but the package 'transparent.sty' is not loaded}%
    \renewcommand\transparent[1]{}%
  }%
  \providecommand\rotatebox[2]{#2}%
  \ifx\svgwidth\undefined%
    \setlength{\unitlength}{518.91503254bp}%
    \ifx\svgscale\undefined%
      \relax%
    \else%
      \setlength{\unitlength}{\unitlength * \real{\svgscale}}%
    \fi%
  \else%
    \setlength{\unitlength}{\svgwidth}%
  \fi%
  \global\let\svgwidth\undefined%
  \global\let\svgscale\undefined%
  \makeatother%
  \begin{picture}(1,0.42751683)%
    \put(0,0){\includegraphics[width=\unitlength,page=1]{statmech.pdf}}%
    \put(-0.00479812,0.43474282){\color[rgb]{0,0,0}\makebox(0,0)[lt]{\begin{minipage}{0.30138249\unitlength}\raggedright $(a)$\end{minipage}}}%
    \put(0.37875273,0.43474483){\color[rgb]{0,0,0}\makebox(0,0)[lt]{\begin{minipage}{0.26266347\unitlength}\raggedright $(b)$\end{minipage}}}%
    \put(0.79219449,0.43474483){\color[rgb]{0,0,0}\makebox(0,0)[lt]{\begin{minipage}{0.18273001\unitlength}\raggedright $(c)$\end{minipage}}}%
  \end{picture}%
\endgroup%

%% file: Figures/brickwork.pdf_tex
\begingroup%
  \makeatletter%
  \providecommand\color[2][]{%
    \errmessage{(Inkscape) Color is used for the text in Inkscape, but the package 'color.sty' is not loaded}%
    \renewcommand\color[2][]{}%
  }%
  \providecommand\transparent[1]{%
    \errmessage{(Inkscape) Transparency is used (non-zero) for the text in Inkscape, but the package 'transparent.sty' is not loaded}%
    \renewcommand\transparent[1]{}%
  }%
  \providecommand\rotatebox[2]{#2}%
  \ifx\svgwidth\undefined%
    \setlength{\unitlength}{492.66833766bp}%
    \ifx\svgscale\undefined%
      \relax%
    \else%
      \setlength{\unitlength}{\unitlength * \real{\svgscale}}%
    \fi%
  \else%
    \setlength{\unitlength}{\svgwidth}%
  \fi%
  \global\let\svgwidth\undefined%
  \global\let\svgscale\undefined%
  \makeatother%
  \begin{picture}(1,0.3840952)%
    \put(0,0){\includegraphics[width=\unitlength,page=1]{brickwork.pdf}}%
    \put(-0.00505331,0.39170659){\color[rgb]{0,0,0}\makebox(0,0)[lt]{\begin{minipage}{0.22564631\unitlength}\raggedright $(a)$\end{minipage}}}%
    \put(0.51362366,0.39170786){\color[rgb]{0,0,0}\makebox(0,0)[lt]{\begin{minipage}{0.19494371\unitlength}\raggedright $(b)$\end{minipage}}}%
  \end{picture}%
\endgroup%

%% file: Figures/decimation.pdf_tex
\begingroup%
  \makeatletter%
  \providecommand\color[2][]{%
    \errmessage{(Inkscape) Color is used for the text in Inkscape, but the package 'color.sty' is not loaded}%
    \renewcommand\color[2][]{}%
  }%
  \providecommand\transparent[1]{%
    \errmessage{(Inkscape) Transparency is used (non-zero) for the text in Inkscape, but the package 'transparent.sty' is not loaded}%
    \renewcommand\transparent[1]{}%
  }%
  \providecommand\rotatebox[2]{#2}%
  \ifx\svgwidth\undefined%
    \setlength{\unitlength}{480.79672377bp}%
    \ifx\svgscale\undefined%
      \relax%
    \else%
      \setlength{\unitlength}{\unitlength * \real{\svgscale}}%
    \fi%
  \else%
    \setlength{\unitlength}{\svgwidth}%
  \fi%
  \global\let\svgwidth\undefined%
  \global\let\svgscale\undefined%
  \makeatother%
  \begin{picture}(1,0.38560068)%
    \put(0,0){\includegraphics[width=\unitlength,page=1]{decimation.pdf}}%
    \put(-0.00517939,0.39340174){\color[rgb]{0,0,0}\makebox(0,0)[lt]{\begin{minipage}{0.28524083\unitlength}\raggedright $(a)$\end{minipage}}}%
    \put(0.56753071,0.39340087){\color[rgb]{0,0,0}\makebox(0,0)[lt]{\begin{minipage}{0.36880746\unitlength}\raggedright $(b)$\end{minipage}}}%
  \end{picture}%
\endgroup%